\documentclass[aps,prb,notitlepage,nofootinbib]{revtex4-2}
\usepackage{graphicx}
\usepackage{amsfonts}
\usepackage{amssymb}
\usepackage{amsmath}
\usepackage{mathtools}
\usepackage{afterpage}
\usepackage[mathscr]{euscript}
\usepackage{braket}
\usepackage{subfigure}

\let\savecorresponds\corresponds
\let\corresponds\relax
\usepackage{mathabx}
\let\corresponds\savecorresponds

\usepackage[pdfencoding=auto, psdextra]{hyperref}
\hypersetup{colorlinks=true,citecolor=blue,linkcolor=blue, urlcolor=blue}

\usepackage{amsthm}
\newtheorem{theorem}{Theorem}[section]

\newtheorem{corollary}[theorem]{Corollary} 
\newtheorem{proposition}[theorem]{Proposition} 

\usepackage{multirow}
\usepackage{float}
\usepackage{qcircuit}
\usepackage{leftidx}

\usepackage[utf8]{inputenc} 
\usepackage{yfonts}
\usepackage{dsfont}
\usepackage{cancel}
\usepackage[scr=boondox]{mathalpha}
\usepackage{pst-all}

\newcommand{\coho}[1]{\textswab{#1}}
\newcommand{\cohosub}[1]{\protect\scalebox{0.7}{\textswab{#1}}}

\begin{document}
	\def\U{\mathrm{U}(1)}
	\def\H{\mathcal{H}}
	\def\F{\mathcal{F}}
	\def\E{\mathbb{E}}
	\def\TT{\mathsf{T}}
	\def\A{\mathcal{A}}
	\def\L{\mathcal{L}}
	\def\w{\mathfrak{w}}
	\def\O{\coho{O}}
	\def\C{\widecheck{\mathcal{C}}}
	\def\SO{\mathrm{SO}}

\newcommand{\Aut}{\text{Aut}}
\newcommand{\Z}{\mathbb Z}
\newcommand{\R}{\mathbb R}
    \renewcommand{\L}{\mathcal L}
	\newcommand{\commentmb}[1]{\textcolor{purple}{[MB: #1]}}
	\newcommand{\commentdb}[1]{\textcolor{blue}{[DB: #1]}}
	
\newcommand{\personalcomment}[1]{\textcolor{red}{[#1]}}

\newcommand{\mathbbm}[1]{\text{\usefont{U}{bbm}{m}{n}#1}}

    \title{Fermionic symmetry fractionalization in (2+1)D}
        
        \author{Daniel Bulmash}
        \author{Maissam Barkeshli}
         \affiliation{Condensed Matter Theory Center and Joint Quantum Institute, Department of Physics, University of Maryland, College Park, Maryland 20472 USA}

        \begin{abstract}
            We develop a systematic theory of symmetry fractionalization for fermionic topological phases of matter in (2+1)D with a general fermionic symmetry group $G_f$. In general $G_f$ is a central extension of the bosonic symmetry group $G_b$ by fermion parity, $(-1)^F$, characterized by a non-trivial cohomology class $[\omega_2] \in \mathcal{H}^2(G_b, \Z_2)$. We show how the presence of local fermions places a number of constraints on the algebraic data that defines the action of the symmetry on the super-modular tensor category that characterizes the anyon content. We find two separate obstructions to defining symmetry fractionalization, which we refer to as the bosonic and fermionic symmetry localization obstructions. The former is valued in $\mathcal{H}^3(G_b, K(\mathcal{C}))$, while the latter is valued in either $\mathcal{H}^3(G_b,\mathcal{A}/\{1,\psi\})$ or $Z^2(G_b, \Z_2)$ depending on additional details of the theory. $K(\mathcal{C})$ is the Abelian group of functions from anyons to $\U$ phases obeying the fusion rules, $\mathcal{A}$ is the Abelian group defined by fusion of Abelian anyons, and $\psi$ is the fermion. When these obstructions vanish, we show that distinct symmetry fractionalization patterns form a torsor over $\mathcal{H}^2(G_b, \mathcal{A}/\{1,\psi\})$. 
            We study a number of examples in detail; in particular we provide a characterization of fermionic Kramers degeneracy arising in symmetry class DIII within this general framework, and we discuss fractional quantum Hall and $\Z_2$ quantum spin liquid states of electrons. 
        \end{abstract}

        \maketitle
        \tableofcontents

        \section{Introduction}
        
        A fundamental property of topological phases of matter is the possibility of quasiparticles carrying fractional quantum numbers. Well-known examples of this include the fractional electric charge carried by anyons in fractional quantum Hall (FQH) states or spinons with spin 1/2 in quantum spin liquids with global $\SO(3)$ spin rotational symmetry \cite{wen04}. In the past several years, it has been understood how to mathematically characterize symmetry fractionalization in (2+1)D bosonic topological phases of matter in complete generality~\cite{barkeshli2019}. This includes topological states with arbitrary global symmetry groups including both unitary and anti-unitary symmetry actions, Abelian or non-Abelian topological phases of matter, and cases where global symmetries can permute the anyons. These results have led to significant progress in developing a comprehensive characterization and classification of (2+1)D bosonic symmetry-enriched topological phases of matter (SETs). In particular the systematic understanding of symmetry fractionalization has allowed predictions of novel fractional quantum numbers in the presence of space group symmetries for quantum spin liquids and FQH states~\cite{essin2013,essin2014,qi2018spinliquid,manjunath2021,manjunath2020}, methods to strengthen the Lieb-Schulz-Mattis theorem for $(2+1)$D topological phases~\cite{cheng2016lsm}, and general methods to compute anomalies in (2+1)D topological phases of matter~\cite{barkeshli2019rel,barkeshli2019tr,bulmash2020,TataSETAnomaly}.
        
        In this paper we generalize the systematic mathematical framework of symmetry fractionalization to the case of general fermionic topological phases of matter in (2+1)D. Fermionic systems possess a special ``symmetry" operation\footnote{Symmetry is in quotes because fermion parity, in contrast with all other symmetries, can never be spontaneously broken.}, fermion parity, which is denoted $(-1)^F$, and we denote the corresponding order-2 group as $\Z_2^f$. The complete symmetry group that acts on both fermionic and bosonic operators is denoted $G_f$, while the quotient group that only acts on bosonic operators is $G_b = G_f / \Z_2^f$. The possibility of fermions that can be created by local fermionic operators imposes a number of non-trivial constraints in the mathematical description of the SET, which we systematically study. 
        
        \subsection{Summary of results}
        
        It is generally believed that in the absence of any symmetry, a general (2+1)D topological phase of matter can be fully characterized by two mathematical objects, 
        $(\mathcal{C}, c_-)$, where $\mathcal{C}$ is a unitary modular tensor category (UMTC) in the case of bosonic systems and a unitary super-modular tensor category for fermionic systems \cite{bruillard2017a,bruillard2017b,bonderson2018}. $c_-$ is the chiral central charge of the (1+1)D edge theory. $\mathcal{C}$ determines $c_-$ modulo $8$ for bosonic systems and modulo $1/2$ for fermionic systems. In both the bosonic and fermionic cases, $\mathcal{C}$ is a unitary braided fusion category (UBFC), which keeps track of the braiding and fusion properties of the anyons. 
        
        Since UMTCs and super-modular tensor categories are both special cases of UBFCs, one might naively attempt to immediately extend the known symmetry fractionalization framework for bosonic topological phases to the fermionic case. However, fermionic topological phases are different from bosonic topological phases in two crucial ways, both of which must be accounted for in the theory of symmetry fractionalization. 
        
        One difference arises from the fact that the fermionic topological phase contains a local fermion; although the fermion directly appears in the UBFC data, its locality does not. Carefully tracking the locality of the fermion leads to significant consequences due to reduced gauge freedom in data involved in the symmetry action. In particular, autoequivalences of the super-modular tensor category are only equivalent up to what we term ``locality-respecting natural isomorphisms." 
        The first step in defining how a symmetry acts on a bosonic topological phase is to define a group homomorphism 
        \begin{align}
        [\rho]:G \rightarrow \Aut(\mathcal{C}), 
        \end{align}
        where $\Aut(\mathcal{C})$ is, roughly stated, a group formed by the set of (braided tensor) autoequivalences of the UMTC $\mathcal{C}$ modulo a set of ``trivial" transformations called natural isomorphisms \cite{barkeshli2019} (see Section \ref{sec:bosonicSymmFrac} for a brief review). 
        
        In fermionic topological phases, the set of natural isomorphisms can be distinguished by whether they are ``locality-respecting" or ``locality-violating" with respect to the local fermion. 
        We may then define two groups, $\Aut(\mathcal{C})$ and $\Aut_{LR}(\mathcal{C})$, depending on whether we mod out by all natural isomorphisms or only the locality-respecting ones. Depending on $\mathcal{C}$, these groups may or may not be isomorphic, but in either case, we find that the first step towards defining how symmetry acts on a fermionic topological phase is to specify a map
        \begin{align}
        \label{lrHom}
        [\rho]:G_b \rightarrow \Aut_{LR}(\mathcal{C}).
        \end{align}
        Here any representative $\rho_{\bf g}$ of the equivalence class $[\rho_{\bf g}]$ satisfies $\rho_{\bf gh} = \kappa_{\bf g, \bf h} \circ \rho_{\bf g} \circ \rho_{\bf h}$, where $\kappa_{\bf g, \bf h}$ is a natural isomorphism, as reviewed in Section \ref{sec:bosonicSymmFrac}. 
        
        One consequence of Eq. \ref{lrHom} is that if $\Aut_{LR}(\mathcal{C})$ and $\Aut(\mathcal{C})$ are not isomorphic, then simply keeping track of how anyons are permuted under the symmetries is not enough to fully determine an element in $\Aut_{LR}(\mathcal{C})$. 
        (By contrast, the permutation action often, but not always\footnote{For examples of UMTCs with non-permuting but non-trivial autoequivalences, see~\cite{davydov2014}, section 3. One such example is the Drinfeld center of the group $G$, where $G$ is an order-64 group with the presentation
        \begin{equation}
            \left\langle a,b,c \phantom{i}\bigg| \phantom{i} a^2=b^2=1,\phantom{a} c^2=[a,c], \phantom{a} [c,b]=[[c,a],a], \phantom{a} [[b,a],G]=1, \phantom{a} [G,[G,[G,G]]]=1 \right\rangle. \nonumber
        \end{equation}
        }, \textit{does} uniquely determine an element of $\Aut(\mathcal{C})$.) In a large class of examples where elements of $\Aut_{LR}(\mathcal{C})$ are not uniquely determined by the way they permute anyons, we provide in Eq.~\ref{eqn:Lambda} a gauge-invariant quantity which distinguishes classes with the same permutation action.
        
        A second important difference in the fermionic case arises from the presence of the fermionic symmetry group $G_f$, which includes an additional piece of data $[\omega_2] \in \mathcal{H}^2(G_b, \Z_2)$ specifying $G_f$ in terms of a central extension of $G_b$ by $\Z_2^f$. Our formalism accounts for $\omega_2$ by viewing it as endowing the fermion with fractional quantum numbers under $G_b$. This leads us to a \textit{constrained} theory of $G_b$ symmetry fractionalization,  where the constraint directly encodes the way that $G_b$ is embedded in the full symmetry group $G_f$.
     
        To briefly summarize the constraints, we note that $[\rho_{\bf g}]$ and symmetry fractionalization pattern in general corresponds to a set of data $\{\rho_{\bf g}, U_{\bf g}(a,b;c), \eta_a({\bf g}, {\bf h})\}$. Here $U_{\bf g}(a,b;c)$ is a set of $N_{ab}^c \times N_{ab}^c$ matrices, where $N_{ab}^c$ are the fusion coefficients, that specify the action of a representative autoequivalence $\rho_{\bf g}$ on the fusion and splitting spaces of $\mathcal{C}$. $\eta_a({\bf g}, {\bf h})$ is a $\U$ phase for each anyon $a$. These data are subject to a set of consistency conditions and gauge transformations. In particular, we have a set of so-called symmetry-action gauge transformations, which correspond to changing the representative map $\rho_{\bf g}$, and which transform the data as
        \begin{align}
            \eta_a({\bf g}, {\bf h}) &\rightarrow \frac{\gamma_a({\bf g h})}{[\gamma_{\,^{\bf \overline{g}}a}({\bf h})]^{\sigma({\bf g})} \gamma_a({\bf g})} \eta_a({\bf g,h})
            \nonumber \\ 
            U_{\bf g}(a,b;c) &\rightarrow \frac{\gamma_a({\bf g}) \gamma_b({\bf g})}{\gamma_c({\bf g})} U_{\bf g}(a,b;c) ,
            \label{eqn:UetaGaugeTransform}
        \end{align}
        where $\gamma_a({\bf g})$ is a $U(1)$ phase. 
        Here $\sigma({\bf g}) = 1$ or $*$ is a $\Z_2$ grading on $G_b$ which determines whether ${\bf g}$ is a unitary or anti-unitary symmetry. 
        The constraints alluded to above then take the form \cite{bondersonNote}:
        \begin{align}
        \label{fermConstraints}
         \eta_\psi({\bf g}, {\bf h}) &= \omega_2({\bf g}, {\bf h})
         \nonumber \\
            U_{\bf g}(\psi, \psi; 1) &= 1
            \nonumber \\
            \gamma_\psi({\bf g}) &= 1
        \end{align}
        Here $\psi$ is the local fermion, which is treated as a non-trivial object in the super-modular tensor category. 
        We note that a microscopic specification of a quantum many-body system and the representation of the symmetries also specifies a representative $2$-cocycle $\omega_2$ which enters the above constraints. 
        
        After accounting for all of these constraints, we take $[\rho_{\bf g}]$ as defined by Eq.~\ref{lrHom} as given, and then we determine the obstructions to defining a consistent theory of symmetry fractionalization. We find two distinct such obstructions.
        
        We find that there is a ``bosonic" obstruction, 
        \begin{equation}
            [\Omega] \in \H^3(G_b,K(\mathcal{C})),
        \end{equation}
        where $K(\mathcal{C})$ is the Abelian group of functions $\Omega_a$ from anyon labels to $\U$ which obey the fusion rules in the sense $\Omega_a \Omega_b = \Omega_c$ when $N_{ab}^c>0$. If there does not exist an element of $K(\mathcal{C})$ with $\Omega_\psi = -1$, then $K(\mathcal{C}) \simeq \mathcal{A}/\{1,\psi\}$, otherwise $K(\mathcal{C})$ is an extension of $\Z_2$ by  $\mathcal{A}/\{1,\psi\}$.
        Here $\A$ is the Abelian group formed by fusion of the Abelian anyons in $\mathcal{C}$. The equivalence by $\{1,\psi\}$ means different elements in $\mathcal{A}$ that differ by fusion with $\psi$ are regarded as equivalent. 
        $[\Omega]$ is a symmetry localization obstruction as discussed in \cite{barkeshli2019,barkeshli2018,fidkowski2015}, and can be viewed as an obstruction to finding any consistent pattern of symmetry fractionalization for the $G_b$ symmetry, ignoring $G_f$. 
        
        In attempting to extend the symmetry fractionalization to the full $G_f$ symmetry, we find that once the bosonic obstruction vanishes, there is the possibility of a second ``fermionic" obstruction $[\coho{O}_f]$,
        \begin{align}
            [\coho{O}_f] \in 
            \begin{cases}
                \mathcal{H}^3(G_b, \A/\{1,\psi\}) \\
                Z^2(G_b, \Z_2)
            \end{cases}
        \end{align}
        Here $Z^2$ is the group of $2$-cocycles on $G_b$. 
        Whether $[\coho{O}_f]$ is valued in $\mathcal{H}^3(G_b, \A/\{1,\psi\})$ or $Z^2(G_b, \Z_2)$ depends on whether there exists a set of phases $\zeta_a$ which satisfy the fusion rules with $\zeta_\psi = -1$. Moreover, depending on whether we fix $\omega_2$ or just its cohomology class $[\omega_2]$, $Z^2(G_b, \Z_2)$ may be replaced with $\mathcal{H}^2(G_b, \Z_2)$. 
        
        The fermionic obstruction $[\coho{O}_f]$ is essentially an obstruction to finding a symmetry fractionalization class that is consistent with the choice of $\omega_2$ that defines $G_f$. Alternatively, it can  be viewed as an obstruction to ``lifting" the fractional quantum numbers of $\psi$ under $G_b$ to the full super-modular tensor category $\mathcal{C}$. 
        
        In the case where $\Aut(\mathcal{C})$ and $\Aut_{LR}(\mathcal{C})$ are not isomorphic, then given a $G_b$ and $[\rho_{\bf g}]$, we find that there is at most one group extension $G_f$ for which the fermionic obstruction vanishes. In particular, this implies that for any continuous symmetry group, such as $U(1)$ or $SO(3)$, a non-trivial $[\omega_2]$ is incompatible with any super-modular category for which $\Aut(\mathcal{C})$ and $\Aut_{LR}(\mathcal{C})$ are not isomorphic. 
        
       If the above obstructions vanish, then it is possible to define  symmetry fractionalization in a way consistent with $G_f$ and locality of the fermion. We then show that the set of possible symmetry fractionalization classes for the anyons form a torsor over $\mathcal{H}^2(G_b, \A/\{1,\psi\})$, which thus provides a classification of symmetry fractionalization for anyons in fermionic topological phases \cite{bondersonNote}. Note that this analysis does not include symmetry fractionalization for the fermion parity vortices. 
        
        Once our formalism is established, we consider several example symmetry groups, providing and physically interpreting gauge-invariant quantities that characterize patterns of symmetry fractionalization in fermionic topological phases. 
        As an example, our results allow us to precisely understand the notion of ``fermionic Kramers degeneracy" presented in \cite{Fidkowski13,metlitski2014} within this systematic formalism (this result was announced previously in \cite{TataSETAnomaly}). When $G_b = \Z_2^{\bf T}$ and $G_f = \Z_4^{{\bf T}, f}$ (i.e. ${\bf T}^2 = (-1)^F$), we have a quantity
        \begin{align}
            \eta_a^{\bf T} := \eta_a({\bf T}, {\bf T}) U_{\bf T}(a, \psi; a \times \psi) F^{a,\psi,\psi}= \pm i
        \end{align}
        which is gauge invariant when $\,^{\bf T}a = a \times \psi$. For such anyons, $\eta_a^{\bf T}$ can be viewed as the ``local ${\bf T}^2$ eigenvalue" of that anyon. We also generalize this notion of fermionic Kramers degeneracy to $G_f = \Z_8^{{\bf T},f}$ and investigate fermionic fractional quantum Hall and quantum spin liquid states of electrons in our framework.
        
        In the special case where $G_f = G_b \times \Z_2^f$,  Ref.~\cite{fidkowski2018} proposed the possibility of an  $\mathcal{H}^2(G_b, \mathcal{A}/\{1,\psi\})$ structure in the classification of symmetry fractionalization, however a complete derivation and a specification of how to treat the braided auto-equivalences were not provided. For general $G_f$, some of our results correspond to results on categorical fermionic actions in the mathematical literature \cite{galindo}. However, the results of  Ref.~\cite{galindo} do not account for the locality of the fermion, which ultimately leads to a different classification. Where there is overlap, our work provides a physical understanding of the mathematical results in \cite{galindo} and a formulation in terms of the ``skeletonization" of the super-modular tensor category. 
    
        The rest of this paper is organized as follows. In Sec.~\ref{sec:supermodular}, we discuss some basic definitions and facts regarding super-modular tensor categories and their use in modeling fermionic topological phases. In Sec.~\ref{sec:bosonicSymmFrac}, we review the symmetry fractionalization formalism for bosonic topological phases. In Sec.~\ref{fermSymLocSec}, we develop a theory of fermionic symmetry fractionalization by constraining a theory of bosonic symmetry fractionalization and develop the concept of a ``locality-respecting natural isomorphism." In Sec.~\ref{sec:obstructions}, we compute the obstructions to fermionic symmetry localization and, if the obstructions vanish, classify fermionic symmetry fractionalization patterns. Sec.~\ref{sec:examples} consists of a number of examples of the use of our formalism for different symmetry groups, and we conclude with some general discussion in Sec.~\ref{sec:discussion}.

    \section{Super-modular and spin modular categories and fermionic topological phases of matter}
    \label{sec:supermodular}
        
        In this paper we will assume familiarity with unitary braided fusion categories (UBFCs). These are specified by a list of anyon labels $\lbrace a, b, c, \ldots\rbrace$, fusion spaces $V_{ab}^c$ and their dual splitting spaces $V^{ab}_c$, 
        non-negative integer fusion coefficients $N_{ab}^c = \text{ dim } V_{ab}^c$, $F$-symbols $F^{abc}_{def}$ which specify the associativity of fusion, and $R$-symbols $R^{ab}_c$ which specify braiding data, all subject to the well-known pentagon and hexagon consistency conditions.
        See, e.g., Refs.~\cite{Bonderson07b,barkeshli2019} for a review of UBFCs and conventions; our conventions are essentially identical to those of Section II of Ref. \cite{barkeshli2019}. We will for simplicity often restrict our attention to the case where all fusion coefficients $N_{ab}^c \leq 1$; the generalization is straightforward.
        
        We will use the scalar monodromy
        \begin{equation}
            M_{ab} = \frac{S_{ab}^{\ast}S_{11}}{S_{1a}S_{1b}},
        \end{equation}
        where $S_{ab}$ is the topological $S$-matrix, extensively in this paper. $M_{ab}$ is always a phase if $a$ or $b$ is Abelian, in which case it has a physical interpretation as a phase arising from a double braid, but $M_{ab}$ may or may not be a phase if both $a$ and $b$ are non-Abelian.
        
        A super-modular tensor category~\cite{bruillard2017a,bruillard2017b,bonderson2018} $\mathcal{C}$ can be defined as a UBFC with a single nontrivial invisible particle, $\psi$, such that $\psi$ is a fermion, i.e., its topological twist $\theta_\psi=-1$, and satisfies $\Z_2$ fusion rules, $\psi \times \psi = 1$. ``Invisible" means that $\psi$ braids trivially with all particles in $\mathcal{C}$, that is, its double braid $M_{a,\psi}=+1$ for all $a \in \mathcal{C}$.
        
        The existence of a single invisible fermion $\psi$ with $\Z_2$ fusion rules implies that the set of anyon labels of a super-modular tensor category decomposes as $\mathcal{C} = \mathcal{B} \times \{1,\psi\}$, but the fusion rules need not respect this decomposition. On the other hand, the topological $S$-matrix of $\mathcal{C}$ does respect this decomposition:
        \begin{align}
            S = \tilde{S} \otimes \frac{1}{\sqrt{2}} \left( \begin{matrix} 1 & 1 \\ 1 & 1 \end{matrix} \right),
        \end{align}
        where $\tilde{S}$ is unitary. 
        
        Physically, the super-modular tensor category keeps track of the topologically non-trivial quasi-particle content in a fermionic topological phase of matter. The theory explicitly keeps track of the fermion as well, which is topologically trivial in the sense that it can be created or annihilated by a local fermion operator. 
        
        In the rest of this paper, the symbol $\mathcal{C}$ will always refer to a super-modular tensor category. We denote the unitary braided fusion subcategory of Abelian anyons as $\A \subset \mathcal{C}$. In general in an Abelian super-modular tensor category, the fermion decouples~\cite{MengFermionicLSM}:
        \begin{align}
            \A \simeq \tilde{\A} \boxtimes \{1,\psi\}.
            \label{eqn:abelianDecomp}
        \end{align}
        Here the symbol $\boxtimes$ is the Deligne product and physically means stacking two decoupled topological orders, and $\{1,\psi\}$ denotes the UBFC with just two particles, $\{1,\psi\}$. 
        
        We will also use the symbol $\A$ to refer to the Abelian group defined by fusion of Abelian anyons. The above implies that as an Abelian group $\A \simeq \tilde{\A} \times \Z_2$, where $\Z_2 = \{1,\psi\}$ is associated with the invisible fermion. 

        It is a mathematical theorem~\cite{johnsonFreydModularExt} that every super-modular tensor category admits a minimal modular extension $\C$, i.e. a UMTC $\C$ that contains $\mathcal{C}$ as a subcategory and which has the minimal possible total quantum dimension of
        \begin{align}
            \mathcal{D}^2_{\C} = 2 \mathcal{D}^2_{\mathcal{C}}. 
        \end{align}
        The ``16-fold way" theorem states that there are precisely 16 distinct minimal modular extensions, with chiral central charges differing by $\nu/2$ for $\nu = 0,\cdots, 15$ \cite{bruillard2017a}. 
        
        A minimal modular extension $\C$ is an example of a spin modular category. A spin modular category is defined to be a UMTC, $\C$, together with a preferred choice of fermion $\psi$, which has topological twist $\theta_\psi = -1$ and such that $\psi \times \psi = 1$. 
        
        The spin modular category $\C$ which gives a minimal modular extension of $\mathcal{C}$ possesses a natural $\Z_2$ grading determined by braiding with the fermion $\psi$:
        \begin{align}
            \C &= \C_0 \oplus \C_1 ,
            \nonumber \\
            \C_0 &\simeq \mathcal{C} 
        \end{align}
        That is, $M_{a,\psi} = +1$ if $a \in \C_0$ and $M_{a,\psi} = -1$ if $a \in \C_1$, and fusion respects this grading. The anyons in $\C_0 \simeq \mathcal{C}$ correspond to the original particles in the super-modular theory. The anyons in $\C_1$ are physically interpreted as fermion parity vortices, which can be understood as symmetry defects associated with the fermion parity symmetry $\Z_2^f$.
        
        Note that $\mathcal{D}_{\C}^2 = \mathcal{D}_{\C_0}^2 + \mathcal{D}_{\C_1}^2 = 2 \mathcal{D}_{\mathcal{C}}^2$
        implies that 
        \begin{align}
            \mathcal{D}_{\C_{1}}^2 = \mathcal{D}_{\C_0}^2 = \mathcal{D}_{\mathcal{C}}^2 
        \end{align}
        
        $\C_1$ can be decomposed according to whether the anyons can absorb the fermion:
        \begin{align}
            \C_1 = \C_v \oplus \C_\sigma ,
        \end{align}
        such that 
        \begin{align}
            a \times \psi \neq a \text{ if } a \in \C_v
            \nonumber \\
            a \times \psi = a \text{ if } a \in \C_\sigma
        \end{align}
        In general, this decomposition is not in any sense respected by the fusion rules; for example, fusing a (non-Abelian) anyon with a $\sigma$-type vortex can produce $v$-type fusion products, while fusing two $\sigma$-type vortices produces anyons (which automatically do not absorb $\psi$).
        
        We note that super-modular tensor categories allow a canonical gauge-fixing
        \begin{equation}
        \label{FgaugeFix}
            F^{a\psi \psi} = F^{\psi \psi a} = 1
        \end{equation}
        for all $a \in \mathcal{C}$. In the standard BFC diagrammatic calculus, this gauge-fixing allows fermion lines to be ``bent" freely.
        
        One technical issue which will play a key role in the rest of this paper is to characterize sets of phases $\zeta_a \in \U$ for $a \in \mathcal{C}$ such that
        \begin{equation}
            \zeta_a \zeta_b = \zeta_c \text{ if }N_{ab}^c>0.
        \end{equation}
        Strictly speaking these are functions from the set of anyon labels to $\U$; mathematically, $d_a \zeta_a$ defines a character of the Verlinde ring for the super-modular category. Such functions form an Abelian group which we call $K(\mathcal{C})$, following similar notation in~\cite{galindo}. Clearly $\zeta_\psi = \pm 1$, which gives a $\Z_2$ grading on $K(\mathcal{C})$. The set of such functions with $\zeta_\psi = +1$ form a subgroup $K_+(\mathcal{C}) \subset K(\mathcal{C})$.
        
        One important property of $K(\mathcal{C})$ is the following, proven in Appendix~\ref{app:characters}: if $\zeta_a \in K(\mathcal{C})$, then
        \begin{equation}
            \zeta_a = M_{a,x}
        \end{equation}
        for some $x \in \C$, where $\C$ is any minimal modular extension of $\mathcal{C}$. In particular, if $\zeta_a \in K_+(\mathcal{C})$, then $x \in \A$. Since $M_{a,x}=M_{a,x\times \psi}$, we conclude that $K_+(\mathcal{C}) \simeq \A/\{1,\psi\}$.
        
        There are several possibilities for the full group $K(\mathcal{C})$, which we briefly overview now and discuss further in Sec.~\ref{sec:LRNatIso}. It may be that $K(\mathcal{C})=K_+(\mathcal{C})$, that is, there simply does not exist a set of phases which obey the fusion rules with $\zeta_\psi=-1$. Alternatively, such phases may indeed exist, in which case $K(\mathcal{C})$ is a group extension of $\Z_2$ given by the short exact sequence
        \begin{equation}
            1 \rightarrow K_+(\mathcal{C}) \stackrel{i}{\rightarrow} K(\mathcal{C}) \stackrel{r_\psi}{\rightarrow} \Z_2 \rightarrow 1
        \end{equation}
        where $i$ is the inclusion map and $r_\psi$ restricts $\zeta_a$ to $a=\psi$. This group extension is given by an element $[\lambda] \in \H^2(\Z_2,K_+(\mathcal{C})) \simeq \H^2(\Z_2,\A/\{1,\psi\})$; although the cohomology group is quite simple, we do not know a general method to compute the particular element $[\lambda]$ from the fusion rules. In the special case where $\mathcal{C}$ splits, i.e., can be written $\mathcal{C} = \mathcal{B} \boxtimes \{1,\psi\}$ for some modular $\mathcal{B}$, then $K(\mathcal{C}) = K_+(\mathcal{C}) \times \Z_2$, that is, $[\lambda] = 0$.

        \subsection{Fermionic topological phases}
        
        Consider a (2+1)D system with a Hilbert space which is a tensor product of local Hilbert spaces containing fermions governed by a local Hamiltonian. (We will refer to this as a ``microscopic system.") We assume that the Hamiltonian has a gap in the thermodynamic limit. Any such systems which can be continuously connected without closing the gap (allowing the addition of fermionic degrees of freedom with a trivial Hamiltonian) are said to be in the same fermionic topological phase. As stated above, fermionic topological phases are believed to be fully characterized by a super-modular tensor category $\mathcal{C}$ together with a chiral central charge $c_-$. Equivalently, a fermionic topological phase can be characterized by a spin modular category $\C$, together with a choice of chiral central charge $c_-$. The spin modular category determines the chiral central charge modulo $8$, while the super-modular tensor category only determines the chiral central charge modulo $1/2$. 
        
        We see that $\mathcal{C}$ describes the anyon content of the fermionic topological phase while $c_- \text{ mod 8}$ specifies the minimal modular extension. The spin modular category $\C$ determines the fusion and braiding of fermion parity vortices which is a crucial part of data that specifies the full fermionic topological phase. In particular, to define the fermionic system on non-trivial surfaces and with arbitrary spin structures (i.e. arbitrary boundary conditions), we need the full spin modular theory, since this requires different patterns of fermion parity flux through non-contractible cycles  (see e.g. \cite{delmastro2021}).
        
        In the present paper, we restrict our attention to symmetry fractionalization for $\mathcal{C}$, i.e., we do not characterize how symmetries act on the minimal modular extension $\C$. Attempting to lift the SET data from $\mathcal{C}$ to $\C$ leads to a cascade of obstructions which characterize the 't Hooft anomaly of the fermionic SET and will be studied in upcoming work~\cite{bulmash2021Anomalies}. Some, but not all, of these obstructions have been understood previously in the UBFC framework~\cite{fidkowski2018}.
        
        In using super-modular categories to model fermionic topological phases, the following technical issue arises. In general in a fusion category there are vertex basis gauge transformations, $\Gamma^{ab}_c$, which are basis transformations in the splitting space $V^{ab}_c$, i.e. $\Gamma^{ab}_c: V^{ab}_c \rightarrow V^{ab}_c$. We will assume that the vertex basis gauge transformations must always satisfy 
        \begin{align}
        \label{GammaFix}
            \Gamma^{\psi,\psi}_1  = 1.
        \end{align}
        We do not have a completely satisfactory microscopic justification for this assumption, although we note allowing $\Gamma^{\psi,\psi}_1 \neq 1$ leads to consequences that contradict a number of known results. For example, one can show that this gauge freedom may be used to identify the two fermionic symmetry fractionalization patterns of the semion-fermion topological order with $G_f=\Z_4^{{\bf T},f}$. These two symmetry fractionalization patterns respectively appear on the surfaces of the $\nu=2$ and $\nu=-2$ elements of the (3+1)D DIII topological superconductors; considering them to be gauge-equivalent would collapse the known $\Z_{16}$ classification down to $\Z_4$.
        We note that the requirement of Eq. \ref{GammaFix} is always compatible with the gauge fixing in Eq. \ref{FgaugeFix}.
        
\section{Review of symmetry fractionalization in bosonic systems}
        \label{sec:bosonicSymmFrac}
        
        In this section we review the formalism of~\cite{barkeshli2019} describing symmetry fractionalization in bosonic systems. The starting point is a UMTC $\mathcal{B}$ and a symmetry group $G$. 
        
        \subsection{Topological symmetries}
        
        A unitary topological symmetry, or braided autoequivalence, of $\mathcal{B}$ is an invertible map
        \begin{equation}
            \rho: \mathcal{B} \rightarrow \mathcal{B}
        \end{equation}
        which preserves all topological data. In particular, gauge-invariant quantities are left invariant, while gauge-dependent quantities are left invariant up to a gauge transformation. One can also define anti-unitary topological symmetries, which complex conjugate the data, up to gauge transformations. 
        
        Certain unitary braided autoequivalences are ``trivial" in that they leave all of the basic data of the theory completely unchanged. These autoequivalences are called natural isomorphisms, and their action is written
        \begin{equation}
            \Upsilon\left(\ket{a,b;c}\right) = \frac{\gamma_a \gamma_b}{\gamma_c}\ket{a,b;c}
        \end{equation}
        with $\gamma_a \in \U$ for all $a \in \mathcal{B}$. In bosonic systems, modifying a braided autoequivalence by a natural isomorphism is a form of gauge freedom. We therefore define the group $\Aut(\mathcal{B})$ to be the group of braided autoequivalences of $\mathcal{B}$ modulo natural isomorphisms. 
        
        Natural isomorphisms themselves have redundancy, in that modifying
        \begin{equation}
            \gamma_a \rightarrow \gamma_a \zeta_a
        \end{equation}
        for phases $\zeta_a$ which obey the fusion rules, that is, $\zeta_a \zeta_b = \zeta_c$ whenever $N_{ab}^c>0$, does not change the action of the natural isomorphism on any fusion space. This redundancy will be particularly important when we consider the fermionic case.
        
        The first step towards specifying how a symmetry acts on a bosonic topological phase is to choose a group homomorphism
        \begin{equation}
            [\rho_{\bf g}]: G \rightarrow \Aut(\mathcal{B}).
        \end{equation}
        Choosing a representative $\rho_{\bf g}$ of the class $[\rho_{\bf g}]$ specifies data $U_{\bf g}(a,b;c)$ via the equation
        \begin{equation}
            \rho_{\bf g}(\ket{a,b;c}) = U_{\bf g}\left(\,^{\bf g}a,\,^{\bf g}b;\,^{\bf g}c\right)\ket{\,^{\bf g}a,\,^{\bf g}b;\,^{\bf g}c},
        \end{equation}
        where $|a,b;c\rangle \in V_{ab}^c$ is a state in a splitting space of $\mathcal{B}$. 
        This data is subject to the consistency conditions that the $F$- and $R$-symbols are preserved:
        \begin{align}
            U_{\bf g}(\,^{\bf g}a,\,^{\bf g}b;\,^{\bf g}e)U_{\bf g}(\,^{\bf g}c,\,^{\bf g}e;\,^{\bf g}d)F^{\,^{\bf g}a,\,^{\bf g}b,\,^{\bf g}c}_{\,^{\bf g}d,\,^{\bf g}e,\,^{\bf g}f}U_{\bf g}^{-1}(\,^{\bf g}b,\,^{\bf g}c;\,^{\bf g}f)U_{\bf g}^{-1}(\,^{\bf g}a,\,^{\bf g}f;\,^{\bf g}d) &= \left(F^{abc}_{def}\right)^{\sigma({\bf g})} \label{eqn:UFconsistency}\\
            U_{\bf g}(\,^{\bf g}b,\,^{\bf g}a;\,^{\bf g}c)R^{\,^{\bf g}a,\,^{\bf g}b}_{\,^{\bf g}c} U_{\bf g}^{-1}(\,^{\bf g}a,\,^{\bf g}b;\,^{\bf g}c) &= \left(R^{ab}_c\right)^{\sigma({\bf g})}. \label{eqn:URconsistency}
        \end{align}
        Here 
        \begin{equation}
            \sigma({\bf g}) = \begin{cases}
                1 & {\bf g} \text{ unitary}\\
                \ast & {\bf g} \text{ anti-unitary}
            \end{cases}
        \end{equation}
        Modifying the representative $\rho_{\bf g}$ by a ${\bf g}$-dependent natural isomorphism, $\rho_{\bf g} \rightarrow \Upsilon_{\bf g} \circ \rho_{\bf g}$, is a form of gauge freedom which changes
        \begin{equation}
            U_{\bf g}(a,b;c) \rightarrow \frac{\gamma_a({\bf g}) \gamma_b({\bf g})}{\gamma_c({\bf g})} U_{\bf g}(a,b;c).
            \label{eqn:UGaugeTransform}
        \end{equation}
        Since $[\rho_{\bf g}]$ is a group homomorphism, the map
        \begin{equation}
            \kappa_{\bf g,h} = \rho_{\bf gh}\circ \rho_{\bf h}^{-1} \circ \rho_{\bf g}^{-1}
        \end{equation}
        is a natural isomorphism. Translated into the action on fusion spaces, this means
        \begin{equation}
            \kappa_{\bf g,h}(a,b;c) = \frac{\beta_a({\bf g,h})\beta_b({\bf g,h})}{\beta_c({\bf g,h})} = U_{\bf g}^{-1}(a,b,c)\left(U_{\bf h}^{-1}(\,^{\overline{\bf g}}a,\,^{\overline{\bf g}}b;\,^{\overline{\bf g}}c)\right)^{\sigma({\bf g})}U_{\bf gh}(a,b;c)
        \end{equation}
        for some phases $\beta_a({\bf g,h})$. We use the compact notation ${\overline{\bf g}}={\bf g}^{-1}$. Thanks to the redundancies in natural isomorphisms discussed above, if $\nu_a({\bf g,h})$ are phases that obey the fusion rules for each ${\bf g,h}$, then the data $\beta_a({\bf g,h})$ and $\beta_a({\bf g,h})\nu_a({\bf g,h})$ should be considered gauge equivalent. Modifying $\rho_{\bf g}$ by a natural isomorphism also modifies
        \begin{equation}
            \beta_a({\bf g,h})\rightarrow\frac{\gamma_a({\bf gh})}{\gamma_a({\bf g})\gamma_{\,^{\overline{\bf g}}a}({\bf h})^{\sigma({\bf g})}}\beta_a({\bf g,h}),
            \label{eqn:betaGammaTransform}
        \end{equation}
        so these $\beta_a$ should also be considered gauge equivalent. We use the canonical gauge-fixing $U_{\bf g}(1,a,a)=U_{\bf g}(a,1,a) = 1$ for all $a \in \mathcal{B}$, where the identity anyon in $\mathcal{B}$ is denoted $1$. Maintaining this gauge fixing requires $\gamma_1({\bf g})=1$ in all gauge transformations.
        
     \subsection{Symmetry localization and fractionalization}
     
     Now suppose that we have a bosonic quantum many-body Hilbert space with a local Hamiltonian such that the system is in the phase given by the topological order $\mathcal{B}$. Let $R_{\bf g}$ be the representation of ${\bf g} \in G$ on the quantum many-body Hilbert space. We assume that $R_{\bf g}$ is generated onsite, that is,
     \begin{equation}
     R_{\bf g} = \prod_i R_{\bf g}^{(i)} K^{q({\bf g})}
     \label{eqn:localGeneration}
     \end{equation}
     where the $R_{\bf g}^{(i)}$ are local (that is, supported on a region with length on the order of the correlation length or smaller) unitary operators on disjoint patches $i$ of space, $q({\bf g})=0$ if $R_{\bf g}$ has a unitary action, and $q({\bf g})=1$ if $R_{\bf g}$ has an anti-unitary action. 
     Let 
     $\ket{\Psi_{\{a_i\}}}$ be a state of the many-body system with anyons $a_i$ localized at well-separated positions $i$. Then, at least in principle, one may determine the map $[\rho_{\bf g}]:G \rightarrow \Aut(\mathcal{B})$, where the autoequivalences are defined modulo natural transformations.
     
     Upon defining $[\rho_{\bf g}]$, we may ask if the symmetry can be localized. That is, we ask if the global symmetry operator $R_{\bf g}$ acting on $\ket{\Psi_{\{a_i\}}}$ can be decomposed, up to exponentially small corrections in the ratio of the correlation length and the separation between anyons, according to the following ansatz
     \begin{equation}
         R_{\bf g}\ket{\Psi_{\{a_i\}}} \approx \prod_i U_{\bf g}^{(i)} \rho_{\bf g}\ket{\Psi_{\{a_i\}}},
         \label{eqn:symmFracAnsatz}
     \end{equation}
     where $\rho_{\bf g}$ is an operator which acts only on the topological data of the state $\ket{\Psi_{\{a_i\}}}$ and whose action is given by the element $[\rho_{\bf g}]$ of $\Aut(\mathcal{B})$, and $U_{\bf g}^{(i)}$ is a local operator near the anyon at position $i$. 
     
     One can show~\cite{barkeshli2019} that localizing the symmetry amounts to choosing a set of phases $\eta_a({\bf g,h})$. These phases define the symmetry fractionalization data and which characterize the extent to which the local operators $U_{\bf g}^{(i)}$ fail to obey the group law:
     \begin{equation}
         \eta_{a_j}({\bf g,h})U_{\bf gh}^{(j)}\ket{\Psi_{\{a_i\}}} = U_{\bf g}^{(j)}\rho_{\bf g}U_{\bf h}^{(j)}\rho_{\bf g}^{-1}\ket{\Psi_{\{a_i\}}}
         \label{eqn:etaDef}
     \end{equation}
     The $\eta_a$ obey consistency conditions. One arises from enforcing associativity of the $U_{\bf g}^{(j)}$:
     \begin{equation}
         \eta_a({\bf g,h})\eta_a({\bf gh,k}) = \eta^{\sigma({\bf g})}_{\,^{\overline{\bf g}}a}({\bf h,k})\eta_a({\bf g,hk}).
         \label{eqn:etaConsistency}
     \end{equation}
     The other consistency condition enforces the consistency between the local data and the global part of the symmetry $\rho_{\bf g}$:
     \begin{equation}
         \frac{\eta_c({\bf g,h})}{\eta_a({\bf g,h})\eta_b({\bf g,h})} = U^{\sigma({\bf g})}_{\bf h}(\,^{\overline{\bf g}}a,\,^{\overline{\bf g}}b;\,^{\overline{\bf g}}c)U_{\bf g}(a,b;c)U^{-1}_{\bf gh}(a,b;c) = \kappa^{-1}_{\bf g,h}(a,b;c).
         \label{eqn:etaUConsistency}
     \end{equation}
     Using the explicit form of $\kappa$, we find that
     \begin{equation}
         \frac{\beta_a({\bf g,h})\beta_b({\bf g,h})}{\beta_c({\bf g,h})} = \frac{\eta_a({\bf g,h})\eta_b({\bf g,h})}{\eta_c({\bf g,h})}
     \end{equation}
     so we may re-encode the symmetry fractionalization data into a set of phases
     \begin{equation}
         \omega_a({\bf g,h})= \frac{\beta_a({\bf g,h})}{\eta_a({\bf g,h})}
         \label{eqn:omegaInTermsOfBetaEta}
     \end{equation}
     These $\omega_a$ need not be 1, but they do have the convenient property that they obey the fusion rules. We emphasize that, given $\rho_{\bf g}$, the phases $\omega_a$ and $\eta_a$ are equivalent encodings of the symmetry fractionalization data, and we may choose to work with either one depending on convenience. By an argument which we review and generalize in Appendix~\ref{app:characters}, the fact that the $\omega_a$ obey the fusion rules means that, since $\mathcal{B}$ is modular,
     \begin{equation}
         \omega_a({\bf g,h})=M_{a,\cohosub{w}({\bf g,h})}
         \label{eqn:omegaAsMutualStats}
     \end{equation}
     where $\coho{w} \in \A$, with $\A$ the set of Abelian anyons of $\mathcal{B}$, and $M_{ab}$ is the scalar monodromy between the anyons $a,b$.
     
     There is an obstruction to the symmetry localization ansatz of Eq. \ref{eqn:symmFracAnsatz} being consistent with associativity of $R_{\bf g}$, which precludes defining any consistent symmetry fractionalization pattern. This obstruction, called the symmetry localization obstruction, is an element $[\O]\in \H^3(G,\A)$. One finds this obstruction by constructing the following phase factors
     \begin{equation}
         \Omega_a({\bf g,h,k}) = \frac{\beta^{\sigma({\bf g})}_{\,^{\overline{\bf g}}a}({\bf h,k})\beta_a({\bf g,hk})}{\beta_a({\bf g,h})\beta_a({\bf gh,k})}
     \end{equation}
     One can show that these phases obey the fusion rules. Accordingly,
     \begin{equation}
         \Omega_a({\bf g,h,k})=M_{a,\cohosub{O}({\bf g,h,k})}
     \end{equation}
     where $\O({\bf g,h,k}) \in \A$. If the symmetry can be localized, we can use Eqs.~\ref{eqn:etaConsistency} and \ref{eqn:omegaInTermsOfBetaEta} to write an equivalent expression for $\Omega_a$, namely
     \begin{equation}
         \Omega_a({\bf g,h,k}) = \frac{\omega^{\sigma({\bf g})}_{\,^{\overline{\bf g}}a}({\bf h,k})\omega_a({\bf g,hk})}{\omega_a({\bf g,h})\omega_a({\bf gh,k})}
     \end{equation}
     which translates into the equation
     \begin{equation}
         \O({\bf g,h,k}) = (d\coho{w})({\bf g,h,k})
         \label{eqn:standardH3Obstruction} ,
     \end{equation}
     where $d$ is the differential in the group cohomology. That is, the existence of a consistent symmetry fractionalization pattern requires that $[\O]$ be trivial as an element of $\H^3(G,\A)$. Equivalently, if $[\O]$ is not trivial in $\H^3(G,\A)$, then the symmetry localization ansatz Eq. \ref{eqn:symmFracAnsatz} is inconsistent and thus we have an obstruction to obtaining any consistent symmetry fractionalization pattern. 
     
     Another perspective on the symmetry localization obstruction is that the TQFT defined by $\mathcal{B}$ is only compatible with a 2-group symmetry, where $G$ and $\mathcal{A}$ are the 0-form and 1-form symmetries, respectively, and $[\O] \in \mathcal{H}^3(G, \mathcal{A})$ characterizes the 2-group.\cite{barkeshli2019,benini2019} 
     
     Symmetry fractionalization data is subject to a set of gauge transformations that arise from ambiguities in the ansatz of Eq.~\ref{eqn:symmFracAnsatz}. One may freely modify the local operators $U_{\bf g}^{(i)}$ to act on states via
     \begin{equation}
         U_{\bf g}^{(j)}\ket{\Psi_{\{a_i\}}} \rightarrow \gamma_{a_j}({\bf g})^{-1} U_{\bf g}^{(j)} \ket{\Psi_{\{a_i\}}}
     \end{equation}
     for any $\U$ phases $\gamma_{a_j}({\bf g})$; this will not change the action of $R_{\bf g}$ as long as there is a compensating modification of the operator $\rho_{\bf g}$ by the natural isomorphism given by $\gamma_a({\bf g})$. This transformation of the local operators, when inserted into the symmetry fractionalization ansatz, effectively modifies $\rho_{\bf g}$ by the natural isomorphism given by $\gamma_a({\bf g})$. This transformation modifies $U$ and $\eta$ according to Eq.~\ref{eqn:UetaGaugeTransform} while also redefining $\beta_a$ according to Eq.~\ref{eqn:betaGammaTransform}, but it leaves $\omega_a$ invariant. On the other hand, redefining $\beta_a \rightarrow \beta_a \nu_a$ using the redundancy of natural isomorphisms redefines $\omega_a \rightarrow \omega_a \nu_a$ while leaving $\eta_a$ and $U_{\bf g}$ invariant.
     
     Finally, symmetry fractionalization patterns in bosonic systems form an $\H^2(G,\A)$ torsor. To see this, we note that different solutions $\coho{w}$ and $\coho{w}'$ of Eq.~\ref{eqn:standardH3Obstruction} are related by
     \begin{equation}
         \coho{w}'({\bf g,h})=\coho{t}({\bf g,h})\times \coho{w}({\bf g,h})
     \end{equation}
     for some choice of $\coho{t}\in Z^2(G,\A)$. Tracing through the definitions, one finds that $\coho{t}$ transforms by a coboundary if the $U_{\bf g}^{(i)}$ are modified by a natural isomorphism which is equivalent to the identity natural isomorphism, that is, if we choose a natural isomorphism for which the $\gamma_a({\bf g})$ obey the fusion rules. Hence only different $[\coho{t}] \in \H^2(G,\A)$ produce different symmetry fractionalization patterns.
        
\section{Symmetry localization in fermionic systems}
\label{fermSymLocSec}
        
        We will build our theory of symmetry fractionalization in fermionic systems by applying the basic formalism of Sec.~\ref{sec:bosonicSymmFrac} to a super-modular tensor category and then demanding that the symmetry localization ansatz be compatible with the full fermionic symmetry group $G_f$ and the locality of the fermion. In this section, we assume that the symmetry can be localized in the sense of Sec.~\ref{sec:bosonicSymmFrac} and determine the constraints required for this compatibility. In the subsequent section we will then study the fundamental obstructions to symmetry localization. 
        
        \subsection{Fermionic symmetries}
        \label{sec:fermionicSymm}
        
        Consider a microscopic fermionic system in the sense described in Sec.~\ref{sec:fermionicSymm}, that is, we have a many-body Hilbert space which is a tensor product of local fermionic Hilbert spaces. We assume there is a local Hamiltonian with a gap such that the system is in a fermionic topological phase associated to the super-modular tensor category $\mathcal{C}$ with transparent fermion $\psi$. Such a system has a symmetry group $G_f$, which is the group of transformations of fermionic operators that keep the Hamiltonian and (by assumption) ground state invariant. The fermion parity operator $(-1)^F$ is defined from the Hilbert space and determines the $\Z_2^f$ subgroup of the full fermionic symmetry group $G_f$. Then the group $G_b = G_f/\Z_2^f$ describes the set of transformations of all bosonic operators that keep the Hamiltonian and ground state invariant.
        
        To each element ${\bf g} \in G_b$ we have an operator $R_{\bf g}$ on the full Hilbert space which implements the symmetry. For the discussion below, we assume the $R_{\bf g}$ are locally generated in the sense of Eq.~\ref{eqn:localGeneration}. Nevertheless we expect that the final results basically hold for spatial symmetries as well\footnote{See e.g. \cite{manjunath2020} for a recent discussion on the applicability of applying the $G$-crossed braided tensor category formalism to spatial symmetries.}.
        
        The operators $R_{\bf g}$ and $(-1)^FR_{\bf g}$ are physically distinct (for example, if the fermions carry spin, a global spin flip is distinct from a spin flip times fermion parity) and are fixed from the outset; this fact will play an important role later.
        
        In general, the $R_{\bf g}$ operators do not form a linear representation of $G_b$; instead, they multiply projectively on states of the many-body quantum system with odd fermion parity, that is,
        \begin{equation}
            R_{\bf g}R_{\bf h} = \left(\omega_2({\bf g,h})\right)^F R_{\bf gh}
        \end{equation}
        where $F$ is the fermion parity operator, with eigenvalues $0$ and $1$, and $\omega_2({\bf g,h}) \in \{\pm 1\} \simeq \Z_2$.
        
        Since both the local fermion operators and $R_{\bf g}$ are defined in the microscopic Hilbert space, the local fermion operators have fixed transformation rules. That is, let $f_{i,\alpha}$ be a basis of (Majorana) fermionic local operators localized near position $i$. Then, since $R_{\bf g}$ is locally generated,
        \begin{equation}
            R_{\bf g} f_{i,\alpha} R_{\bf g}^{-1} = \left(\tilde{U}_{\bf g}^{(i)}\right)_{\alpha \beta}f_{i,\beta} ,
            \label{eqn:microscopicFermionTransform}
        \end{equation}
        where $i$ labels the position of the local fermion operator $f$, $\alpha,\beta$ label local degrees of freedom, and $\tilde{U}_{\bf g}^{(i)}$ is some matrix acting only on local degrees of freedom. We include $i$-dependence on $\tilde{U}_{\bf g}^{(i)}$ for full generality; this dependence disappears only if the local fermion Hilbert space and the symmetry action on it are translation invariant.
        
        Demanding that the $R_{\bf g}$ operators multiply associatively enforces that $\omega_2 \in Z^2(G_b,\Z_2)$. As such, $G_f$ is a central extension of $G_b$ by $\Z_2^f$, described by the short exact sequence
        \begin{align}
            1 \rightarrow \Z_2^f \rightarrow G_f \rightarrow G_b \rightarrow 1
        \end{align}
        and characterized by the cohomology class  $[\omega_2] \in \H^2(G_b, \Z_2)$. When $[\omega_2]$ is trivial, then $G_f = \Z_2^f \times G_b$. When $[\omega_2]$ is non-trivial, then it follows that the fermions carry fractional quantum numbers of $G_b$ which cannot be screened by any bosons. 
        
        With these definitions in hand, our starting point to describe the localization of $G_f$ on the fermionic topological order described by $\mathcal{C}$ is to simply write down a theory of $G_b$ symmetry localization on $\mathcal{C}$ in the sense of Eq.~\ref{eqn:symmFracAnsatz}. Assuming the symmetry can be localized, we will obtain symmetry fractionalization data $\eta_a({\bf g,h})$ in the usual way via Eq.~\ref{eqn:etaDef}. However, we must modify the theory to account for two things. First, $\psi$ is a local excitation, not simply an emergent one. Second, we must account for the presence of the group extension $\omega_2$. We will find that if our symmetry fractionalization pattern for $G_b$ is to describe a theory with a local $\psi$ and $G_f$ symmetry, then the symmetry fractionalization pattern is constrained in a few ways.
        
        Consider a basis of states $\ket{\Psi^\alpha_{\psi_0;(a_i;c)}}$ on a closed manifold for the topological sector containing a set of anyons $a_i$ at positions $i=1,2,\ldots,n$ that have a definite total fusion channel $c$ and which also contains a fermion at position 0. The label $\alpha$ labels local degrees of freedom near position 0, in the sense that given the basis of (Majorana) fermionic operators $f_{0,\alpha}$ near position 0, we define
        \begin{equation}
            \ket{\Psi^\alpha_{\psi_0;(a_i;c)}} = f_{0,\alpha}\ket{\Psi_{(a_i;c)}},
        \end{equation}
        for every topological sector $(a_i;c)$, where $\ket{\Psi_{(a_i;c)}}$ is a state which is locally in the vacuum state at position 0 (we suppress indices for local degrees of freedom away from position 0). For this state to be realizable microscopically, we require
        \begin{equation}
        c \in \{1,\psi\}.
        \end{equation}
        We can evaluate the action of $R_{\bf g}$ on the state in two ways. First, we can pull out a local fermion operator before applying the symmetry transformation:
        \begin{align}
            R_{\bf g}\ket{\Psi^\alpha_{\psi_0;(a_i;c)}} &= R_{\bf g}f_{0,\alpha}\ket{\Psi_{(a_i;c)}} \\
            &= \sum_\beta \left(\tilde{U}^{(0)}_{\bf g}\right)_{\alpha \beta}f_{0,\beta}R_{\bf g}\ket{\Psi_{(a_i;c)}}\\
            &= \sum_{\beta} \left(\tilde{U}^{(0)}_{\bf g}\right)_{\alpha \beta} f_{0,\beta} \prod_i U_{\bf g}^{(i)} U_{\bf g}(\{\,^{\bf g}a_i\};c)\ket{\Psi_{(\,^{\bf g}a_i;c)}} \label{eqn:symmActionLastLine},
        \end{align}
        where $\beta$ runs over the local basis of fermion operators. Alternatively, we can apply the symmetry transformation first:
        \begin{align}
            R_{\bf g}\ket{\Psi^\alpha_{\psi_0;(a_i;c)}} &= \prod_{k=0,i} U_{\bf g}^{(k)}U_{\bf g}(\psi_0;(\{\,^{\bf g}a_i\};c))\ket{\Psi^\alpha_{\psi_0;(\,^{\bf g}a;c)}}\\
            &= \prod_{k=0,i} U_{\bf g}^{(i)}U_{\bf g}(\psi,c;c \times \psi)U_{\bf g}(\{\,^{\bf g}a_i\};c)f_{0,\alpha}\ket{\Psi_{(\,^{\bf g}a_i;c)}} \\
            &= U_{\bf g}^{(0)}f_{0,\alpha} \prod_{i} U_{\bf g}^{(i)}U_{\bf g}(\psi,c;c \times \psi)U_{\bf g}(\{\,^{\bf g}a_i\};c)\ket{\Psi_{(\,^{\bf g}a_i;c)}}, \label{eqn:alternateSymmActionLastLine}
        \end{align}
        The symbol $U_{\bf g}(\{\,^{\bf g}a_i\};c)$ is shorthand for the $U$-symbol for the entire fusion tree of anyons away from position 0, and likewise $U_{\bf g}(\psi_0;(\{\,^{\bf g}a_i\};c))$ is the $U$-symbol for the entire fusion tree of the whole state. The fact that the same label $\alpha$ appears in both cases arises from the fact that we are using the same basis of local operators near position 0 to define the basis of states in each topological sector, which defines the $U_{\bf g}^{(0)}$ operators. Comparing Eqs.~\ref{eqn:symmActionLastLine} and~\ref{eqn:alternateSymmActionLastLine}, we obtain an important equation
        \begin{equation}
            \sum_\beta \left(\tilde{U}^{(0)}_{\bf g} \right)_{\alpha \beta} f_{0,\beta} \ket{\Psi_{(\,^{\bf g}a_i;c)}}= U_{\bf g}^{(0)}f_{0,\alpha} U_{\bf g}(\psi,c;c \times \psi)\ket{\Psi_{(\,^{\bf g}a_i;c)}},
            \label{eqn:localTopologicalFermionConsistency}
        \end{equation}
        for all $\alpha$. The first thing to observe from this equation is that the left-hand side consists of a local operator acting on a state, and the right-hand side is a local operator acting on the same state multiplied by the phase $U_{\bf g}(\psi,c;c \times \psi)$, which depends on the non-local overall fusion channel $c$ of the state in question. In order for this to be true, we must have that $U_{\bf g}(\psi,c;c \times \psi)$ is independent of $c$. Since $c \in \{1,\psi\}$, this implies
        \begin{equation}
            U_{\bf g}(\psi,\psi;1) = U_{\bf g}(\psi,1; \psi) = 1.
            \label{eqn:UpsiEquals1}
        \end{equation}
        This is one of the primary constraints on the SET data and is needed to enforce the locality of $\psi$. This turns Eq.~\ref{eqn:localTopologicalFermionConsistency} into
        \begin{equation}
            \sum_\beta \left(\tilde{U}^{(0)}_{\bf g} \right)_{\alpha \beta} f_{0,\beta} \ket{\Psi_{(\,^{\bf g}a_i;c)}}= U_{\bf g}^{(0)}f_{0,\alpha}\ket{\Psi_{(\,^{\bf g}a_i;c)}} \text{ if }c\in\{1,\psi\}.
            \label{eqn:localTopologicalFermionConsistencyNoU}
        \end{equation}
        
        The next observation is that the action of $U_{\bf g}^{(0)}$ on states with an anyon $\psi$ at position 0 is determined entirely by its action on the basis of local fermion operators, that is, by $\tilde{U}_{\bf g}^{(0)}$.  Normally we can modify the local operators $U_{\bf g}^{(i)}$ by a phase $\gamma_{a_i}({\bf g})$, which would transform the right-hand side of Eq.~\ref{eqn:localTopologicalFermionConsistencyNoU} by $\gamma_\psi({\bf g})$. However, the left-hand side of Eq.~\ref{eqn:localTopologicalFermionConsistencyNoU} is unambiguously determined by the microscopic action of the symmetry on local operators. Correspondingly, we have no freedom in the right-hand side and thus we obtain the important constraint
        \begin{equation}
            \gamma_\psi({\bf g})=1
            \label{eqn:gammaPsiEquals1}
        \end{equation}
        in all gauge transformations. Physically speaking, the local action $U_{\bf g}^{(0)}$ of the symmetry on states with a fermion at position $0$ is completely fixed by the symmetry action on the microscopic Hilbert space as determined by  Eq.~\ref{eqn:microscopicFermionTransform}.
        
        To further understand the constraint Eq.~\ref{eqn:gammaPsiEquals1}, we comment that a transformation $\gamma_\psi({\bf g})=-1$ amounts to a redefinition of the microscopic $G_b$ symmetry operators by a ${\bf g}$-dependent factor of fermion parity, that is, relabeling
        \begin{equation}
            R_{\bf g} \rightarrow R_{\bf g}(\gamma_{\psi}({\bf g}))^{F}.
            \label{eqn:RgRedef}
        \end{equation}
        For a familiar example, if $G_b = \Z_2^{\bf T}$, then this redefinition interchanges the operators ${\bf T}\leftrightarrow {\bf T}(-1)^F$. As discussed at the beginning of this section, these operators are physically distinct thanks to the locality of the fermion. We should therefore consider this to be a different theory rather than a gauge-equivalent one. As we saw at the beginning of this section, a microscopic set of symmetry operators defines $\omega_2$ at the cocycle level, not just at the level of cohomology.
        
        Also note that the above transformation, Eq. \ref{eqn:RgRedef}, implements an automorphism of $G_f$ which changes the decomposition of $G_f$ into $G_b \times \Z_2^f$ as a set, and if $\gamma_\psi({\bf g})$ is not closed as a 1-cochain, this will also change the representative cocycle $\omega_2$ for the group extension. 
        
        Finally, we can use the above to constrain $\eta_\psi$. Suppose that we can define consistent symmetry fractionalization. The symmetries cannot permute $\psi$, so $\eta_\psi({\bf g,h})$ is constrained to be a $\U$-valued 2-cocycle. Inserting Eq.~\ref{eqn:UpsiEquals1} into the consistency condition Eq.~\ref{eqn:etaUConsistency} forces
        \begin{align}
            \eta_\psi({\bf g,h})^2 = 1,
        \end{align}
        that is, $\eta_\psi({\bf g,h}) \in Z^2(G_b,\Z_2)$.
        
        Next, observe that
        \begin{align}
            R_{\bf g}R_{\bf h}f_{i,\alpha} R_{\bf h}^{-1}R_{\bf g}^{-1}&=\omega_2({\bf g,h})R_{\bf gh}f_{i,\alpha} R_{\bf gh}^{-1} = \sum_\beta \omega_2({\bf g,h)}\left(\tilde{U}_{\bf gh}^{(i)}\right)_{\alpha \beta}f_{i,\beta}\\
            &=\sum_\beta R_{\bf g}\left(\tilde{U}_{\bf h}^{(i)}\right)_{\alpha \beta}f_{i,\beta} R_{\bf g}^{-1} = \sum_\delta \left(\,^{\bf g}\tilde{U}_{\bf h}^{(i)}\tilde{U}_{\bf g}^{(i)}\right)_{\alpha \delta}f_{i,\delta}.
        \end{align}
        Applying both far-right-hand-side expressions to a parity-odd state $\ket{\Psi_{(a_j;\psi)}}$, we obtain
        \begin{align}
            \sum_\beta \omega_2({\bf g,h)}\left(\tilde{U}_{\bf gh}^{(i)}\right)_{\alpha \beta} f_{i,\beta} \ket{\Psi_{(a_j;\psi)}} = \sum_\delta \left(\,^{\bf g}\tilde{U}_{\bf h}^{(i)}\tilde{U}_{\bf g}^{(i)}\right)_{\alpha \delta}f_{i,\delta}\ket{\Psi_{(a_j;\psi)}}.
        \end{align}
        We can now use Eq.~\ref{eqn:localTopologicalFermionConsistencyNoU} to exchange the $\tilde{U}^{(i)}$ matrices for the operators $U^{(i)}$ which appear in the symmetry localization ansatz, then absorb the fermion operators into the state:
        \begin{align}
            \omega_2({\bf g,h)}U_{\bf gh}^{(i)} f_{i,\alpha} \ket{\Psi_{(a_j;\psi)}} &= \,^{\bf g}U_{\bf h}^{(i)}U_{\bf g}^{(i)}f_{i,\alpha}\ket{\Psi_{(a_j;\psi)}}\\
            \Rightarrow \omega_2({\bf g,h)}\ket{\Psi^\alpha_{\psi_i;(a_j;\psi)}} &=  \left(U_{\bf gh}^{(i)}\right)^{-1}\,^{\bf g}U_{\bf h}^{(i)}U_{\bf g}^{(i)}\ket{\Psi^\alpha_{\psi_i;(a_j;\psi)}}\\
            &= \beta^{-1}_{\psi}({\bf g,h})W_{\bf g,h}^{(i)}\ket{\Psi^\alpha_{\psi_i;(a_j;\psi)}}\\
            &= \frac{\omega_{\psi}({\bf g,h})}{\beta_{\psi}(\bf g,h)}\ket{\Psi^\alpha_{\psi_i;(a_j;\psi)}}\\
            &= \eta^{-1}_{\psi}({\bf g,h})\ket{\Psi^\alpha_{\psi_i;(a_j;\psi)}}.
        \end{align}
        The last line follows from the definition Eq.~\ref{eqn:omegaInTermsOfBetaEta}, and here $W_{\bf g,h}^{(i)}$ is the operator
        \begin{equation}
            W_{\bf g,h}^{(i)} = \left(U_{\bf gh}^{(i)}\right)^{-1}\,^{\bf g}U_{\bf h}^{(i)}U_{\bf g}^{(i)}
        \end{equation}
        defined in~\cite{barkeshli2019} whose action on states with an anyon $a_i$ at position $i$ is
        \begin{equation}
            W_{\bf g,h}^{(i)}\ket{\Psi_{a_i}} = \omega_a({\bf g,h})\ket{\Psi_{a_i}}.
        \end{equation} 
        Note that $\omega_a$, which is a $\U$ phase that obeys the fusion rules in the sense
        \begin{equation}
            \omega_a({\bf g,h})\omega_b({\bf g,h}) = \omega_c({\bf g,h}) \text{ whenever } N_{ab}^c >0,
        \end{equation}
        and is used to characterize the symmetry fractionalization, should not be confused with $\omega_2 \in \{ \pm 1\}$, which specifies $G_f$ as a $\Z_2^f$ extension of $G_b$.
        
        We also comment that in the bosonic case we would proceed by writing Eq.~\ref{eqn:omegaAsMutualStats}. This step becomes more subtle in the fermionic case; we will discuss it in Sec.~\ref{sec:obstructions}.
        
        Since the group extension cocycle $\omega_2$ is $\Z_2$-valued, we conclude that
        \begin{equation}
            \eta_{\psi}({\bf g,h})=\omega_2({\bf g,h}) ,
            \label{eqn:etaPsiConstraint}
        \end{equation}
        so as claimed, we should incorporate the group extension as symmetry fractionalization data $\eta_{\psi}$.
        
        To summarize, we have found that locality of fermion operators requires that
        \begin{align}
            U_{\bf g}(\psi, \psi;1)&=1\\
            \eta_{\psi}({\bf g,h})&=\omega_2({\bf g,h})\\
            \gamma_{\psi}({\bf g}) &=1
        \end{align}
        
        \subsection{Locality-respecting natural isomorphisms}
        \label{sec:LRNatIso}
        
        In bosonic topological phases given by a UMTC $\mathcal{B}$, we define a group homomorphism $[\rho]: G \rightarrow \Aut(\mathcal{B})$, where elements of $\Aut(\mathcal{B})$ are defined modulo natural isomorphisms. In the fermionic case, we first need to enforce Eq.~\ref{eqn:UpsiEquals1}. Every element of $\Aut(\mathcal{C})$ has a representative with $U_{\bf g}(\psi,\psi;1)=+1$; given any choice of representative, if this constraint is not respected, then modify the autoequivalence by a natural isomorphism with $\gamma_\psi({\bf g}) = U_{\bf g}(\psi,\psi;1)^{-1/2}$ (either sign of the square root will work) to obtain a representative which respects the constraint. Hence the group of autoequivalences which respect the constraint $U_{\bf g}(\psi,\psi;1)=1$ is isomorphic to $\Aut(\mathcal{C})$; we will therefore refer to the former as $\Aut(\mathcal{C})$ as well, but we implicitly are disallowing any representatives which violate the constraint.
        
        Once we have accounted for the above constraint, as we saw above, natural isomorphisms with $\gamma_\psi({\bf g}) \neq +1$ do not in general respect the locality of the fermion. As such, we should only consider symmetries to be equivalent if they differ by a natural isomorphism which \textit{respects} the locality of the fermion. We define the group of equivalence classes of braided autoequivalences under this restricted equivalence to be the group $\Aut_{LR}(\mathcal{C})$, where $LR$ stands for ``locality-respecting."
        
        Accordingly, in a fermionic system we must specify a map
        \begin{equation}
            [\rho_{\bf g}]: G_b \rightarrow \Aut_{LR}(\mathcal{C}),
        \end{equation}
        such that  
        \begin{equation}
            \kappa_{\bf g,h} \circ \rho_{\bf g} \circ \rho_{\bf h} = \rho_{\bf gh},
            \label{eqn:kappaghDef}
        \end{equation}
        where $\kappa_{\bf g,h}$ is a natural isomorphism. In general $\kappa_{\bf g,h}$ need not be a locality-respecting natural isomorphism. When $\kappa_{\bf g,h}$ does respect locality, then $[\kappa_{\bf g,h}]$ is trivial and $[\rho_{\bf g}]$ is a faithful group homomorphism; otherwise the multiplication law for $[\rho_{\bf g}]$ holds up to a factor $[\kappa_{\bf g, h}]$.\footnote{An earlier version of this paper referred to $[\rho_{\bf g}]$ as a group homomorphism, implicitly assuming that $[\kappa_{\bf g,h}]$ is always trivial. However when $\Upsilon_\psi$ violates locality, $[\kappa_{\bf g,h}]$ may indeed be non-trivial, in which case the multiplication law for $[\rho_{\bf g}]$ need only hold projectively; an example is the unobstructed fractionalization of $G_f=\Z_4^{{\bf T},f}$ on $\mathcal{C}=\SO(3)_3$ in Sec.~\ref{subsec:symmLocObstructions}. The corrected discussion agrees with the discussion contained in Ref.~\cite{aasen21ferm}.}
        
        The definition of $\Aut_{LR}(\mathcal{C})$ is more subtle than one might naively expect. Given a microscopic realization of the symmetry, natural isomorphisms are obtained by modifying the local operators $U_{\bf g}^{(k)}$ that appear in the symmetry localization ansatz by anyon-dependent factors $\gamma_a({\bf g})$; from this starting point, only natural isomorphisms with $\gamma_\psi=+1$ are allowed. However, if the starting point is only UBFC data, natural isomorphisms are defined by their action on fusion vertices, i.e., they are autoequivalences $\Upsilon$ of the form
         \begin{equation}
            \Upsilon(\ket{a,b;c}) = \frac{\gamma_a \gamma_b}{\gamma_c}\ket{a,b,c}.
        \end{equation}
        At this level, natural isomorphisms have a redundancy upon redefining
        \begin{equation}
            \tilde{\gamma}_a = \zeta_a \gamma_a,
        \end{equation}
        where the $\zeta_a$ are phases such that $\zeta_a \zeta_b = \zeta_c$ whenever $N_{ab}^c >0$. Such a redefinition does not change the action of the symmetry on any fusion vertices. Therefore, if a natural isomorphism is equivalent under this redundancy to one with $\gamma_\psi = +1$, then it also respects locality, although at the level of microscopics, we must implement the natural isomorphism of the BFC using the equivalent $\gamma_\psi = +1$ realization.
        
        All phases $\zeta_a$ which respect the fusion rules have $\zeta_\psi = \pm 1$, so locality-respecting natural isomorphisms are required to have $\gamma_\psi = \pm 1$; note that this condition also preserves the constraint Eq.~\ref{eqn:UpsiEquals1} that $U(\psi,\psi;1)=+1$. There may or may not exist phases $\zeta_a$ which obey the fusion rules and have $\zeta_\psi = -1$. If such phases $\zeta_a$ do not exist, (in the language of Sec.~\ref{sec:supermodular} this means $K(\mathcal{C})=K_+(\mathcal{C})$), then the naive expectation holds: only natural isomorphisms with $\gamma_\psi = +1$ respect locality, and $\Aut(\mathcal{C}) \neq \Aut_{LR}(\mathcal{C})$. However, if such phases $\zeta_a$ do exist (in which case $K(\mathcal{C})/K_+(\mathcal{C})=\Z_2$), then all natural isomorphisms with $\gamma_\psi \in \Z_2$ are equivalent to one with $\gamma_\psi = +1$ and thus respect locality. In this case, $\Aut(\mathcal{C}) = \Aut_{LR}(\mathcal{C})$.
        
        In the former case where these phases $\zeta_a$ do not exist, there is a rather unfamiliar consequence that equivalence classes in $\Aut(\mathcal{C})$ are \textit{not} uniquely determined by the permutation action on the anyons. To see this, define
        \begin{equation}
            \Upsilon_\psi(\ket{a,b;c}) =\frac{\gamma_a \gamma_b}{\gamma_c}\ket{a,b,c}
        \end{equation}
        with $\gamma_a = +1$ for all $a \neq \psi$ and $\gamma_\psi = -1$. By construction, $\Upsilon_\psi$ has a trivial permutation action on the anyons, but if the aforementioned $\zeta_a$ do not exist, i.e., if $K(\mathcal{C})=K_+(\mathcal{C})$, then $\Upsilon_\psi$ is not a locality-respecting natural isomorphism and its equivalence class in $\Aut_{LR}(\mathcal{C})$ is therefore distinct from the class of the transformation which acts exactly as the identity.
        
        The map $\Upsilon_\psi$ has a natural interpretation as the action of fermion parity, since it inserts a factor of $(-1)$ for every local fermion in a state.
        
        For a general BFC $\mathcal{B}$, in many cases of physical interest, equivalence classes in $\Aut(\mathcal{B})$ modulo \textit{all} (possibly locality-violating) natural isomorphisms are uniquely determined by the way they permute the anyons. This property was proven explicitly in~\cite{benini2019} for theories with $N_{ab}^c\leq 1$ for all $a,b,c \in \mathcal{B}$ and where all $F$-symbols allowed by the fusion rules are nonzero. Some theories which do not have this property can be found in~\cite{davydov2014}.
        
        Suppose we have a theory in which given an autoequivalence $\rho$, all autoequivalences which have the same permutation action as $\rho$ are related to it by a possibly locality-violating natural isomorphism. Assuming $\rho$ satisfies the restriction $U(\psi,\psi;1)=+1$, the only (possibly) locality-violating natural isomorphisms that maintain $U(\psi,\psi;1)=+1$ have $\gamma_\psi=-1$, that is, they are related to $\Upsilon_\psi$ by a locality-respecting natural isomorphism. Therefore, if $\Upsilon_\psi$ respects locality, then locality-respecting equivalence classes in $\Aut(\mathcal{C})$ are uniquely determined by their permutation action. If $\Upsilon_\psi$ does not respect locality, then there are exactly two locality-respecting equivalence classes in $\Aut(\mathcal{C})$ for each permutation action; if $\rho$ is a representative of one such class, then $\Upsilon_\psi \circ \rho$ is a representative of the other class. In this case,
        \begin{equation}
            \Aut(\mathcal{C}) = \Aut_{LR}(\mathcal{C})/\Z_2,
        \end{equation}
        where $\Z_2$ is the subgroup of $\Aut_{LR}(\mathcal{C})$ generated by $[\Upsilon_\psi]$.
        
        We note that we have not proven that the constraints that we have found are exhaustive. Since the fermion $\psi$ is considered to be local, one could imagine a constraint of the sort $U_{\bf g}(a,\psi;a \times \psi)=+1$ for all $a$, not just $a\in \{1,\psi\}$, however we have not found any evidence that such a constraint should be required. If we did have such a more general constraint, then only one of $\rho_{\bf g}$ and $\Upsilon_\psi \circ \rho_{\bf g}$ would be allowed, in which case we would always have $\Aut_{LR}(\mathcal{C}) \simeq \Aut(\mathcal{C})$. In this case, then $[\rho]$ would again be a homomorphism into $\Aut(\mathcal{C})$.
        
        \subsubsection{Examples}        
        
        We presently explain some examples and special cases where it can be determined whether or not $\Upsilon_\psi$ respects locality. 
        
        If any minimal modular extension $\C$ of $\mathcal{C}$ contains an Abelian fermion parity vortex $v$, then $\Upsilon_\psi$ respects locality. Specifically, we can define $\zeta_a = M_{a,v}$, as $M_{a,v} \in \U$ respects the fusion rules of $\C$ and therefore also respects the fusion rules of $\mathcal{C}$. Since $v$ is a fermion parity vortex, $\zeta_\psi=-1$. Clearly this case includes all $\mathcal{C}$ of the form $\mathcal{C} = \lbrace 1,\psi \rbrace \boxtimes \mathcal{B}$ for modular $\mathcal{B}$. 
        
        We prove in Appendix~\ref{app:AbelianParityVortex} that the converse of the above statement is true as well, so that $\Upsilon_\psi$ respects locality if and only if some minimal modular extension $\C$ of $\mathcal{C}$ contains an Abelian fermion parity vortex $v$. 
        
        One physical situation where $\Upsilon_\psi$ does not respect locality is whenever $\mathcal{C}$ contains a fusion of the form $a \times b = c + (c \times \psi) + \cdots$, that is, $N_{ab}^c = N_{ab}^{c \times \psi} > 0 $. We do not know whether or not this condition is necessary for $\Upsilon_\psi$ to violate locality, but the proof that it is sufficient is straightforward. Assume such a fusion rule exists; then any phase $\zeta_a$ which obeys the fusion rules must obey
        \begin{equation}
            \zeta_c = \zeta_a \zeta_b = \zeta_{c \times \psi} = \zeta_c \zeta_\psi
        \end{equation}
        Therefore $\zeta_\psi = +1$. An example where this occurs is the theory describing the anyon content of $\SO(3)_3$ Chern-Simons theory, which we shall simply call\footnote{In some references in the condensed matter literature, e.g., \cite{Fidkowski13}, this BFC is named $\SO(3)_6$ because it is the integer spin sector of the anyons in $\mathrm{SU}(2)_6$ Chern-Simons theory.}  $\mathcal{C}=\SO(3)_3$ (see, e.g., \cite{Fidkowski13} for the explicit BFC data) which has a fusion rule $s \times s = 1+s+\tilde{s}$ with $\tilde{s}=s\times \psi$.
        
        In the case where $\mathcal{C}$ contains a fusion rule where $N_{ab}^c=N_{ab}^{c \times \psi} > 0$, then the quantity
        \begin{equation}
            \Lambda_{ab}^c = U(c,\psi;c\times\psi)U(a,b;c)U^{-1}(a,b;c\times \psi)
            \label{eqn:Lambda}
        \end{equation}
        is nonzero and gauge-invariant within $\Aut_{LR}(\mathcal{C})$. It is straightforward to check that if the autoequivalence is changed from $\rho$ to $\Upsilon_\psi \circ \rho$, then $\Lambda_{ab}^c$ changes to $-\Lambda_{ab}^c$; $\Lambda_{ab}^c$ is therefore an invariant that distinguishes two elements of $\Aut_{LR}(\mathcal{C})$ with the same permutation action. 
        
        \subsubsection{Summary}
        
        To summarize, the following are equivalent for a super-modular category $\mathcal{C}$:
        \begin{itemize}
            \item $\Upsilon_\psi$ respects locality
            \item Some minimal modular extension $\C$ of $\mathcal{C}$ contains an Abelian fermion parity vortex
            \item $\Aut(\mathcal{C}) \simeq \Aut_{LR}(\mathcal{C})$
            \item If elements of $\Aut(\mathcal{C})$ are uniquely determined by their permutation action on the anyons, then so are elements of $\Aut_{LR}(\mathcal{C})$.
            \item $K(\mathcal{C})/K_+(\mathcal{C}) \simeq \Z_2$
            \item There exists a set of phases $\zeta_a$ which obey the fusion rules and have $\zeta_\psi = -1$.
        \end{itemize}
        
        Conversely, the following are also equivalent:
        \begin{itemize}
            \item $\Upsilon_\psi$ violates locality
            \item $\Aut_{LR}(\mathcal{C})/\Z_2 \simeq \Aut(\mathcal{C})$
            \item If elements of $\Aut(\mathcal{C})$ are uniquely determined by their permutation action on the anyons, exactly two elements of $\Aut_{LR}(\mathcal{C})$ have the same permutation action on the anyons.
            \item $K(\mathcal{C}) \simeq K_+(\mathcal{C})$
            \item Any set of phases $\zeta_a$ which obey the fusion rules must have $\zeta_\psi = +1$.
        \end{itemize}
        
\section{Obstructions and classification of symmetry fractionalization}
        \label{sec:obstructions}
        
        We now consider the obstructions to fractionalizing $G_f$ on $\mathcal{C}$. There are two such obstructions. The first is an obstruction to defining any symmetry fractionalization of $G_b$ as a bosonic symmetry group on the super-modular tensor category; this obstruction $[\Omega] \in \H^3(G_b,K(\mathcal{C}))$. We will call this the ``bosonic obstruction" since it is independent of the extension $G_f$ of $G_b$. One can think of this as an obstruction to the symmetry localization ansatz of Eq. \ref{eqn:symmFracAnsatz}, while ignoring the locality of the fermion (that is, ignoring the constraints discussed in Section \ref{fermSymLocSec}). 
        
        Assuming that the bosonic obstruction vanishes, the second obstruction, which we will call the fermionic obstruction, is to finding a symmetry fractionalization pattern which obeys the constraint $\eta_\psi = \omega_2$. If $\Upsilon_\psi$ respects locality, then the fermionic obstruction $[\coho{O}_f]\in  \H^3(G_b,\A/\{1,\psi\})$. If $\Upsilon_\psi$ does not respect locality, then the fermionic obstruction is $\coho{O}_f \in Z^2(G_b,\Z_2)$.
        
        The fermionic obstruction has appeared in the math literature in~\cite{galindo}; we explicitly incorporate the locality restrictions and give a physical understanding of these obstructions.
        
        If both the bosonic and fermionic symmetry localization obstructions vanish, the symmetry localization ansatz is well-defined and is compatible with the locality of the fermion, which implies that there exists some well-defined symmetry fractionalization pattern. We can then classify distinct symmetry fractionalization patterns; we find that these patterns form a torsor over $\H^2(G_b,\A/\{1,\psi\})$, i.e., that different patterns are related to each other by an element of $\H^2(G_b,\A/\{1,\psi\})$ but there is not in general a canonical identification of symmetry fractionalization patterns with cohomology classes.
        
        Before proceeding, recall from the previous subsection that $\Upsilon_\psi$ respects locality if and only if there are phases $\zeta_a \in \U$ which respect the fusion rules and for which $\zeta_\psi = -1$, that is, if $K(\mathcal{C})/K_+(\mathcal{C})=\Z_2$. In what follows, we will only directly use the (non-)existence of such phases rather than explicitly using $\Upsilon_\psi$.
        
        \subsection{Defining the cohomology class of the bosonic obstruction}
        \label{sec:cohomologyClass}
        
        In this subsection we will show how the map
        \begin{align}
        [\rho] : G_b \rightarrow \text{Aut}_{LR}(\mathcal{C})   
        \end{align}
        determines an element $[\Omega] \in \mathcal{H}^3(G_b, K(\mathcal{C}))$. We will provide the interpretation of $[\Omega]$ as an obstruction to symmetry localization in Section~\ref{symLocObs}. 
        
        Choose a representative $\rho_{\bf g}$ of $[\rho_{\bf g}]$. Recall that the natural isomorphisms $\kappa_{\bf g,h}$ are defined by Eq.~\ref{eqn:kappaghDef} and can be decomposed
        \begin{equation}
            \kappa_{\bf g,h}(a,b;c) = \frac{\beta_a({\bf g,h})\beta_b({\bf g,h})}{\beta_c({\bf g,h})}
        \end{equation}
        for phases $\beta$, where $\kappa_{\bf g,h}(a,b;c)$ is the action of $\kappa_{\bf g,h}$ on an $\ket{a,b;c}$ fusion vertex.
        
        Demanding that the two ways to decompose $\rho_{\bf ghk}$ are consistent leads to the condition
        \begin{equation}
            \kappa_{\bf g,hk}\rho_{\bf g}\kappa_{\bf h,k}\rho_{\bf g}^{-1} = \kappa_{\bf gh,k}\kappa_{\bf g,h}
            \label{eqn:kappaConsistency}
        \end{equation}
        Define
        \begin{equation}
            \Omega_a({\bf g,h,k}) = \frac{\beta_{\,^{\overline{\bf g}}a}^{\sigma({\bf g})}({\bf h,k})\beta_a({\bf g,hk})}{\beta_a({\bf g,h})\beta_a({\bf gh,k})}
            \label{eqn:OmegaFromBeta}
        \end{equation}
        By definition $\Omega_a$ is a $U(1)$ valued $3$-cochain: $\Omega_a \in C^3(G_b,\U)$.
        Applying Eq.~\eqref{eqn:kappaConsistency} to a state $\ket{a,b;c}$, we immediately find
        \begin{equation}
            \Omega_a \Omega_b = \Omega_c
        \end{equation}
        whenever $N_{ab}^c \neq 0$ (so that $\ket{a,b;c}$ is a nonzero state). Letting $a$ vary, then, we have $\Omega \in C^3(G_b,K(\mathcal{C}))$, with a group action\footnote{Thanks to the presence of $\overline{\bf g}$ rather than ${\bf g}$, this is actually a right group action of $G_b$ on $K(\mathcal{C})$ rather than a more conventional left group action.} which takes $\Omega_a \rightarrow \Omega_{\,^{\overline{\bf g}}a}$. By direct computation,
        \begin{equation}
            \frac{\Omega_{\,^{\overline{\bf g}}a}^{\sigma({\bf g})}({\bf h,k,l})\Omega_a({\bf g,hk,l})\Omega_a({\bf g,h,k})}{\Omega_a({\bf gh,k,l})\Omega_a({\bf g,h,kl})} = 1
            \label{eqn:dOmega}
        \end{equation}
        Hence $\Omega \in Z^3(G_b,K(\mathcal{C}))$. It is straightforward to check that $\Omega$ is invariant under symmetry gauge transformations Eq.~\ref{eqn:betaGammaTransform} and is therefore independent of the choice of representative $\rho_{\bf g}$.
        
        There is additional gauge freedom in $\Omega$ which arises from the gauge freedom in $\beta$; we may redefine $\beta_a({\bf g,h}) \rightarrow \beta_a({\bf g,h}) \nu_a({\bf g,h})$ for any phases $\nu_a({\bf g,h})$ which obey the fusion rules, that is, for $\nu \in C^2(G_b,K(\mathcal{C}))$. Inserting into Eq.~\ref{eqn:OmegaFromBeta}, we find that this modifies $\Omega \rightarrow \Omega d\nu$, that is, $\Omega$ is ambiguous by an element of $B^3(G_b,K(\mathcal{C}))$. Hence $[\Omega] \in \H^3(G_b,K(\mathcal{C}))$ is a well-defined cohomology class.
        
        We can go a bit further if $\Upsilon_\psi$ violates locality. Then $K(\mathcal{C})=K_+(\mathcal{C})\simeq \A/\{1,\psi\}$; in this case, we know that $\Omega_\psi = +1$ and thus we can write
        \begin{equation}
            \Omega_a({\bf g,h,k})=M_{a,\cohosub{O}_b({\bf g,h,k})}
            \label{eqn:OmegaO}
        \end{equation}
        for some $\coho{O}_b \in C^3(G_b,\A/\{1,\psi\})$. Inserting Eq.~\eqref{eqn:OmegaO} into Eq.~\eqref{eqn:dOmega},
        \begin{align}
            1 &= M^{\sigma({\bf g})}_{\,^{\overline{\bf g}}a,\cohosub{O}_b({\bf h,k,l})}M_{a,\cohosub{O}_b({\bf g,hk,l})}M_{a,\cohosub{O}_b({\bf g,h,k})}M^{\ast}_{a,\cohosub{O}_b({\bf gh,k,l})}M^{\ast}_{a,\cohosub{O}_b({\bf g,h,kl})}\\
            &= M_{a,\,^{\bf g}\cohosub{O}_b({\bf h,k,l})}M_{a,\cohosub{O}_b({\bf g,hk,l})}M_{a,\cohosub{O}_b({\bf g,h,k})}M_{a,\cohosub{O}_b({\bf gh,k,l})}M_{a,\overline{\cohosub{O}_b({\bf g,h,kl})}}\\
            &= M_{a,\,^{\bf g}\cohosub{O}_b({\bf h,k,l}) \times \cohosub{O}_b({\bf g,hk,l}) \times \cohosub{O}_b({\bf g,h,k}) \times \overline{\cohosub{O}_b({\bf gh,k,l})} \times \overline{\cohosub{O}_b({\bf g,h,kl})}}\\
            &= M_{a,d\cohosub{O}_b({\bf g,h,k,l})}, \label{eqn:MaDOmega}
        \end{align}
        for all $a \in \mathcal{C}$. Here we have
        used the symmetry properties of the $S$-matrix and the fact that if $M_{ab}$ is always a phase, then $M_{ab}M_{ac}=M_{ad}$ whenever $N_{bc}^d \neq 0$. Because we are only considering braiding of $\O_b$ with elements of $\mathcal{C}$ and \textit{not} a modular extension, at every step in this process, we could have freely inserted a fermion into any of the $\O_b$ or into the overall fusion product. Hence $d\O_b$ is completely ambiguous by a fermion. 
        
        Super-modularity converts Eq.~\ref{eqn:MaDOmega} into
        \begin{equation}
            d\O_b({\bf g,h,k,l}) \in \{1,\psi\}
            \label{eqn:dOTriviality}
        \end{equation}
        
        Hence $d\O_b =1$ as an element of $\A/\{1,\psi\}$; that is, $\O_b \in Z^3(G_b,\A/\{1,\psi\})$. A similar calculation shows that the coboundary ambiguity in $\Omega$ leads to a coboundary ambiguity in $\coho{O}_b$. The conclusion is that if $\Upsilon_\psi$ violates locality, then $[\coho{O}_b]\in\H^3(G_b,\A/\{1,\psi\})$ is a well-defined cohomology class.
        
        If $\Upsilon_\psi$ respects locality, then in general we cannot say anything further than the above. We will show in Sec.~\ref{symLocObs} that it is possible to choose a cocycle representative $\Omega$ such that $\Omega \in Z^3(G_b,K_+(\mathcal{C}))$, but as we will see, the class $[\Omega] \in \H^3(G_b,K_+(\mathcal{C}))$ can in general be nontrivial even if there is no symmetry localization obstruction; the actual obstruction is the class $[\Omega] \in \H^3(G_b,K(\mathcal{C}))$.
        
        \subsection{Symmetry localization obstructions}
        \label{symLocObs}
        
        Now we determine the obstructions to symmetry localization on $\mathcal{C}$.
        
        As reviewed in Section \ref{sec:bosonicSymmFrac}, localizing the symmetry on $\mathcal{C}$ amounts to choosing a set of phases $\omega_a({\bf g,h})$ which obey the fusion rules and 
        \begin{equation}
            \Omega_a({\bf g,h,k})=\omega_{\,^{\overline{\bf g}}a}({\bf h,k})\omega_a({\bf gh,k})^{-1}\omega_a({\bf g,hk})\omega_a({\bf g,h})^{-1}.
            \label{eqn:OmegaEqualsDomega}
        \end{equation}
        
        Recall that the symmetry fractionalization data $\eta$ is defined using Eq.~\ref{eqn:omegaInTermsOfBetaEta}, and that fermionic symmetry fractionalization means we require $\eta_\psi({\bf g,h})=\omega_2({\bf g,h})$.

        The bosonic symmetry localization obstruction is the obstruction to finding any solution of Eq.~\ref{eqn:OmegaEqualsDomega}, which we may reinterpret as the condition
        \begin{equation}
            \Omega = d\omega
        \end{equation}
        for $\Omega \in Z^3(G_b,K(\mathcal{C}))$ and $\omega \in C^2(G_b,K(\mathcal{C}))$. That is, $[\Omega]\in \H^3(G_b,K(\mathcal{C}))$ is the bosonic symmetry localization obstruction. To characterize this obstruction further and to understand the fermionic symmetry localization obstruction, we proceed in two cases, depending on the locality of $\Upsilon_\psi$.
        
        \subsubsection{Case: \texorpdfstring{$\Upsilon_\psi$}{Ypsi} does not respect locality}
    
        If $\Upsilon_\psi$ does not respect locality, we have $K(\mathcal{C})=K_+(\mathcal{C})$ and, accordingly,
        \begin{equation}
            \Omega_\psi=+1.
            \label{eqn:omegaPsi1}
        \end{equation} If a solution $\omega_a$ to Eq.~\ref{eqn:OmegaEqualsDomega} exists, then we must have $\omega_\psi = +1$ as well because $\omega_a$ obeys the fusion rules. Hence we may write
        \begin{equation}
            \omega_a({\bf g,h}) = M_{a,\mathfrak{w}({\bf g,h})}
        \end{equation}
        for all $a \in \mathcal{C}$ for some $\mathfrak{w} \in \A$, where again $\mathfrak{w}$ is ambiguous by a fermion. Substituting into Eq.~\eqref{eqn:OmegaEqualsDomega}, we find the usual requirement
        \begin{equation}
            M_{a,\cohosub{O}_b}=M_{a,d\cohosub{w}},
        \end{equation}
        Hence, $\O_b=d\coho{w}$, modulo a fermion, that is, we must have $[\O_b] = 0 \in \mathcal{H}^3(G_b,\A/\{1,\psi\})$ in order to have symmetry fractionalization.
        If $[\coho{O}_b]=0$, then by definition there exists such a $\coho{w}$, so there is no additional bosonic obstruction.
        
        To understand the fermionic symmetry localization obstruction, we must attempt to enforce the condition $\eta_\psi = \omega_2$. Recall from Sec.~\ref{sec:cohomologyClass} that $\beta_\psi \in \Z_2$ is gauge-invariant when $\Upsilon_\psi$ does not respect locality. Suppose that some solution $\omega_a$ of Eq.~\ref{eqn:OmegaEqualsDomega} exists, that is, the bosonic obstruction vanishes; it automatically has $\omega_\psi = +1$ as mentioned above. Then using Eqs.~\ref{eqn:omegaInTermsOfBetaEta},~\ref{eqn:omegaPsi1}, we have $\eta_\psi = \beta_\psi/\omega_\psi = \beta_\psi$. Hence, since $\eta_\psi \in Z^2(G_b,\Z_2)$,
        \begin{equation}
        \coho{O}_f = \beta_\psi/\omega_2 \in Z^2(G_b,\Z_2)
        \label{eqn:fermionicObstructionUpsilonViolates}
        \end{equation} 
        is the obstruction to imposing the fermionic symmetry fractionalization condition $\eta_\psi = \omega_2$. If $\beta_\psi/\omega_2=+1$, then we automatically have $\eta_\psi = \omega_2$ and the symmetry fractionalization pattern accounts correctly for the fermionic symmetry.
        
        Note also that, by definition, $\beta_\psi({\bf g,h})=-1$ implies $[\kappa_{\bf g,h}] = [\Upsilon_\psi]$. Therefore, if the fermionic obstruction vanishes,
        \begin{equation}
            [\kappa_{\bf g,h}] = [\Upsilon_\psi]^{(1-\omega_2({\bf g,h}))/2}.
            \label{eqn:kappaUpsilonPsi}
        \end{equation}
        We see that when $\Upsilon_\psi$ violates locality, $[\rho_{\bf g}]:G_b \rightarrow \Aut_{LR}(\mathcal{C})$ can be a group homomorphism without a symmetry fractionalization obstruction only if $\omega_2({\bf g,h}) = +1$.\footnote{The fact that Eq. \ref{eqn:fermionicObstructionUpsilonViolates} implies a constraint relating $\omega_2$ and $[\kappa_{\bf g,h}]$ in the case where $\Upsilon_\psi$ is locality-violating, as summarized in Eq. \ref{eqn:kappaUpsilonPsi}, was also noted previously in \cite{aasen21ferm}.
        Ref.~\cite{aasen21ferm} further observed that a somewhat looser version of Eq.~\ref{eqn:kappaUpsilonPsi} can also be derived when $[\kappa_{\bf g,h}] = [\Upsilon_\psi]$ for some ${\bf g, h}$ by demanding that there exists an (unconstrained) lift of the map $[\rho_{\bf g}]:G_b \rightarrow \Aut_{LR}(\mathcal{C})$ to a group homomorphism $G_f \rightarrow \Aut_{LR}(\mathcal{C})$. We do not enforce such a requirement explicitly, although as stated above the vanishing of $\coho{O}_f$ implies the existence of such a lift with $(-1)^F$ mapping to $[\Upsilon_\psi]$.}
        
        Furthermore, when this fermionic obstruction vanishes, we see that $[\rho_{\bf g}]$ lifts to a group homomorphism $G_f \rightarrow \Aut_{LR}(\mathcal{C})$ such that $(-1)^F$ maps to $[\Upsilon_\psi]$. Therefore, when $\Upsilon_\psi$ violates locality, we can view $\coho{O}_f$ as obstructing the existence of such a lift.
        
        The above fermionic obstruction was also found in the mathematical context of categorical fermionic actions in Ref.~\cite{galindo}, where it was considered to be an element of $\H^2(G_b,\Z_2)$. The question of whether to mod out by 2-coboundaries arises upon consideration of what data is considered given. If one is given only $[\omega_2] \in \H^2(G_b,\Z_2)$, then there is freedom to simply choose a different $\omega_2$ in the cohomology class and we should consider $[\O_f] \in \H^2(G_b,\Z_2)$. As discussed in Sec.~\ref{sec:fermionicSymm}, this choice amounts to a different decomposition of $G_f$ into $G_b\times \Z_2^f$ as sets. Furthermore, if one is only given a map $[\rho_{\bf g}]:G_b \rightarrow \Aut(\mathcal{C})$, then there is freedom to modify $\rho_{\bf g}$ by $\Upsilon_\psi$, which modifies $\O_f$ by a coboundary. In this case, again the obstruction is $[\O_f] \in \H^2(G_b,\Z_2)$. However, as we discussed in Sec.~\ref{sec:fermionicSymm}, a complete specification of a quantum many-body system and its symmetries fixes both a cochain representative of $\omega_2$ and determines $[\rho_{\bf g}]:G \rightarrow \Aut_{LR}(\mathcal{C})$. In this case $\omega_2$ and $\beta_\psi$ have no further gauge freedom. Thus there is no further freedom to change $\O_f$ by a coboundary, and so we take the obstruction to be valued in $Z^2(G_b, \Z_2)$. If one instead were interested in whether there is \textit{any} physical realization of an abstract symmetry group $G_f$ acting on $\mathcal{C}$, the obstruction would be a cohomology-level obstruction.
        
        An example where this fermionic obstruction occurs is in $\SO(3)_3$ Chern-Simons theory with $G_f = \Z_2^{\bf T}\times \Z_2^f$. There is a unique permutation action of the anyons under time-reversal. One can check directly that there is only a valid symmetry fractionalization pattern with $\beta_\psi = \eta_\psi({\bf T,T})=-1$, which would require $G_f = \Z_4^{{\bf T}, f}$. We explain this in more detail in Sec.~\ref{sec:SO33}.
        
        Remarkably, when $\Upsilon_\psi$ violates locality, given a $G_b$ and $[\rho_{\bf g}]$, there is always at most one group extension $G_f$ of $G_b$ which can be unobstructed. This follows from the fact that, given a $[\rho_{\bf g}]$, $\beta_\psi$ is gauge-invariant; the only gauge freedom in $\beta_a$ is to modify $\beta_a \rightarrow \beta_a \nu_a$ where $\nu_a$ obeys the fusion rules, and since $\Upsilon_\psi$ violates locality, $\nu_\psi = +1$. Hence the only possible unobstructed $G_f$ has $\omega_2 = \beta_\psi$.
        
        \subsubsection{Case: \texorpdfstring{$\Upsilon_\psi$}{Ypsi} respects locality}
        
        If $\Upsilon_\psi$ respects locality, then $K(\mathcal{C})/K_+(\mathcal{C})=\Z_2$, and we cannot generically say anything further about the bosonic symmetry localization obstruction; it is simply $[\Omega] \in \H^3(G_b,K(\mathcal{C}))$.
        
        Suppose the bosonic obstruction vanishes so that there exists some consistent $G_b$ symmetry fractionalization data $\eta_a$; we need to enforce the fermionic constraint $\eta_\psi=\omega_2$. As in the bosonic case, all consistent (bosonic) symmetry fractionalization patterns can be obtained from a given pattern $\eta^{(0)}_a$ via
        \begin{equation}
            \eta_a({\bf g,h})=\eta^{(0)}_a({\bf g,h})\tau_a({\bf g,h})
        \end{equation}
        where $\tau_a$ obeys the fusion rules and additionally obeys
        \begin{equation}
            \tau_{\,^{\overline{\bf g}}a}({\bf h,k})\tau_a({\bf g,hk})=\tau_a({\bf g,h})\tau_a({\bf gh,k}). \label{eqn:tauCocycle}
        \end{equation}
        
        Given the symmetry fractionalization pattern $\eta_a^{(0)}$, we may attempt to obtain one which obeys the fermionic constraint $\eta_\psi = \omega_2$ by simply choosing $\tau_a({\bf g,h})$ to be any phase which obeys the fusion rules and
        \begin{equation}
            \frac{\omega_2({\bf g,h})}{\eta_\psi^{(0)}({\bf g,h})} = \tau_\psi({\bf g,h}).\label{eqn:tauPsiConstraint}
        \end{equation}
        Such a $\tau_a$ will always exist, but it may not obey Eq.~\ref{eqn:tauCocycle}. Define
        \begin{equation}
            T_a({\bf g,h,k}) = \tau_{\,^{\overline{\bf g}}a}({\bf h,k})\tau_a({\bf g,hk})\tau_a^{-1}({\bf g,h})\tau_a^{-1}({\bf gh,k})
            \label{eqn:TaDefinition}            
        \end{equation}
        Obviously Eq.~\ref{eqn:tauCocycle} is equivalent to $T_a = 1$, and it is also clear that $T_a$ is a phase which obeys the fusion rules (since the same holds for $\tau_a$). Furthermore, $\omega_2$ and $\eta_\psi^{(0)}$ are both elements of $Z^2(G_b,\Z_2)$; hence $T_\psi = 1$, that is, $T \in K_+(\mathcal{C})$. We therefore conclude that
        \begin{equation}
            T_a({\bf g,h,k})=M_{a,\cohosub{O}_f({\bf g,h,k})}
        \end{equation}
        for some $\O_f \in C^3(G_b,\A/\{1,\psi\})$. We find with a direct computation that $dT_a = 1$, which implies $\O_f \in Z^3(G_b,\A/\{1,\psi\})$.
        
        Our desired condition $T_a = 1$ would force $\O_f \in \{1,\psi\}$; this will not be satisfied in general. However, we could have chosen another $\tau_a$ which satisfies Eq.~\ref{eqn:tauPsiConstraint}; clearly all such $\tau_a$ are obtained by modifying $\tau_a({\bf g, h})\rightarrow \tau_a({\bf g,h}) \lambda_a({\bf g,h})$ where $\lambda_a$ obeys the fusion rules and, crucially, $\lambda_\psi({\bf g,h})=+1$. Hence $\lambda \in C^2(G_b,K_+(\mathcal{C}))$. Such a $\lambda_a$ must be of the form
        \begin{equation}
            \lambda_a({\bf g,h}) = M_{a,\cohosub{v}({\bf g,h})}
        \end{equation}
        with $\coho{v}({\bf g,h})\in C^2(G_b,\A/\{1,\psi\})$. This change to $\tau_a$ modifies $\O_f \rightarrow \O_f \times d\coho{v}$. Therefore, as long as $[\O_f] = 0 \in \H^3(G_b,\A/\{1,\psi\})$, there exists some choice of $\coho{v}$ which will trivialize $T_a$, that is, produce the desired $\tau_a$. Hence $[\O_f] \in \H^3(G_b,\A/\{1,\psi\})$ is the obstruction to fractionalizing $G_f$ on $\mathcal{C}$.
        
        To summarize, if $\Upsilon_\psi$ respects locality, the bosonic obstruction is $[\Omega] \in \H^3(G_b,K(\mathcal{C}))$, while the fermionic obstruction is $[\coho{O}_f]\in H^3(G_b,\A/\{1,\psi\})$. We give an example of a theory with a trivial bosonic symmetry localization obstruction but a nontrivial fermionic symmetry localization obstruction in Sec.~\ref{sec:Sp22}, namely $\mathcal{C} = \mathrm{Sp}(2)_2 \times \{1,\psi\}$ with a particular action of $G_f = \Z_2^{\bf T} \times \Z_2^f$.
        
        \subsubsection{Technical aside on gauge-fixing when \texorpdfstring{$\Upsilon_\psi$}{Ypsi} respects locality}
        
        Suppose that $\Upsilon_\psi$ respects locality and we are given a particular representation $\beta_a({\bf g,h})$ of $\kappa_{\bf g,h}$ as a natural isomorphism. Since $\Upsilon_\psi$ respects locality, $\Aut(\mathcal{C}) = \Aut_{LR}(\mathcal{C})$ and so there exists a gauge transformation $\nu \in C^2(G_b,K(\mathcal{C}))$ such that the gauge-transformed $\beta_a$ obeys $\beta_\psi = +1$. In this gauge, Eq.~\ref{eqn:OmegaFromBeta} immediately implies $\Omega_\psi = +1$, that is, $\Omega \in Z^3(G_b,K_+(\mathcal{C}))$. Certainly if $[\Omega] =0 \in \H^3(G_b,K_+(\mathcal{C}))$ then it is also true that $[\Omega]=0\in \H^3(G_b,K(\mathcal{C}))$. However, it may be that $[\Omega] \neq 0 \in \H^3(G_b,K_+(\mathcal{C}))$ but $[\Omega]=0 \in \H^3(G_b,K(\mathcal{C}))$, so the element $[\Omega] \in \H^3(G_b,K_+(\mathcal{C}))$ is \textit{not} the bosonic symmetry localization obstruction. That is, it may be that despite the gauge-fixing $\Omega_\psi = +1$, any solution of $\Omega=d\omega$ for the particular representative $\Omega$ will necessarily have some $\omega_\psi({\bf g,h})=-1$.
        
        We can rephrase the above more precisely. If $\Upsilon_\psi$ respects locality, then there is a short exact sequence
        \begin{equation}
            1 \rightarrow K_+(\mathcal{C}) \stackrel{i}{\rightarrow} K(\mathcal{C}) \stackrel{r_\psi}{\rightarrow} \Z_2 \rightarrow 1 
        \end{equation}
        where $i$ is the inclusion map and $r_\psi$ is the restriction of a set of phases $\zeta_a$ to $a=\psi$. This leads to a long exact sequence in cohomology, where the relevant piece is
        \begin{equation}
            \cdots \rightarrow \H^2(G_b,\Z_2) \stackrel{\delta}{\rightarrow} \H^3(G_b,K_+(\mathcal{C})) \stackrel{i^\ast}{\rightarrow} \H^3(G_b,K(\mathcal{C})) \stackrel{r_\psi^\ast}{\rightarrow} \H^3(G_b,\Z_2) \rightarrow \cdots
        \end{equation}
        where $\delta$ is the connecting homomorphism. The locality constraint Eq.~\ref{eqn:UpsiEquals1} implies $\beta_\psi \in \Z_2$, in particular that $\beta_\psi \times \beta_\psi = \beta_{\psi \times \psi} = \beta_1 = +1$. Inserting into the definition Eq.~\ref{eqn:OmegaFromBeta} implies $\Omega_\psi = d\beta_\psi$, that is, $r_\psi^\ast([\Omega])=1$. Hence $[\Omega] \in \ker r_\psi^\ast = \mathrm{im} \text{ } i^{\ast}$. A particular choice of the gauge-fixing procedure above amounts to a choice of a particular element in $(i^{\ast})^{-1}([\Omega])$. If $\ker i^\ast$ is trivial, then we can safely conclude that $(i^{\ast})^{-1}([\Omega]) \in \H^3(G_b,K_+(\mathcal{C})) \sim \H^3(G_b,\A/\{1,\psi\})$ is uniquely defined and can therefore also be regarded as the bosonic obstruction. However, $i^{\ast}$ need not be injective; in fact $\ker i^\ast = \mathrm{im }\text{ } \delta$. If $i^\ast$ is not injective, then even if $[\Omega]$ is trivial our gauge-fixing procedure may set $(i^{\ast})^{-1}([\Omega])$ to a nontrivial cohomology class in $\H^3(G_b,K_+(\mathcal{C}))$. Hence we must use $[\Omega] \in \H^3(G_b,K(\mathcal{C}))$ as the bosonic obstruction.
        
        \subsection{Fermionic obstruction as an obstruction to a lift}
        
        There is another perspective on the fermionic symmetry localization obstruction which is more general and is conceptually closely related to the viewpoint of Ref.~\cite{galindo}. 
        
        Consider the following general problem: let $\mathcal{B}$ be any UBFC with some subcategory $\mathcal{B}'\subset \mathcal{B}$. Given a homomorphism $[\rho_{\bf g}]: G \rightarrow \Aut(\mathcal{B})$ (or $\Aut_{LR}(\mathcal{B})$, as appropriate) which preserves $\mathcal{B}'$, we can restrict $[\rho_{\bf g}]$ to $\mathcal{B}'$. Suppose that there exists a symmetry fractionalization pattern for this restricted homomorphism. Then what is the obstruction to lifting the symmetry fractionalization pattern to all of $\mathcal{B}$, i.e. defining a symmetry fractionalization pattern on $\mathcal{B}$ which restricts to the given one on $\mathcal{B}'$?
        
        This problem was considered in Ref.~\cite{fidkowski2018}, wherein it was explained that taking $\mathcal{B}'=\{1\}$ and $\mathcal{B}=\mathcal{C}$ produces the usual bosonic symmetry localization obstruction, while taking $\mathcal{B}'=\mathcal{C}$ and $\mathcal{B}=\C$ produces the $\H^3(G_b,\Z_2)$ anomaly associated to fermionic SETs. Ref.~\cite{fidkowski2018} showed that if $\mathcal{B}$ is modular, then the obstruction is valued in $\H^3(G,T)$ where $T \subset \mathcal{B}$ are the Abelian anyons which braid trivially with all of $\mathcal{B}'$. Those results can be generalized straightforwardly to non-modular $\mathcal{B}$; the primary change is that $T$ is replaced by $T/\mathcal{E}$, where $\mathcal{E}$ is the set of anyons which are transparent to all of $\mathcal{B}$. However, the derivation of Ref.~\cite{fidkowski2018} contains an assumption, which does not hold in general, that one must be able to write
        \begin{equation}
            \omega_{b'}({\bf g,h})=M_{b',\cohosub{w}({\bf g,h})},
            \label{eqn:fidkowskiAssumption}
        \end{equation}
        for all $b' \in \mathcal{B}'$, with $\coho{w}\in \mathcal{A}$ where $\mathcal{A}$ is the set of Abelian anyons in $\mathcal{B}$. 
        
        Our fermionic symmetry localization obstruction is a special case of the above; we are specifying symmetry fractionalization on $\mathcal{B}' = \{1,\psi\}$ with $\eta_\psi = \omega_2$ and asking if this symmetry fractionalization can be lifted to  $\mathcal{B}=\mathcal{C}$. The assumption Eq.~\ref{eqn:fidkowskiAssumption} fails if $\omega_\psi({\bf g,h})$ is not uniformly $+1$. As derived in Sec.~\ref{sec:LRNatIso}, the assumption that $\Upsilon_\psi$ respects locality actually means that we can choose a gauge where the assumption $\omega_\psi = +1$ holds. On the other hand, if $\Upsilon_\psi$ violates locality, there is no such gauge-fixing allowed, so the assumption is violated in general, and accordingly the obstruction is valued in a completely different cocycle (or cohomology) group.
        
        \subsection{Classification of symmetry fractionalization}
        \label{subsubsec:H2Torsor}
        
        Suppose we have two valid patterns of symmetry fractionalization given by $\omega_a(\bf{g,h})$ and $\omega_a'({\bf g,h})$. Then we can define
        \begin{equation}
           \tau_a({\bf g,h}) = \omega_a'({\bf g,h})\omega_a({\bf g,h})^{-1}
        \end{equation}
        Since $\omega_\psi = \omega_2$, we must have
        \begin{equation}
            \tau_\psi({\bf g,h})= +1
        \end{equation}
        and $\tau_a$ must obey the fusion rules. Hence we can write
        \begin{equation}
            \tau_a = M_{a,\cohosub{t}({\bf g,h})}
        \end{equation}
        where $\coho{t} \in C^2(G_b,\A/\{1,\psi\})$. Using the fact that $\omega$ and $\omega'$ both obey Eq.~\ref{eqn:OmegaEqualsDomega}, it is straightforward to check that $d\coho{t} \in \{1,\psi\}$, and therefore $\coho{t} \in Z^2(G_b,\A/\{1,\psi\})$.
  
        Note that, as discussed above, if $\Upsilon_\psi$ violates locality, then we can write 
        \begin{align}
            \omega_a({\bf g},{\bf h}) = M_{a,\cohosub{w}({\bf g},{\bf h})} 
        \end{align}
        for $\coho{w} \in \mathcal{A}/\{1,\psi\}$. If instead $\Upsilon_\psi$ respects locality, then there exists some minimal modular extension $\widecheck{C}$ such that the above equation continues to hold, with $\coho{w}({\bf g},{\bf h}) \in \widecheck{\mathcal{A}}/\{1,\psi\}$, where $\widecheck{\mathcal{A}}$ consists of the Abelian anyons of $\widecheck{C}$. 
        
        As in the bosonic case, there is gauge freedom; we may redefine the local operators $U_{\bf g}^{(i)}$ by a local unitary operator $Z_{\bf g}^{(i)}$ such that
        \begin{equation}
            \prod_{j=1}^n Z_{\bf g}^{(j)} = \mathbbm{1}
            \label{eqn:ZgjEquals1}
        \end{equation}
        on an $n$-quasiparticle state, \textit{provided that} the constraint Eq.~\ref{eqn:localTopologicalFermionConsistency} is maintained. This constraint forces $Z_{\bf g}^{(i)}$ to act trivially on states with topological charge $\psi$ in region $i$. As in the bosonic case~\cite{barkeshli2019}, Eq.~\ref{eqn:ZgjEquals1} means that each $Z_{\bf g}^{(j)}$ can only modify a given state by a phase since the $Z_{\bf g}^{(j)}$ are local and act on well-separated regions of space. The $Z_{\bf g}^{(j)}$ are local operators, so this phase can only depend on the anyon $a_j$ at position $j$, ${\bf g}$, or other local degrees of freedom in $a_j$. Demanding that the action of $Z_{\bf g}^{(j)}$ is a phase when acting on an arbitrary superposition of states in the same superselection sector in fact forces $Z_{\bf g}^{(j)}$ to be independent of local degrees of freedom, that is,
        \begin{equation}
            Z_{\bf g}^{(j)} = \zeta_{a_j}({\bf g})
        \end{equation}
        where the above equation is interpreted to be acting on a state with topological charge $a_j$ in region $j$, $\zeta_{a_j}({\bf g}) \in \U$, $\zeta_{\psi}({\bf g})=1$, and
        \begin{equation}
            \prod_{j=1}^n \zeta_{a_j}({\bf g}) =1.
        \end{equation}
        The above equation implies that $\zeta({\bf g})$ obeys the fusion rules for $\mathcal{C}$, so
        \begin{equation}
            \zeta_{a}({\bf g})= M_{a,\cohosub{z}({\bf g})},
        \end{equation}
        where, in order to maintain $\zeta_\psi = +1$, we have $\coho{z} \in \mathcal{C}$ and thus $\coho{z}\in \A$. Again, $\coho{z}$ is ambiguous by a fermion and is thus valued in $\A/\{1,\psi\}$. Under this transformation,
        \begin{equation}
            \omega_a({\bf g,h}) \rightarrow \omega_a({\bf g,h})\frac{\zeta_{\,^{\overline{{\bf g}}}a}({\bf h})\zeta_a({\bf h})}{\zeta_a({\bf gh})}
        \end{equation}
        which corresponds to transforming $\mathfrak{w}$ by $\coho{t}=d\coho{z}$. Hence $\coho{t}$ related by coboundaries are gauge-equivalent, that is, $\coho{t}\in \mathcal{H}^2(G_b,\A/\{1,\psi\})$.
        
        This means that symmetry fractionalization classes form an $\mathcal{H}^2(G_b,\A/\{1,\psi\})$ torsor. In particular, distinct symmetry fractionalization classes with data $\eta_a$ and $\eta_a'$ are related by
        \begin{equation}
            \eta'_a({\bf g,h})= \eta_a({\bf g,h})M_{a,\cohosub{t}({\bf g,h})}
        \end{equation}
        for cohomologically nontrivial $[\coho{t}]\in\H^2(G_b,\A/\{1,\psi\})$.
        As in the bosonic case, there is not generally a canonical ``trivial" symmetry fractionalization class.
        
        As in the bosonic case, changing the symmetry fractionalization class by an element of $\mathcal{H}^2(G_b,\A/\{1,\psi\})$ may not yield a physically distinct symmetry fractionalization class. This is because two different sets of symmetry fractionalization data may be related to each other by a relabeling of the anyons. More specifically, permuting the anyon labels with a permutation $p$ corresponding to some \textit{unitary} braided autoequivalence of $\mathcal{C}$ will yield a physically equivalent fractionalization class if $p$ commutes with the permutation action on the anyon labels of every $\rho_{\bf g}$.
        
        \section{Examples}
            \label{sec:examples}
        
            \subsection{\texorpdfstring{$G_b=\Z_2^{\bf T}$}{Gb = Z2T} and fermionic Kramers degeneracy}
            
            There are two possible group extensions of $G_b = \Z_2^{\bf T}$; the trivial extension $\eta_{\psi}({\bf T,T})=1$ and the nontrivial one $\eta_{\psi}({\bf T,T})=-1$. These correspond to ${\bf T}^2=1$ and ${\bf T}^2=(-1)^{F}$, respectively.
            
            Consider any $a \in \mathcal{C}$ such that $\,^{\bf T}a = a$. Then $\eta_a^{\bf T} \equiv \eta_a({\bf T,T})\in \Z_2$ is gauge-invariant, just as in the case where the symmetry and topological order is purely bosonic. If $\eta_a^{\bf T} = -1$, then $a$ carries Kramers degeneracy~\cite{barkeshli2019}.
            
            A more interesting possibility occurs when $^{\bf T}a = a \times \psi$. Then it is easy to check that
            \begin{equation}
                \eta_a^{\bf T}\equiv \eta_{a}({\bf T,T})U_{\bf T}(a,\psi;a \times \psi) F^{a,\psi,\psi}
                \label{eqn:fermionicKramersEta}
             \end{equation}
             is gauge-invariant as well and can be interpreted roughly as the action of ${\bf T}^2$ on $a$. The gauge-invariance of this quantity requires that the $\Gamma^{\psi,\psi}_1$ vertex basis transformation is disallowed, as discussed in Sec.~\ref{sec:supermodular}. The quantity $F^{a,\psi,\psi}$ can be canonically fixed to 1 (again see  Sec.~\ref{sec:supermodular}), so we omit it in the future. We comment on the transformation rules for $U$ under vertex gauge transformations in Appendix \ref{app:Utransform}. Using the symmetry fractionalization consistency conditions and the fact that the pentagon equation forces $F^{a,\psi,\psi}=F^{a\psi,\psi,\psi}$, it is straightforward to show that
             \begin{equation}
                 (\eta_a^{\bf T})^2 = \eta_{\psi}({\bf T,T})
             \end{equation}
             That is, if $G_f = \Z_2^{\bf T} \times \Z_2^f$, then $\eta_a^{\bf T}$ is a sign and determines whether or not $a$ carries Kramers degeneracy, but if $G_f = \Z_4^{T,f}$, then $\eta_a^{\bf T} = \pm i$. The latter is the precise, gauge-invariant sense in which we can have ``${\bf T}^2 = \pm i$" on an anyon, as discussed, e.g., for $\mathcal{C}=\SO(3)_3$ in Ref.~\cite{metlitski2014}.
             
             \subsubsection{Dimensional reduction to 1+1D SPTs}
             
             In order to interpret the invariants above, we review~\cite{barkeshli2019tr} the dimensional reduction from an anyon with ${}^{\bf T}a = a$ to a $(1+1)$D $\Z_2^{\bf T}$ symmetry-protected topological state (SPT), then turn to the fermionic case.
             
             Consider a cylinder with a time reversal-invariant anyon $a$ on its left end, $\overline{a}$ on the right end, and vacuum in between, fusing to the identity channel. Then in a ground state $\ket{\Psi}$, we have
             \begin{equation}
                 R_{\bf T}\ket{\Psi} = U^{(L)}_{\bf T}U^{(R)}_{\bf T}U_{\bf T}(a,\overline{a};1)\ket{\Psi}
             \end{equation}
             in the usual symmetry fractionalization ansatz, where $U^{(L,R)}_{\bf T}$ are local unitary operators. Now, if our system is bosonic, we have
             \begin{align}
                 R_{\bf T}^2\ket{\Psi} = \ket{\Psi} &= {}^{\bf T}U^{(L)}_{\bf T}{}^{\bf T}U^{(R)}_{\bf T}U^{\ast}_{\bf T}(a,\overline{a};1)R_{\bf T}\ket{\Psi}\\
                 &={}^{\bf T}U^{(L)}_{\bf T}{}^{\bf T}U^{(R)}_{\bf T}U^{(L)}_{\bf T}U^{(R)}_{\bf T}\ket{\Psi}\\
                 &={}^{\bf T}U^{(L)}_{\bf T}U^{(L)}_{\bf T}{}^{\bf T}U^{(R)}_{\bf T}U^{(R)}_{\bf T}\ket{\Psi}\\
                 &= \eta_a^{\bf T}\eta_{\overline{a}}^{\bf T}\ket{\Psi}
             \end{align}
             Hence the local action of ${\bf T}^2$, i.e. $\,^{\bf T}U_{\bf T}^{(L/R)}U_{\bf T}^{(L/R)}$, on each anyon is given by $\eta_a^{\bf T}$, that is, this quantity characterizes whether each end of the dimensionally reduced cylinder carries a linear or projective representation of $\Z_2^{\bf T}$, subject to the constraint $\eta_a^{\bf T}\eta_{\bar a}^{\bf T}=+1$ that the global representation is linear. From the consistency conditions for $\eta$ it is simple to show that $\eta_a^{\bf T} = \pm 1$. Therefore, $\eta_a^{\bf T}=-1$ means that an endpoint of the cylinder, or equivalently $a$, carries Kramers degeneracy. In this case, the dimensionally reduced system has a Kramers pair on each end, that is, it is a nontrivial $(1+1)$D $\mathbb{Z}_2^{\bf T}$ SPT.
             
             In the presence of fermions, we can ask whether the bosonic $(1+1)$D $\Z_2^{\bf T}$ SPT is trivial or nontrivial in the fermionic classification. Following Ref.~\cite{WangGu}, one can check that in class BDI, the dimensionally reduced system is in the $\nu = 4$ class of the $\Z_8$ fermion SPT classification. In class DIII, the dimensionally reduced system is a trivial SPT. This latter case is straightforward to understand physically; although there is no local bosonic operator that removes the Kramers pair associated with the anyon $a$ at the end of the system, in class DIII the local fermion $\psi$ carries Kramers degeneracy. Hence there is a fermionic operator which trivializes the end of the system.
             
             Now let us run a similar argument for an anyon $a$ with ${}^{\bf T}a = a\times \psi$. We again place $a$ and $\overline{a}$, fusing to the identity, on the ends of a cylinder. This time
             
            \begin{equation}
                 R_{\bf T}\ket{\Psi} = U^{(L)}_{\bf T}U^{(R)}_{\bf T}U_{\bf T}(a\times \psi ,\overline{a} \times \psi;1)\ket{\Psi}
             \end{equation}
             Proceeding as before,
             
             \begin{align}
                R_{\bf T}^2\ket{\Psi} = \ket{\Psi} &= \,^{\bf T}U^{(L)}_{\bf T}\,^{\bf T}U^{(R)}_{\bf T}U^{\ast}_{\bf T}(a\times \psi,\overline{a} \times \psi;1)R_{\bf T}\ket{\Psi}\\
                    &=\,^{\bf T}U^{(L)}_{\bf T}\,^{\bf T}U^{(R)}_{\bf T}U^{\ast}_{\bf T}(a\times \psi,\overline{a} \times \psi;1)U^{(L)}_{\bf T}U^{(R)}_{\bf T}U_{\bf T}(a,\overline{a} ;1)\ket{\Psi}
            \end{align}
            
            Using the consistency conditions, it is not hard to show that
            \begin{equation}
                U^{\ast}_{\bf T}(a\times \psi,\overline{a} \times \psi;1)U_{\bf T}(a,\overline{a} ;1) = -U_{\bf T}(a,\psi; a \times \psi)U_{\bf T}(\overline{a},\psi; \overline{a} \times \psi)
            \end{equation}
            where the minus sign is crucial and comes from the presence of fermions. Hence,
            \begin{align}
                R_{\bf T}^2\ket{\Psi} = \ket{\Psi} = -\eta_a^{\bf T}\eta_{\overline{a}}^{\bf T}\ket{\Psi}
            \end{align}
            Hence $\eta_a^{\bf T}\eta_{\overline{a}}^{\bf T}=-1$.
            
            The same argument, mutatis mutandis, on a state $\ket{\Psi'}$ with $a$ on one end of the cylinder and $\overline{a}\times \psi$ on the other end, in the $\psi$ fusion channel, shows that
            \begin{equation}
                \eta_a^{\bf T}(\eta_{\overline{a}}^{\bf T})^{\ast}=-\eta_{\psi}({\bf T,T})
            \end{equation}
            where we need to use $R_{\bf T}^2\ket{\Psi'} = \eta_{\psi}({\bf T,T})\ket{\Psi'}$ by our choice of group extension to $G_f$. This can be used to conclude that $(\eta_a^{\bf T})^2 = \eta_{\psi}({\bf T,T})$ as expected.
            
            According to the above argument, the local action of ${\bf T}^2$ on each end of the dimensionally reduced system is given by $\eta_a^{\bf T}$, possibly up to local fermion parity. Therefore, $\eta_a^{\bf T}$ diagnoses the SPT phase of the dimensionally reduced system. This action for anyons has been discussed as ``fermionic Kramers parity" in a rather different language in Ref.~\cite{metlitski2014}. Ref.~\cite{metlitski2014} also shows that a DIII $(1+1)$D SPT should have local action $\pm i (-1)^{F}$ at its ends. In our formalism, DIII corresponds to the nontrivial group extension under which anyons can carry $\eta_a^{\bf T} = \pm i$, so such anyons lead to a nontrivial DIII SPT upon dimensional reduction. Similarly, in the $\nu=2$ class of BDI, one end of the SPT should have local action $+1$ and one should have local action $-1$. This class corresponds in our formalism to the trivial group extension, under which anyons carry $\eta_a^{\bf T}=\pm 1$. 
            
            For time-reversal operations we are considering here which change the local fermion parity, this local ${\bf T}^2$ eigenvalue is not quite multiplicative under fusion of anyons or, in the dimensionally reduced picture, layering of SPTs; there is an extra minus sign. Specifically, if $N_{a,b}^c>0$, $\,^{\bf T}a = a \psi$, $\,^{\bf T}b = b\psi$, and $\,^{\bf T}c=c$, then
            \begin{equation}
                \eta_a^{\bf T}\eta_b^{\bf T}=-\eta_c^{\bf T},
                \label{eqn:multiplicationOfKramers}
            \end{equation}
            independent of the group extension (i.e., in the SPT language, this equation holds for class DIII and for class BDI). Eq.~\ref{eqn:multiplicationOfKramers} was derived in the SPT language in Ref.~\cite{metlitski2014} and arises from carefully tracking fermion minus signs. Schematically, if the local action of ${\bf T}$ on $a$ and $b$ respectively consists of fermionic operators $c_a$ and $c_b$ where $c_{a,b}^2 \sim \eta_{a,b}^{\bf T}$, then we can equivalently think of the local action of ${\bf T}^2$ on the fusion of $a$ and $b$ as
            \begin{equation}
                (c_a c_b)^2 = -(c_a)^2 (c_b)^2 \sim -\eta_a^{\bf T}\eta_b^{\bf T}
            \end{equation}
            or as the local action $\eta_c^{\bf T}$ of ${\bf T}^2$ on $c$. The equality of these two local actions motivates Eq.~\ref{eqn:multiplicationOfKramers}. Our symmetry fractionalization framework allows an alternate derivation in the UBFC language, which we give in Appendix~\ref{app:multiplicationOfKramersProof}; the proof is a laborious but straightforward use of the consistency conditions, with the minus sign coming from an appearance of $R^{\psi \psi}$ in a hexagon equation.  
                 
            \subsection{\texorpdfstring{$G_b=\Z_4^{\bf T}$}{Gb = Z4T}}
            
            There are two group extensions $G_f$ of $\Z_4^{\bf T}$ by $\Z_2^f$, the trivial one $\Z_4^{\bf T} \times \Z_2^f$ and the nontrivial one $\Z_8^{{\bf T},f}$.
            
            For anyons with $\,^{\bf T}a = a$, the invariant
            \begin{equation}
                \eta_a^{\bf T} = \frac{\eta_a({\bf T,T}^2)\eta_a({\bf T}^2,{\bf T}^2)}{\eta_a({\bf T}^2,{\bf T})}
            \end{equation}
            is gauge-invariant and detects a nontrivial $\Z_4^{\bf T}$ projective representation, just as in the bosonic case~\cite{barkeshli2019rel}. When $\,^{\bf T}a=a$, $\eta_a({\bf g,h}) \in Z^2(\Z_4^{\bf T},\U)$, so the above expression is simply the cohomology invariant, equal to $\pm 1$, which detects whether $\eta_a({\bf g,h})$ characterizes a linear or projective representation of $\Z_4^{\bf T}$. In particular, $\eta_\psi^{\bf T}$ characterizes the group extension $\omega_2({\bf g,h}) = \eta_\psi({\bf g,h})$, and can be thought of as indicating whether or not $\psi$ has fractional charge under the unitary symmetry ${\bf T}^2$. If the charge is fractional, then schematically ${\bf T}^4 = (-1)^F$, and indeed the group extension is nontrivial. 
            
            The consistency conditions demonstrate that this quantity obeys the fusion rules in the sense that if $\,^{\bf T}a=a$,$\,^{\bf T}b=b$,$\,^{\bf T}c=c$, and $N_{ab}^c > 0$, then 
            \begin{equation}
                \eta_a^{\bf T}\eta_b^{\bf T} = \eta_c^{\bf T}
                \label{eqn:multiplicativeInvariants}
            \end{equation}
            Notably, $\eta_a^{\bf T} = \pm \eta_{a\psi}^{\bf T}$ with the upper sign for $G_f = \Z_4^{\bf T} \times \Z_2^f$ and the lower sign for $G_f=\Z_8^{{\bf T},f}$.
            
            If, on the other hand, $\,^{\bf T}a = a\times \psi$, then
            \begin{equation}
                \eta_a^{\bf T} = \frac{\eta_a({\bf T,T}^2)\eta_a({\bf T}^2,{\bf T}^2)}{\eta_a({\bf T}^2,{\bf T})}U_{{\bf T}^2}(a,\psi;a \times \psi)
            \end{equation}
            is the appropriate gauge-invariant object. In contrast to the case $G_b = \Z_2^{\bf T}$, we do not need any $F$-symbols to preserve invariance under vertex basis transformations; $U_{\bf T^2}(a,\psi;a\times \psi)$ is invariant under such transformations because ${\bf T}^2$ is unitary. Using the consistency conditions, one can check straightforwardly that
            \begin{equation}
                (\eta_a^{\bf T})^2 = \eta_\psi({\bf T}^2,{\bf T}^2)
            \end{equation}
            where the right-hand side is the cohomology invariant determining the group extension $[\omega_2]$. This quantity also obeys the fusion rules in the sense that if $\,^{\bf T}a=a \times \psi$, $\,^{\bf T}b = b \times \psi$, $\,^{\bf T}c=c$, and $N_{ab}^c>0$, then Eq.~\ref{eqn:multiplicativeInvariants} holds with $\eta_a^{\bf T}$ interpreted appropriately for the transformation properties of the anyons.
            
        \subsection{\texorpdfstring{$G_b=\U$}{Gb = U(1)}}
            \label{sec:U1Example}
            
        There are two extensions of $\U$ by $\Z_2^f$, i.e., $\H^2(\U,\Z_2)=\Z_2$. One extension is the trivial extension $G_f = \U \times \Z_2^f$ while the other is called $\U^f$ and is characterized by the cocycle
        \begin{equation}
            \omega_2({\bf g,h})=e^{i \left(\varphi_{\bf g}+\varphi_{\bf h} -[\varphi_{\bf g}+\varphi_{\bf h}]\right)/2}
        \end{equation}
        where ${\bf g}=e^{i\varphi_{\bf g}}$, $\varphi_{\bf g} \in [0,2\pi)$ and $[x] = x \mod 2\pi$. 
        
        Since $G_b$ is continuous, the condition that $\rho_{\bf g}$ obeys the group multiplication laws up to natural isomorphisms forces $\rho_{\bf g}$ to either be the identity or $[\Upsilon_\psi]$. In particular, this implies that $[\rho_{\bf g}]$ has trivial permutation action on the anyons. 
        
        Interestingly, if $\Upsilon_\psi$ does not respect locality, then any symmetry fractionalization of $\U^f$ automatically has a fermionic symmetry localization obstruction.
        Identify $[\rho_{\bf g}] \in \{[\text{Id}],[\Upsilon_\psi]\} \simeq \Z_2$, and encode whether $[\rho_{\bf g}]$ is the trivial or nontrivial element of $\Z_2$ by the function $\phi({\bf g}): G_b \rightarrow \Z_2$.\footnote{The map $\phi({\bf g})$ need not be continuous because $[\kappa_{\bf g,h}]$ need not be continuous.} Calculating $[\kappa_{\bf g,h}]$ directly from Eq.~\ref{eqn:kappaghDef}, we see that
        \begin{equation}
            \beta_\psi({\bf g,h})= d\phi({\bf g,h})
        \end{equation}
        It immediately follows from Eq.~\ref{eqn:fermionicObstructionUpsilonViolates} that $\coho{O}_f$ is trivial if and only if $\omega_2({\bf g,h})=d\phi$, which is only possible if $G_f = \U \times \Z_2^f$. Therefore, not all super-modular categories can be compatible with $\U^f$ symmetry. 
        
        The gauge-invariant quantity characterizing symmetry fractionalization is given as follows. For a fixed anyon $a$, let $n$ be the smallest integer such that $a^n$ contains the identity as a fusion product. Choose a sequence of anyons $a, a^2, \ldots, a^n=1$ such that $a \times a^k$ contains $a^{k+1}$ as a fusion product. Then define
        \begin{equation}
            e^{2\pi i Q_a} = \prod_{m=1}^{n-1} \eta_a\left(e^{2\pi i/n},e^{2\pi i m/n}\right)U_{e^{2\pi i/n}}(a,a^m;a^{m+1})
            \label{eqn:ourU1Invariant}
        \end{equation}
        One can check directly that this quantity is gauge-invariant.
        We immediately see that
        \begin{equation}
            e^{2\pi i Q_\psi} = \omega_2(-1,-1)
        \end{equation}
        which can be checked to be a cohomology invariant characterizing the group extension.
        
        Here $Q_a$ (which is only defined modulo an integer) can be interpreted as the fractional charge of the anyon $a$ under $\U$. One way to understand this interpretation is by noting that the fermion plays no role in this quantity, so we could instead consider the same invariant for a bosonic topological order, that is, taking $\mathcal{C}$ to instead be a UMTC. In that case, Ref.~\cite{manjunath2020} assumes the existence of a $G$-crossed MTC and defines an invariant which is equivalent to
        \begin{equation}
            e^{2\pi i Q_a}=\left(R^{a,0_{\bf g}}R^{0_{\bf g},a}\right)^n\prod_{m=1}^n \eta_a\left({\bf g},{\bf g}^m \right)
            \label{eqn:NarenU1Invariant}
        \end{equation}
        where ${\bf g}=e^{2\pi i/n}$, $a$ is any anyon, and $0_{\bf g}$ is any defect carrying ${\bf g}$ flux. Ref.~\cite{manjunath2020} shows that $Q_a$ has the interpretation as the fractional $\U$ charge of $a$ because this quantity can be understood as the double braid of $a$ with a $2\pi$ flux of the $\U$ symmetry. Combining the $G$-crossed heptagon equations for clockwise and counterclockwise braids, we find
        \begin{equation}
            \left(R^{a,0_{\bf g}}R^{0_{\bf g},a}\right)\left(R^{a^m,0_{\bf g}}R^{0_{\bf g},a^m}\right) = \left(R^{a^{m+1},0_{\bf g}}R^{0_{\bf g},a^{m+1}}\right)U_{\bf g}(a,a^m;a^{m+1})
        \end{equation}
        Inserting this identity into Eq.~\ref{eqn:NarenU1Invariant} shows that it is equivalent to Eq.~\ref{eqn:ourU1Invariant}, that is, we expect that if a $G$-crossed theory exists when $\mathcal{C}$ is super-modular, then Eq.~\ref{eqn:ourU1Invariant} also gives the fractional $\U$ charge of the anyon $a$.
        
        Note that for the nontrivial group extension $\omega_2(-1,-1)=-1$, the fermion carries $Q_\psi = 1/2$. This must be so because the fermion carries unit charge of the full symmetry group $\U^f$, which is a double cover of $G_b=\U$. We see, then, how the nontrivial group extension is encoded by giving the fermion a fractional quantum number under $G_b$, although it carries an integer charge under $G_f$. Note also that in this context, the physical charge carried by an anyon $a$, which is really the $\U^f$ charge, is $2Q_a$.
        
        We can characterize symmetry fractionalization with $G_b = \U$ a bit further if $\Upsilon_\psi$ respects locality. In this case, there exists a modular extension $\C$ with an Abelian fermion parity vortex $v_0$ and $[\rho_{\bf g}]=[\mathrm{Id}]$, where Id is the identity map. We can thus fix a gauge (for convenience) in which $\rho_{\bf g}$ is the identity (i.e. $U_{\bf g} = 1$), in which it is straightforward to check that
        \begin{equation}
            \eta_a^{\text{ref}}({\bf g,h}) = M_{a,\cohosub{v}_{\text{ref}}({\bf g,h})}
        \end{equation}
        with 
        \begin{equation}
            \coho{v}_{\text{ref}}({\bf g,h})=
            \begin{cases} v_0^{\left(\varphi_{\bf g}+\varphi_{\bf h}-[\varphi_{\bf g}+\varphi_{\bf h}]\right)/2\pi} & \text{ if }G_f = \U^f\\
            1 & \text{ if } G_f = \U \times \Z_2^f
            \end{cases}
        \end{equation}
        satisfies the consistency equations Eq.~\ref{eqn:etaConsistency} and~\ref{eqn:etaUConsistency} and the constraint Eq.~\ref{eqn:etaPsiConstraint}. According to the discussion in Sec.~\ref{subsubsec:H2Torsor}, all other symmetry fractionalization classes have representatives of the form
        \begin{equation}
            \eta_a({\bf g,h}) = \eta_a^{\text{ref}}({\bf g,h})M_{a,\cohosub{t}({\bf g,h})}
            \label{eqn:relativeEtaU1}
        \end{equation}
        with $[\coho{t}] \in \H^2(G_b,\A/\{1,\psi\})$. One can check that for each $x \in \A/\{1,\psi\}$, the function
        \begin{equation}
            \coho{t}({\bf g,h}) = x^{\left(\varphi_{\bf g}+\varphi_{\bf h}-[\varphi_{\bf g}+\varphi_{\bf h}]\right)/2\pi}
            \label{eqn:tghForU1}
        \end{equation}
        represents a distinct class $[\coho{t}]\in\H^2(G_b,\A/\{1,\psi\})$. 
        
        Hence, if $G_f = \U^f$, given an Abelian fermion parity vortex $v_0$ in some minimal modular extension $\C$ of $\mathcal{C}$, symmetry fractionalization is characterized by an anyon $x \in \A/\{1,\psi\}$. The anyon $x$ has a physical interpretation as a ``relative vison" between the reference fractionalization class and the class given by $x$, that is, inserting a $2\pi$ flux of $\U$ will insert an extra anyon $x$ (modulo a fermion) in the state corresponding to $x$ compared to carrying out the same process in the reference state. The ``absolute vison," that is, the anyon associated with insertion of a $2\pi$ flux, is actually a fermion parity vortex (again, modulo a fermion) and is thus valued in the minimal modular extension corresponding to the physical realization of the fermionic topological order; the vison is not an object in the super-modular category.
        
        If $G_f=\U  \times \Z_2^f$, then we do not need to specify a fermion parity vortex, and $x$ is the absolute vison.
        
        It also follows immediately from Eq.~\ref{eqn:ourU1Invariant} that
        \begin{equation}
            e^{2\pi i Q_a} = \begin{cases}
            M_{a,v_0 \times x} & \text{ if } G_f = \U^f\\
            M_{a,x} & \text{ if } G_f = \U \times \Z_2^f
            \end{cases}.
        \end{equation}
        
        These results are very similar to the bosonic case. The key differences are that in a bosonic theory, $v_0$ is never present, which means that the (absolute) vison $x$ is always valued in the base UMTC which defines the bosonic topological order, and there is no fermion ambiguity in $x$.
        
        \subsubsection{\texorpdfstring{$G_f = U(1)^f$}{Gf = U(1)f} and \texorpdfstring{$1/n$}{1/n} Laughlin FQH states}
        
        As an application of our formalism, we can consider the example of fermionic fractional quantum Hall states. In the simplest case, take the $1/n$ Laughlin state for $n$ odd, which has Abelian topological order described by the super-modular category $\mathcal{C} = \Z_n \times \{1,\psi\}$. This category contains a particle with topological twist $\theta = \pi/n$; denote this particle $[1]$. Then the generator of the factor $\tilde{\A} = \Z_n$ in Eq.~\ref{eqn:abelianDecomp} is the particle $[2]$. Furthermore, $[n] = \psi$. After gauge-fixing $U=1$, the consistency condition Eq.~\ref{eqn:etaUConsistency} requires the phases $\eta_a({\bf g,h})$ to obey the fusion rules, so
        \begin{equation}
            \left[\eta_{[1]}({\bf g,h})\right]^n = \eta_{[n]}({\bf g,h}) = \eta_\psi({\bf g}, {\bf h}) = \omega_2({\bf g,h})
        \end{equation}
        Hence
        \begin{equation}
            \eta_{[1]}({\bf g,h}) = e^{i q (\varphi_{\bf g}+\varphi_{\bf h}-[\varphi_{\bf g}+\varphi_{\bf h}])/2n},
        \end{equation}
        for any odd integer $q$, with $q \sim q+2n$ producing the same pattern. We have simply taken the $n$th root of $\omega_2$, hence the $q$ ambiguity, and enforced $\eta_{[1]}({\bf 1,h})=1$. This leads to
        \begin{equation}
            Q_{[1]} = \frac{q}{2n} = \frac{q}{n}Q_\psi.
        \end{equation}
        As expected, this fractionalization pattern assigns $q/n$ of the electron charge to each quasiparticle. We remind the reader that $Q_a$ is the charge under $G_b$, and that the ``true" $\U$ charge is the charge under $G_f$, i.e. $2Q_a$. The $n$ distinct possible values for $Q_{[1]}$ are consistent with our classification result, which yields $\H^2(\U,\A/\{1,\psi\})=\H^2(\U,\Z_n) = \Z_n$. In particular, there is a canonical reference state where 
        \begin{equation}
            \eta_{[k]}^{\text{ref}} = \omega_2({\bf g,h})^k
        \end{equation}
        for each integer $k$, i.e. $\eta_a = 1$ for $a \in \tilde{\A}$ but still obeys the fermionic constraint $\eta_\psi = \omega_2$. In this reference state, $[k]$ carries an electron charge for $k$ odd and and an integer charge under $G_b$ (even integer charge under $G_f$, which can be screened by a boson) for $k$ even. Following the definitions in Eqs.~\ref{eqn:relativeEtaU1},~\ref{eqn:tghForU1}, each anyon $x \in \widetilde{A}=\Z_n$ defines a distinct fractionalization pattern relative to the reference. The corresponding fractionalization data is
        
        \begin{equation}
            \eta_{[k]}({\bf g,h}) = \eta_{[k]}^{\text{ref}}({\bf g,h})M_{[k],\cohosub{t}({\bf g,h})}
        \end{equation}
        for each $k$. If $x=[2\ell]$, then
        \begin{equation}
            Q_{[1]} = \frac{2\ell}{n}+\frac{1}{2}
        \end{equation}
        For the case $n=3$, one can easily check that the choice $\ell=0$ corresponds to $[1]$ carrying integer $\U^f$ charge, i.e. the charge of an electron. The choice $\ell=1$ corresponds to $2Q_{[1]} = 1/3 \mod 1$ so that $[1]$ is the the Laughlin quasihole, while $\ell=2$ corresponds to $2Q_{[1]}=-1/3 \mod 1$ so that $[1]$ is the Laughlin quasielectron. The physical difference between these cases is whether the Laughlin quasihole carries topological twist $\theta=\pi/3$ and the quasielectron carries topological twist $\theta=4\pi/3$, or vice-versa.
        
        As an aside, $\Upsilon_\psi$ respects locality in this case because there exists an Abelian minimal modular extension of $\mathcal{C}$, e.g., $\Z_n \times (\text{toric code})$. Explicitly, the phases $\zeta_{[k]}=e^{\pi i k/n}$ obey the fusion rules, and $\zeta_{[n]}=\zeta_\psi=-1$.
        
        \subsection{\texorpdfstring{$G_f = U(1)^f \times [\Z^2 \rtimes \Z_M]$}{Gf = U(1)f x Z2 rtimes ZM} and fractional Chern insulators}
        
        We can now also straightforwardly extend the above analysis to include lattice space group symmetries as well, which give us a general understanding of fractional quantum numbers for fractional Chern insulators with charge conservation and space group symmetries. This complements a recent comprehensive analysis in the bosonic case \cite{manjunath2020,manjunath2021}.
        
        For simplicity, let us consider the case of the $1/n$ Laughlin topological order with $n$ odd, for which 
        we have $\mathcal{A} = \Z_n \times \Z_2$, $\tilde{\mathcal{A}} = \Z_n$, $\mathcal{A}/\{1,\psi\} = \Z_n$, and consider the case where symmetries do not permute anyons, so $[\rho]: G_b \rightarrow \text{Aut}_{LR}(\mathcal{C})$ is the trivial map. 
        
        We have $G_b = U(1) \times [\Z^2 \rtimes \Z_M]$, so we can use the result of Ref. \cite{manjunath2020,manjunath2021}:
        \begin{align}
            \mathcal{H}^2(G_b, \mathcal{A}/\{1,\psi\}) = \mathcal{H}^2(U(1) \times [\Z^2 \rtimes \Z_M], \Z_n) = \Z_n \times \Z_n \times (K_M \otimes \Z_n) \times \Z_{(n,M)} .
        \end{align}
        Here $K_M = \Z_1, \Z_2 \times \Z_2, \Z_3, \Z_2,\Z_1$ for $M = 1,2,3,4,6$, respectively. $\otimes$ is the tensor product of finite groups; we have $K_M \otimes Z_n = \Z_1, \Z_{(2,n)} \times \Z_{(2,n)}, \Z_{(3,n)}, \Z_{(2,n)}, \Z_1$ for $M = 1,2,3,4,6$, respectively. $(a,b)$ refers to greatest common divisor of $a$ and $b$. 
        
        Since we can pick a canonical reference state for which $\eta^{\text{ref}}_a = 1$ if $a \in \widetilde{\A}$ and $\eta_\psi = \omega_2$, we can characterize the other symmetry fractionalization classes by a set of anyons 
        $[\coho{t}] = (x, m, [\vec{t}], [s]) \in \mathcal{H}^2(G_b, \mathcal{A}/\{1,\psi\})$. Here $x \in \Z_n$ is the relative vison defined in Eq.~\ref{eqn:tghForU1}, $m \in \Z_n$ is the anyon per unit cell, $\vec{t} = \Z_n^2$ is the discrete torsion vector discussed in the bosonic case in  \cite{manjunath2020,manjunath2021}, and $s \in \Z_n$ is the discrete spin vector. The square brackets imply certain equivalence relations for $\vec{t}$ and $s$, so that $[s] \in \Z_{(n,M)}$ and $[\vec{t}] \in K_M \otimes \Z_n$ for which we refer the reader to \cite{manjunath2020,manjunath2021}. 
        
        The reference state and the choice of $[\coho{t}]$ determine the fractionalization class via
        \begin{align}
            \eta_a({\bf g},{\bf h}) = \eta_a^{\text{ref}}({\bf g}, {\bf h}) M_{a, \cohosub{t}({\bf g}, {\bf h})}. 
        \end{align}
        
        Physically, these fractionalization classes define the fractional $U(1)$ charge of the anyons (determined by $x$), the fractionalization of the translation algebra (determined by $m$), the fractional orbital angular momentum (determined by $[s]$), and the fractional linear momentum (determined by $[\vec{t}]$). The choices of $x$, $m$, $[s]$ and $[\vec{t}]$ also have non-trivial consequences for the fractional quantum numbers of lattice dislocations, disclinations, and magnetic flux, as discussed in \cite{manjunath2020,manjunath2021}. 
        
        Let us take as an example the case $n = 3$. Then, for $M = 2$ and $4$ (e.g. rectangular and square lattices), $[\vec{t}]$ and $[s]$ are automatically trivial, so the only non-trivial choices are the choice of relative vison $x \in \Z_3$ and anyon per unit cell $m \in \Z_3$. These two are further constrained by the fractional part of the filling fraction $\nu$: the fractional charge per unit cell, which sets the fractional part of $\nu$, must be equal to the charge of the anyon per unit cell, $m$. 
        
        For $M = 6$, e.g. the triangular lattice, $[\vec{t}]$ is still trivial, however we have a possible choice of $[s] \in \Z_{(3,6)} = \Z_3$. This specifies the fractional orbital angular momentum of the anyons, along with the fractional charge and fractional angular momentum of disclinations \cite{manjunath2020,manjunath2021}. 
        
        Finally, for $M = 3$, e.g. the honeycomb lattice, we finally have the possibility that both $[\vec{t}]$ and $[s]$ can be non-trivial. In particular, $[\vec{t}] \in \Z_3$, which implies the possibility of a non-trivial fractional linear momentum of the anyons, a non-trivial fractional charge of lattice dislocations, and other fractional quantized response properties \cite{manjunath2020,manjunath2021}.

        \subsection{\texorpdfstring{$G_f = U(2)^f$}{Gf = U(2)f} and \texorpdfstring{$\Z_2$}{Z2} quantum spin liquids }
        
        The case of spinful electrons forming an insulating state that respects both charge conservation symmetry and spin rotational symmetry corresponds to an on-site unitary symmetry $G_f = U(2)^f$, and $G_b = U(1) \times \SO(3)$. Let us consider this systematically for the case of the gapped $\Z_2$ spin liquid. In this case, the super-modular category is described by $D(\Z_2) \boxtimes \{1,\psi\}$, where $D(\Z_2)$ is the quantum double of $\Z_2$. Therefore, we have $\mathcal{A} = \{1,e,m,f\} \boxtimes \{1,\psi\}$, with $f = e \times m$ the emergent fermion. Our general classification then gives
        \begin{align}
           \mathcal{H}^2(G_b,  \mathcal{A}/\{1,\psi\}) = \mathcal{H}^2(U(1) \times \SO(3),  \Z_2 \times \Z_2) = \Z_2^4 .
        \end{align}
        We see that there are at most 16 distinct symmetry fractionalization classes. However we will see that most of these symmetry fractionalization class are in fact physically equivalent under relabeling the anyons, and that there are actually only 3 distinct symmetry fractionalization classes. 
        
        One can check that $\eta_a({\bf g,h})$ factors into a $\U$ part and a $\SO(3)$ part
        \begin{align}
            \eta_a((e^{i\varphi_{\bf g}},M_{\bf g}),(e^{i\varphi_{\bf_h}},M_{\bf h})) &= \eta_a^q(e^{i \varphi_{\bf g}},e^{i\varphi_{\bf h}}) \eta_a^s(M_{\bf g},M_{\bf h})
        \end{align}
        where $M_{\bf g,h}$ are $\SO(3)$ matrices and $q$ and $s$ label the charge and spin parts. Define functions
        \begin{equation}
            \phi_k(e^{i\varphi_1},e^{i\varphi_2}) = e^{ik(\varphi_1+\varphi_2-[\varphi_1+\varphi_2])}
        \end{equation}
        for $k=0,1/2 \mod 1$. We find that, in the gauge $U=1$, a representative $\eta_a$ is
        \begin{align}
            \eta_a^q(e^{i \varphi_{\bf g}},e^{i\varphi_{\bf h}}) &= \phi_{q_a}(e^{i \varphi_{\bf g}},e^{i\varphi_{\bf h}})\\
            \eta_a^s(M_{\bf g},M_{\bf h}) &= \phi_{s_a}(\alpha_{\bf g},\alpha_{\bf h})\phi_{s_a}(\beta_{\bf g},\beta_{\bf h})\phi_{s_a}(\gamma_{\bf g},\gamma_{\bf h})
        \end{align}
        for $q_a,s_a \in \{0,1/2\}$. The notation parameterizes $M_{\bf g} \in \SO(3)$ by the Euler angles $(\alpha_{\bf g},\beta_{\bf g},\gamma_{\bf g})$. The fermionic constraint forces $q_\psi = s_\psi = 1/2$, while the $U-\eta$ consistency condition Eq.~\ref{eqn:etaUConsistency} forces the $\eta_a^q$ and $\eta_a^s$ to obey the fusion rules. Hence, independent choices of $q_e,s_e,q_m,s_m \in \{0,1/2\}$ distinguish the different symmetry fractionalization classes. This generates the $\Z_2^4$ classification. Following the discussion in Sec.~\ref{sec:U1Example}, $q_a$ measures the charge of $a$ under the $\U$ part of $G_b$, and hence $2q_a$ measures the physical $\U$ charge of $a$. Similarly, $s_a$ measures the spin of $a$ (there is no doubling of the quantum number here since half-integer spin under $\SO(3)$ corresponds to half-integer spin under its double cover $\mathrm{SU}(2)$).
        
        However, many of these patterns are physically equivalent up to relabeling anyons. For example, relabeling $e \leftrightarrow m$ swaps the quantum numbers of $e$ and $m$. One could also relabel $e \leftrightarrow f\psi$ without changing $m$; this is an autoequivalence of the theory and interchanges the classes $(q_e,s_e) \leftrightarrow (q_e+1/2,s_e+1/2)$, where we are always taking addition modulo 1. Likewise, interchanging $m \leftrightarrow f\psi$ without changing $e$ is an autoequivalence and interchanges $(q_m,s_m)\leftrightarrow (q_m+1/2,s_m+1/2)$. Accounting for all of these relabelings, we obtain only three distinct classes, with representative choices of $(q_e,s_e,q_m,s_m) \in \{ (0,0,0,0),(1/2,0,0,0), (1/2,0,0,1/2)\}$. The first case is trivial; all anyons' quantum numbers can be screened by local excitations, i.e., they either have quantum numbers allowed by for a local boson or the quantum numbers of an electron (up to local bosons). In the second, $e$ (up to relabeling) carries the electron charge but no spin, $e\psi$ carries spin-$1/2$ but no charge, and $m$ has quantum numbers which can be screened by local excitations. In the third, $e$ carries the electron charge but no spin, while $m$ carries spin-$1/2$ but no charge. In the latter two cases, we can think of the local fermion as fractionalizing into a chargeon $e$ (up to relabeling) and spinon $e\psi$. The difference between the two cases is whether or not the $\Z_2$ flux $m$ carries quantum numbers which can be screened by local excitations.
    
        \subsection{Symmetry localization obstructions}
        \label{subsec:symmLocObstructions}        
        \subsubsection{\texorpdfstring{$Z^2(G_b, \Z_2)$}{Z2(Gb, Z2)} fermionic obstruction, \texorpdfstring{$\Upsilon_\psi$}{Ypsi} locality-violating, and \texorpdfstring{$\SO(3)_3$}{SO(3)3}}
        \label{sec:SO33}
        
        A simple example of a non-trivial $Z^2(G_b,\Z_2)$ fermionic symmetry localization obstruction is the super-modular category $\SO(3)_3$ with $G_f = \Z_2^{\bf T}\times \Z_2^f$. The theory $\SO(3)_3$ has four anyons $1,s,\tilde{s},\psi$ with quantum dimensions $1,1+\sqrt{2},1+\sqrt{2},1$ and topological spins $1,i,-i,-1$. The fusion rules are
        \begin{align}
            \psi \times \psi &= 1 \nonumber \\
            \psi \times s &= \tilde{s} \nonumber \\
            s \times s = \tilde{s} \times \tilde{s} &= 1+s+\tilde{s} \nonumber \\
            s \times \tilde{s} &= \psi + s + \tilde{s}
        \end{align}
        Time reversal must exchange $s \leftrightarrow \tilde{s}$. One can show that consistent fermionic symmetry fractionalization data exists with $G_f = \Z_4^{{\bf T},f}$, i.e., $\eta_\psi({\bf T,T})=-1$; see Refs.~\cite{Fidkowski13,TataSETAnomaly}.
        
        To see that $G_f = \Z_2^{\bf T} \times \Z_2^f$ should have a fermionic symmetry localization obstruction, we use the fact that on general grounds~\cite{BarkeshliChengLocalizationAnomalies}, if $\,^{\bf T}b = b$ and there exists any anyon $a$ such that $N_{a \,^{\bf T}a}^b$ is odd, then
        \begin{equation}
            \eta_b({\bf T,T}) = \theta_b
            \label{eqn:etaThetaRelation}
        \end{equation}
        The above equation applies for $b=\psi$ since $N_{s \,^{\bf T}s}^{\psi} = 1$. Hence $\eta_\psi({\bf T,T})=-1$ for any consistent symmetry fractionalization. However $G_f = \Z_2^{\bf T} \times \Z_2^f$ would require $\eta_\psi({\bf T}, {\bf T}) = 1$, which is therefore inconsistent. 
        
        We can state the above in the language of our present work as follows. Since $\SO(3)_3$ contains the fusion rules $N_{s,s}^s = N_{s,s}^{\psi \times s}=1$, $\Upsilon_\psi$ must violate locality. That is, all phases $\zeta_a$ which respect the fusion rules have $\zeta_\psi = +1$. One can check that the action of time reversal must have $\beta_\psi({\bf T,T}) = -1$ (this can be computed directly, and also follows from the knowledge that there exists consistent fractionalization data with $\eta_\psi({\bf T,T})=-1$ and $\omega_\psi({\bf T,T})=1$). Hence, the fermionic obstruction for $\omega_2({\bf T,T})=+1$ is characterized by the cohomology invariant $\O_f({\bf T,T}) = \beta_\psi({\bf T,T})/\omega_2({\bf T,T}) = -1$. That is, $\O_f$ is nontrivial in $Z^2(\Z_2^{\bf T},\Z_2)$ and $[\O_f]$ is also nontrivial in $\H^2(\Z_2^{\bf T},\Z_2)$.
        
        We note in passing that a related phenomenon was discovered in Ref. \cite{barkeshli2019tr} for $D(S_3)$, the quantum double of $S_3$, the permutation group on three elements, which describes the anyon content of $S_3$ gauge theory (see Section VIII D of Ref. \cite{barkeshli2019tr}). The anyons of $D(S_3)$ can be labeled as $A,B,C,D,E,F,G,H$, and the theory admits an action of $\mathbb{Z}_2^{\bf T}$ such that $C \leftrightarrow F$ and $G \leftrightarrow H$ under time-reversal ${\bf T}$. This permutation action forces $\eta_B^{\bf T} = -1$. Thus, one can consider the subcategory generated by $\{1,B\}$, in which case the choice $\eta_B^{\bf T} = 1$ would be obstructed when attempting to lift the action of ${\bf T}$ to the full category. 
        
        \subsubsection{\texorpdfstring{$\mathcal{H}^3(G_b, \mathcal{A}/\{1,\psi\})$}{H3(Gb, A/1,psi} fermionic obstruction, \texorpdfstring{$\Upsilon_\psi$}{Ypsi} locality-preserving, and \texorpdfstring{$\mathrm{Sp}(2)_2 \boxtimes \{1,\psi\}$}{Sp(2)2 x 1,psi}}
        \label{sec:Sp22}
        
        To demonstrate the bosonic symmetry localization obstruction and the fermionic symmetry localization obstruction when $\Upsilon_\psi$ respects locality, we consider the super-modular theory $\mathcal{C} = \mathrm{Sp}(2)_2 \times \{1,\psi\}$ with $G_b=\Z_2^{\bf T}$. The notation here is $\mathrm{Sp}(2)_2 = \mathrm{USp}(4)_2$.
        
        The theory $\mathrm{Sp}(2)_2$ was studied in detail in~\cite{barkeshli2018}. Here we show that one can define a symmetry action of $\Z_2^{\bf T}$ on $\mathcal{C}$ that has a bosonic symmetry localization obstruction, which is essentially identical to the symmetry localization obstruction for $\mathrm{Sp}(2)_2$ as a bosonic topological order, and one with no bosonic obstruction. For the symmetry action with no bosonic obstruction, we show that $G_f=\Z_2^{\bf T} \times \Z_2^f$ has a fermionic symmetry localization obstruction, while $G_f = \Z_4^{{\bf T},f}$ is unobstructed.
        
        The anyon content of the UMTC $\mathrm{Sp}(2)_2$ is $\{1,\epsilon, \phi_1,\phi_2, \psi_+, \psi_-\}$ (the transparent fermion $\psi$, with no subscript, is not related to $\psi_\pm$), with topological spins $1, 1, e^{4\pi i /5}, e^{-4\pi i/5}, i, -i$, respectively. These have quantum dimensions $1,1,2,2,\sqrt{5}, \sqrt{5}$, respectively. Under time-reversal, there are two choices of actions on $\mathcal{C}$:
        \begin{enumerate}
            \item $\phi_1 \leftrightarrow \phi_2$, $\psi_+ \leftrightarrow \psi_-$, with $\epsilon$ and $\psi$ staying invariant.
            \item $\phi_1 \leftrightarrow \phi_2$, $\psi_+ \leftrightarrow \psi \psi_+$, $\psi_- \leftrightarrow \psi \psi_-$, with $\epsilon$ and $\psi$ invariant. 
        \end{enumerate}
        The first choice has a bosonic $\H^3(\Z_2^{\bf T},\A/\{1,\psi\})$ obstruction. One can directly compute the $U$-symbols and calculate the bosonic obstruction. One finds that $\coho{O}_b({\bf T,T,T}) = \epsilon$ as an element of $\mathcal{A}/\{1,\psi\}$; this is a cohomology invariant, and therefore the bosonic obstruction is nontrivial.
        
        For the second choice of symmetry action, one can directly check that the bosonic obstruction vanishes, and we can ask about the fermionic obstruction. In this theory, $\Upsilon_\psi$ respects locality because there is a minimal modular extension of $\mathcal{C}$ with an Abelian fermion parity vortex (namely the product of $\mathrm{Sp}(2)_2$ and any Abelian minimal modular extension of $\{1,\psi\}$). Hence the fermionic obstruction is valued in $\H^3(\Z_2^{\bf T},\A/\{1,\psi\}) = \H^3(\Z_2^{\bf T},\Z_2) = \Z_2$, characterized by the cohomology invariant $\O_f({\bf T,T,T})$.
        
        There is in fact a fermionic obstruction only when $G_f = \Z_2^{\bf T} \times \Z_2^f$. We can again use the criterion Eq.~\ref{eqn:etaThetaRelation} to see why. With the second symmetry action above, taking $a=\psi_+$, we find
        \begin{equation}
            \psi_+ \times \,^{\bf T}\psi_+ = \psi_+ \times \psi \psi_+ = \psi + \psi \phi_1 + \psi \phi_2
        \end{equation}
        Hence we can take $b=\psi$ in Eq.~\ref{eqn:etaThetaRelation} to see
        \begin{equation}
            \eta_\psi({\bf T,T}) = \theta_\psi = -1.
        \end{equation}
        Therefore, symmetry fractionalization with $\eta_\psi({\bf T,T}) =+1$ must be inconsistent, that is, there is a fermionic symmetry localization obstruction for $G_f = \Z_2^{\bf T} \times \Z_2^f$.
        
        We may also see the obstruction at the level of cohomology as follows. Suppose we are given a fractionalization pattern with $G_f = \Z_4^{{\bf T},f}$; one can check with tedious calculation that two such patterns exist and are specified by 
        $\eta^{\bf T}_\epsilon = -1$, $\eta^{\bf T}_{\psi_+} = \pm i$, $\eta^{\bf T}_{\psi_-} = \mp i$, $\eta^{\bf T}_\psi = -1$. Then we may attempt to find a new fractionalization pattern with $G_f=\Z_2^{\bf T}\times \Z_2^f$ by choosing a phase $\tau_a({\bf T,T})$ which obeys the fusion rules and which obeys $\tau_\psi({\bf T,T})=-1$. One such choice is $\tau_a({\bf T,T}) = (-1)^F$, where $(-1)^F$ measures the $\psi$ parity. We compute the failure of $\tau_a$ to be a group cocycle (see Eq.~\ref{eqn:TaDefinition}):
        \begin{equation}
            T_a({\bf T,T,T})=\tau_{\,^{\bf T}a}({\bf T,T})\tau_a({\bf T,T}) = \begin{cases}
              -1 & a = \psi_\pm,\psi_\pm \psi\\
              +1 & \text{else}
            \end{cases}
        \end{equation}
        It is straightforward to check that
        \begin{equation}
            T_a({\bf T,T,T}) = M_{a,\epsilon}
        \end{equation}
        so that $\O_f({\bf T,T,T})=\epsilon$, that is, $[\O_f] \neq 0 \in \H^3(\Z_2^{\bf T},\A/\{1,\psi\})$. Hence there is indeed a fermionic symmetry localization obstruction for $G_f=\Z_2^{\bf T}\times \Z_2^f$.
        
        In fact, one can use the same basic calculation for any theory of the form $\mathcal{C} = \mathcal{B} \boxtimes \{1,\psi\}$, where $\mathcal{B}$ is modular and with $G_b=\Z_2^{\bf T}$. If the bosonic symmetry localization obstruction vanishes, then there is some choice of $\eta_\psi({\bf T,T})$ and thus a corresponding choice of group extension $G_f$ which gives a consistent symmetry fractionalization pattern. Given the consistent pattern, one can try to see if a fractionalization pattern exists for the other group extension. One can always calculate the obstruction by choosing $\tau_a({\bf T,T}) = (-1)^F$, in which case one finds
        \begin{equation}
            T_a({\bf T,T,T}) = \begin{cases}
              +1 & a \text{ and } \,^{\bf T}a \text{ have the same fermion parity}\\
              -1 & a \text{ and } \,^{\bf T}a \text{ have opposite fermion parity}
            \end{cases}
        \end{equation}
        in which case the cohomology invariant $\O_f({\bf T,T,T}) \in \A/\{1,\psi\}$ must be nontrivial if any $a$ and its time-reverse have opposite fermion parity. Therefore, for this form of $\mathcal{C}$ with trivial bosonic symmetry localization obstruction, both group extensions for $G_b=\Z_2^{\bf T}$ are unobstructed if and only if $a$ and $\,^{\bf T}a$ have the same fermion parity for all $a$.
        
        \section{Discussion}
        \label{sec:discussion}
        
        We have provided a systematic analysis of symmetry fractionalization in (2+1)D fermionic symmetry-enriched topological phases of matter. We saw that much of the formalism in the bosonic case goes through, with important modifications arising from the locality of fermions. We find that symmetry fractionalization depends on a choice $[\rho] : G_b \rightarrow \text{Aut}_{LR}(\mathcal{C})$, where $\text{Aut}_{LR}(\mathcal{C})$ is the group of locality-respecting autoequivalences. 
        
        Furthermore, the choice $[\rho]$ and $\omega_2$ may lead to bosonic or fermionic localization obstructions. The bosonic one is an obstruction to having any symmetry fractionalization class for $G_b$, regardless of $G_f$. 
        The fermionic one is an obstruction to having any symmetry fractionalization class for $G_f$, assuming the bosonic obstruction vanishes.  
        
        If these obstructions both vanish, then the symmetry fractionalization data is specified by a set of phases $\eta_a({\bf g}, {\bf h})$, which form a torsor over $\mathcal{H}^2(G_b, \mathcal{A}/\{1,\psi\})$. 
        
        The presence of local fermions leads to several consequences which are uncommon in the usual bosonic case, in particular that there can be more than one physically distinct class of symmetry action, $[\rho]$, with the same permutation action on the anyons. It would be interesting to find any microscopic model where a non-permuting but non-trivial symmetry action $\Upsilon_\psi$ occurs.
        
        We note that in using super-modular categories to model fermionic topological phases of matter, we required that the vertex basis gauge transformation $\Gamma^{\psi,\psi}_1 = 1$. It would be useful to develop a first principles derivation of such a constraint, which is so far lacking in our current understanding. 
        
        Having established this framework for symmetry fractionalization, it is an important issue to use it to understand the 't Hooft anomalies of fermionic SETs, which provide obstructions to gauging the full $G_f$ symmetry. Ref.~\cite{TataSETAnomaly}
        provided a general method to compute 't Hooft anomalies, but understanding in more detail the categorical origin of these anomalies in general is a non-trivial problem. There has been some progress in special cases~\cite{fidkowski2018,delmastro2021}. A general understanding has recently been provided in \cite{bulmash2021Anomalies,aasen21ferm}. 
        
        The formalism that we have developed here provides a partial understanding of fermionic SETs in (2+1)D by developing the theory of symmetry fractionalization for the anyons of a fermionic topological phase. A complete analysis requires also developing a theory of $G_f$ symmetry defects applicable to fermionic topological phases, to mirror the $G$-crossed braided tensor category approach for bosonic topological phases. In particular, such an analysis would incorporate symmetry fractionalization for the fermion parity vortices as well. Once symmetry fractionalization for the anyons is fixed, we expect that distinct $G_f$ defect classes can be obtained by stacking invertible fermionic topological phases.
        Recently, Ref.~\cite{manjunath21inv,aasen21ferm} has developed a comprehensive understanding of invertible fermionic topological phases with symmetry by augmenting the formalism of $G$-crossed braided tensor categories. The results of Ref.~\cite{manjunath21inv} suggest that the more general case of fermionic SETs may proceed by gauging fermion parity and classifying the resulting possible $G_b$-crossed braided tensor categories, while keeping track of additional flux labels that determine how $G_b$ defects arise from $G_f$ defects. 

        \section{Acknowledgements}
        
        We thank Parsa Bonderson, Meng Cheng, and Zhenghan Wang for discussions, and Srivatsa Tata and Ryohei Kobayashi for recent collaborations on related work. In particular, MB also thanks Parsa Bonderson for sharing insights on fermionic topological phases and preliminary results from \cite{aasen21ferm} on fermionic symmetry fractionalization and SETs, some of which we later reproduced independently in this work. This work is supported by NSF CAREER (DMR- 1753240) and JQI-PFC-UMD. 
        
        \it Note added: \rm This paper appeared on the arXiv at the same time as a number of other closely related papers, Ref. \cite{bulmash2021Anomalies,manjunath21inv,aasen21ferm}. In particular, Ref. \cite{bulmash2021Anomalies} provides a comprehensive account of obstructions to gauging $G_f$, leading to a systematic understanding of anomalies in (2+1)D fermionic topological phases. Ref. \cite{manjunath21inv} develops a systematic characterization and classification of invertible (2+1)D fermionic topological phases using the framework of $G$-crossed BTCs. Ref. \cite{aasen21ferm} independently develops a comprehensive characterization of (2+1)D fermion SETs, containing many of the results of this paper and of Ref. \cite{bulmash2021Anomalies, manjunath21inv}. Ref. \cite{aasen21ferm} additionally also develops the theory of symmetry fractionalization for fermion parity vortices and symmetry defects in fermionic SETs.
        
        Several changes in this revision, namely the discussion around Eq. \ref{eqn:kappaghDef} and the addition of Eq. \ref{eqn:kappaUpsilonPsi} have overlap with the discussion of Ref. \cite{aasen21ferm}, as noted in the main text.

        \appendix 
        
        \section{Characters of the fusion algebra of super-modular tensor categories}
        \label{app:characters}
        
        Given a BFC $\mathcal{C}$, we say that a function $\chi(a)$ is a character of the fusion algebra of $\mathcal{C}$ if \begin{equation}
            \chi(a)\chi(b) = \sum_{c \in \mathcal{C}} N_{ab}^c \chi(c),
        \end{equation}
        for all $a,b \in \mathcal{C}$.
        
        In this appendix, we will prove the following \cite{charNote}:
        \begin{theorem}
        Let $\mathcal{C}$ be a super-modular category and $\C$ be any minimal modular extension $\C$. Then any character $\chi(a)$ of the fusion algebra of $\mathcal{C}$ is of the form
        \begin{equation}
        \label{chiDef}
            \chi(a) = d_a M_{a,x}
        \end{equation}
        for some $x \in \C$ which is unique up to fusion with $\psi$. 
        \end{theorem}
        
        This is particularly useful to us because of the following corollary:
        \begin{corollary}
        Suppose that $e^{i\phi_a} \in K(\mathcal{C})$, that is, $e^{i\phi_a}$ is a phase obeying
        \begin{equation}
            e^{i\phi_a}e^{i\phi_b}=e^{i\phi_c}
        \end{equation}
        whenever $N_{ab}^c \neq 0$ for $a,b,c\in\mathcal{C}$. Then
        \begin{equation}
            e^{i\phi_a} = M_{a,x}
        \end{equation}
        for some $x \in \C$, and if $e^{i\phi_\psi}=+1$, then $x \in \A$.
        \end{corollary}
        
        \begin{proof} (Corollary)
        Given such an $e^{i\phi_a}$, we see that $d_a e^{i\phi_a}$ is a character of the fusion algebra since
        \begin{align}
        d_a d_b &= \sum_{c \in \mathcal{C}} N_{ab}^c d_c
        \nonumber \\
        d_a e^{i \phi_a}d_b e^{i \phi_b} &= \sum_{c \in \mathcal{C}} N_{ab}^c d_c e^{i \phi_a} e^{i \phi_b} = \sum_{c \in \mathcal{C}} N_{ab}^c d_c e^{i \phi_c}
        \end{align}
        Eq. \ref{chiDef} would then imply that $e^{i\phi_a}=M_{a,x}$ for some $x \in \C$. If $e^{i\phi_\psi}=1$, then $x \in \mathcal{C}$ and therefore $x \in \A$.
        \end{proof}

        \begin{proof} (Theorem)
        By modularity, all characters of $\C$ are of the form $\widecheck{\chi}_x(a) = d_a M_{a,x}$ for each $x \in \C$. Their restrictions $\chi_x(a)$ to $a \in \C_0=\mathcal{C}$ are clearly characters of $\mathcal{C}$.
        
        The super-modular tensor category $\mathcal{C}$ has at most $|\mathcal{C}|$ distinct characters.\footnote{Characters of the fusion algebra are given by the eigenvalues of the fusion matrices, which are all simultaneously diagonalizable, and thus the number of distinct characters is at most the number of eigenvalues.} Therefore, if we show that the collection $\{\chi_x\}$ define $|\mathcal{C}|$ distinct characters of $\mathcal{C}$, then every character must be of the form $\chi_x$ for some $x$.
        
        The $\chi_x$ define at most $|\C|$ distinct characters, one for each $x$. However, certainly $\chi_x = \chi_{x \times \psi}$. We are therefore overcounting; if $x \in \C_0$ or $x \in \C_v$, then there we should only count one of $x$ or $x \times \psi$ as defining a possibly distinct character. Hence the $\chi_x$ define at most $n=|\C_0|/2+|\C_v|/2+|\C_\sigma|$ distinct characters. By a theorem of~\cite{bruillard2017b}, $n=|\mathcal{C}|$, so we need to show that these $\chi_x(a)$ are indeed all distinct on $\mathcal{C}$. 
        We show the contrapositive, i.e., that if $x$ and $y \in \C$ define the \textit{same} character of $\mathcal{C}$, then $y = x$ or $y=x \times \psi$. Equivalently, we wish to show that $\chi_x(a) = \chi_y(a)$ for all $a \in \mathcal{C}$ implies $N_{x1}^y+N_{x\psi}^y>0$. We use the Verlinde formula in the modular category $\C$:
        
        \begin{align}
            N_{x1}^y+N_{x\psi}^y &= \sum_{z \in \C}\frac{S_{xz}S_{1z}S_{yz}^{\ast}}{S_{1z}} +\sum_{z \in \C } \frac{S_{xz}S_{\psi z}S_{yz}^{\ast}}{S_{1z}}\\
            &= 2\sum_{z \in \C_0} S_{xz}S_{yz}^{\ast}
        \end{align}
        where we have used $S_{\psi z} = \pm S_{1z}$ with the upper sign for $z \in \C_0$ and the lower sign for $z \in \C_1$. Now we insert factors of the identity and use the definition of scalar monodromy:
        \begin{align}
            N_{x1}^y+N_{x\psi}^y &= 2\sum_{z \in \C_0}  \frac{S_{1z}^2 S_{1x}S_{1y}^{\ast}}{S_{11}^2} \frac{S_{xz}S_{11}}{S_{1x}S_{1z}}\left(\frac{S_{yz}S_{11}}{S_{1y}S_{1z}}\right)^{\ast}\\
            &= 2d_x d_y\mathcal{D}^2\sum_{z \in \C_0}   d_z^2M_{xz}M_{yz}^{\ast}\\
            &= 2d_x d_y\mathcal{D}^2\sum_{z \in \C_0} \chi_x(z)\chi_y(z)^{\ast} = 2d_x d_y\mathcal{D}^2\sum_{z \in \C_0} |\chi_x(z)|^2>0 ,
        \end{align}
        which is what we wanted to show.
        \end{proof}
        
        \section{Comments on gauge transformations of \texorpdfstring{$U$}{U}}
        \label{app:Utransform}
        
        In order for Eq.~\ref{eqn:fermionicKramersEta} to be gauge-invariant, we require
        \begin{equation}
            \widecheck{U}_{\bf g}(a,b;c) = \left(\Gamma^{\,^{\overline{\bf g}}a,\,^{\overline{\bf g}}b}_{\,^{\overline{\bf g}}c}\right)^{\sigma({\bf g})}U_{\bf g}(a,b;c)\left(\Gamma^{a,b}_c\right)^{-1}
            \label{eqn:Utransform}
        \end{equation}
        where the check denotes the gauge-transformed quantity (in this appendix we will never be discussing modular extensions, so checks will always refer to gauge transformations).
        
        The origin of this transformation law is somewhat subtle, so we discuss it presently from two points of view.
        
        First, in the usual formalism, for ${\bf g}$ antiunitary, we should pick a basis $\ket{a,b;c}$ for the fusion space $V^{a,b}_c$ and define the antilinear operator $\rho_{\bf g}$ by
        \begin{equation}
            \rho_{\bf g}\left(\ket{a,b;c}\right) = U_{\bf g}(\,^{\bf g}a,\,^{\bf g}b;\,^{\bf g}c)\ket{\,^{\bf g}a,\,^{\bf g}b;\,^{\bf g}c}
        \end{equation}
        on the basis states (we suppress the internal indices if $N_{ab}^c>1$). We then extend the definition of $\rho_{\bf g}$ to the rest of the space by antilinearity. Making a vertex basis transformation means defining a new basis
        \begin{equation}
            \widecheck{\ket{a,b;c}} = \Gamma^{a,b}_c\ket{a,b;c}
        \end{equation}
        and then defining the gauge-transformed $U$ by
        \begin{equation}
             \rho_{\bf g}\left(\widecheck{\ket{a,b;c}}\right) = \widecheck{U}_{\bf g}(\,^{\bf g}a,\,^{\bf g}b;\,^{\bf g}c)\widecheck{\ket{\,^{\bf g}a,\,^{\bf g}b;\,^{\bf g}c}}
        \end{equation}
        We can now compute $\widecheck{U}$ directly:
        
        \begin{align}
            \rho_{\bf g}\left(\widecheck{\ket{a,b;c}}\right) &= \rho_{\bf g}\left(\Gamma^{a,b}_c\ket{a,b;c}\right)\\
            &= \left(\Gamma^{a,b}_c\right)^{\sigma(\bf g)}\rho_{\bf g}\left(\ket{a,b;c}\right)\\
            &=\left(\Gamma^{a,b}_c\right)^{\sigma(\bf g)}U_{\bf g}(\,^{\bf g}a,\,^{\bf g}b;\,^{\bf g}c)\ket{\,^{\bf g}a,\,^{\bf g}b;\,^{\bf g}c}\\
            &= \left(\Gamma^{a,b}_c\right)^{\sigma(\bf g)}U_{\bf g}(\,^{\bf g}a,\,^{\bf g}b;\,^{\bf g}c)\left(\Gamma^{\,^{\bf g}a,\,^{\bf g}b}_{\,^{\bf g}c}\right)^{-1}\widecheck{\ket{\,^{\bf g}a,\,^{\bf g}b;\,^{\bf g}c}}\\
            &= \widecheck{U}_{\bf g}(\,^{\bf g}a,\,^{\bf g}b;\,^{\bf g}c)\widecheck{\ket{\,^{\bf g}a,\,^{\bf g}b;\,^{\bf g}c}}.
        \end{align} 
        Comparing the last two lines leads directly to Eq.~\eqref{eqn:Utransform}.
        
        One must be careful in treating $\rho_{\bf g}$ as antilinear. Naively writing
        \begin{equation}
            \text{``}\rho_{\bf g}\ket{a,b;c} = U_{\bf g}(\,^{\bf g}a,\,^{\bf g}b;\,^{\bf g}c)K\ket{a,b;c}\text{"}
        \end{equation}
        (we have put quotes around this equation to emphasize that it can be misleading) and attempting to derive, for example, the consistency equation Eq.~\ref{eqn:UFconsistency} leads to an incorrect result with incorrect complex conjugations.
        
        For another perspective on the vertex basis transformations of $U$, consider the higher-category point of view on antiunitary symmetry in~\cite{bulmash2020}. Here, fusion vertices live in vector spaces $\ket{a,b;c;{\bf g}}\in \tilde{V}^{ab}_c({\bf g})$ which carry a ${\bf g}$ label, which can be roughly interpreted as a local spacetime orientation. The theory comes equipped with maps
        \begin{equation}
            \alpha_{\bf h}: \tilde{V}^{ab}_c({\bf g}) \rightarrow \tilde{V}^{ab}_c({\bf g \overline{h}})
        \end{equation}
        which are linear if the action of ${\bf h}$ is unitary and anti-linear if the action of ${\bf h}$ is antiunitary. The data of the theory is equivariant under these $\alpha_{\bf h}$ maps, e.g.,
        \begin{equation}
            \widetilde{F}^{abc}_{def}({\bf h}) = \left(\widetilde{F}^{abc}_{def}({\bf h\overline{g}})\right)^{\sigma({\bf g})},
            \label{eqn:higherCatF}
        \end{equation}
        where the tildes are present as a reminder that we are in the higher category formalism. The ``tilded" data is related to the usual ``untilded" data by simply setting ${\bf g}={\bf 1}$, e.g.,
        \begin{equation}
            F^{abc}_{def}=\tilde{F}^{abc}_{def}({\bf 1}).    
        \end{equation}
        In this formalism, antiunitary transformations act via a \textit{unitary} map
        \begin{equation}
            \widetilde{\rho}_{\bf g}\ket{a,b;c;{\bf h}} = \widetilde{U}(\,^{\bf g}a,\,^{\bf g}b;\,^{\bf g}c;{\bf gh,h})\ket{\,^{\bf g}a,\,^{\bf g}b;\,^{\bf g}c;{\bf gh}}
        \end{equation}
        
        Since $\widetilde{\rho}_{\bf g}$ is always unitary, it is easy to check its transformation law under gauge transformations
        \begin{equation}
            \widecheck{\ket{a,b;c;{\bf h}}} = \widetilde{\Gamma}^{a,b}_{c}({\bf h})\ket{a,b;c;{\bf h}}
        \end{equation}
        as follows:
        \begin{align}
            \widetilde{\rho}_{\bf g}\left(\widecheck{\ket{a,b;c;{\bf h}}}\right) &=  \widetilde{\rho}_{\bf g} \left(\Gamma^{a,b}_c({\bf h})\ket{a,b;c;{\bf h}}\right)\\
            &=\widetilde{\Gamma}^{a,b}_c({\bf h})\widetilde{U}(\,^{\bf g}a,\,^{\bf g}b;\,^{\bf g}c;{\bf gh,h})\ket{\,^{\bf g}a,\,^{\bf g}b;\,^{\bf g}c;{\bf gh}}\\
            &=\widetilde{\Gamma}^{a,b}_c({\bf h})\widetilde{U}(\,^{\bf g}a,\,^{\bf g}b;\,^{\bf g}c;{\bf gh,h})\left(\widetilde{\Gamma}^{\,^{\bf g}a,\,^{\bf g}b}_{\,^{\bf g}c}({\bf gh})\right)^{-1}\widecheck{\ket{\,^{\bf g}a,\,^{\bf g}b;\,^{\bf g}c;{\bf gh}}}\\
            &= \widecheck{\widetilde{U}}(\,^{\bf g}a,\,^{\bf g}b;\,^{\bf g}c;{\bf gh,h})\widecheck{\ket{\,^{\bf g}a,\,^{\bf g}b;\,^{\bf g}c;{\bf gh}}}
        \end{align}
        Hence
        \begin{equation}
            \widecheck{\tilde{U}}(\,^{\bf g}a,\,^{\bf g}b;\,^{\bf g}c;{\bf gh,h}) = \widetilde{\Gamma}^{a,b}_c({\bf h})\tilde{U}(\,^{\bf g}a,\,^{\bf g}b;\,^{\bf g}c;{\bf gh,h})\left(\widetilde{\Gamma}^{\,^{\bf g}a,\,^{\bf g}b}_{\,^{\bf g}c}({\bf gh})\right)^{-1}
            \label{eqn:tildedTransform}
        \end{equation}
        Equivariance of the $F$-symbols Eq.~\ref{eqn:higherCatF}
        forces equivariance of the gauge transformations
        \begin{equation}
            \widetilde{\Gamma}^{ab}_c({\bf h}) = \left(\widetilde{\Gamma}^{ab}_c({\bf h \overline{g}}\right)^{\sigma({\bf g})}.
        \end{equation}
        Hence, using equivariance to remove tildes from Eq.~\eqref{eqn:tildedTransform} leads directly to Eq.~\eqref{eqn:Utransform}.
        
        Yet another way to see how $U$ transforms using the tilded language is using the maps $\alpha_{\bf g}$. The ``untilded" map $\rho_{\bf g}$ is the composition
        \begin{equation}
            \rho_{\bf g}=\alpha_{\bf g} \circ \tilde{\rho}_{\bf g}\bigg|_V,
        \end{equation}
        where the restriction means that $\rho_{\bf g}$ is only defined on the ``untilded" vector spaces $V^{ab}_c = \tilde{V}^{ab}_c({\bf 1})$.
        We can now directly compute the transformation rules for $U$:
        \begin{align}
            \rho_{\bf g}\left(\widecheck{\ket{a,b;c;{\bf 1}}}\right) &= \alpha_{\bf g} \circ \widetilde{\rho}_{\bf g}\left(\widetilde{\Gamma}^{a,b}_c({\bf 1})\ket{a,b;c;{\bf 1}}\right)\\
            &= \left(\widetilde{\Gamma}^{a,b}_c({\bf 1})\right)^{\sigma({\bf g})}\alpha_{\bf g}\left(\widetilde{U}(\,^{\bf g}a,\,^{\bf g}b;\,^{\bf g}c;{\bf g},{\bf 1})\ket{\,^{\bf g}a,\,^{\bf g}b;\,^{\bf g}c;{\bf g}}\right)\\
            &= \left(\widetilde{\Gamma}^{a,b}_c({\bf 1})\right)^{\sigma({\bf g})}\widetilde{U}(\,^{\bf g}a,\,^{\bf g}b;\,^{\bf g}c;{\bf g},{\bf 1})^{\sigma(\bf g)}\ket{\,^{\bf g}a,\,^{\bf g}b;\,^{\bf g}c;{\bf 1}}\\
            &= \left(\widetilde{\Gamma}^{a,b}_c({\bf 1})\right)^{\sigma({\bf g})}\widetilde{U}(\,^{\bf g}a,\,^{\bf g}b;\,^{\bf g}c;{\bf g},{\bf 1})^{\sigma(\bf g)}\left(\widetilde{\Gamma}^{\,^{\bf g}a,\,^{\bf g}b}_{\,^{\bf g}c}({\bf 1})\right)^{-1}\widecheck{\ket{\,^{\bf g}a,\,^{\bf g}b;\,^{\bf g}c;{\bf 1}}}
        \end{align}
        Using the relation
        \begin{equation}
            U_{\bf g}(a,b;c) = \widetilde{U}_{\bf g}^{\sigma({\bf g})}(a,b;c;{\bf g},{\bf 1}) ,
        \end{equation}
        we indeed obtain Eq.~\eqref{eqn:Utransform}.
        
    \section{Proof of when \texorpdfstring{$\Upsilon_\psi$}{Ypsi} respects locality}
        \label{app:AbelianParityVortex}
        
        In Sec.~\ref{sec:LRNatIso}, we showed that $\Upsilon_\psi$ respects locality if and only if there exists some $\zeta_a$ which is a phase for all $a \in \mathcal{C}$, respects the fusion rules, and has $\zeta_\psi = -1$. We will prove in this appendix that such a $\zeta_a$ exists if and only if some minimal modular extension $\C$ of $\mathcal{C}$ contains an Abelian fermion parity vortex. If some $\C$ contains an Abelian fermion parity vortex, then we further claim that no minimal modular extension of $\mathcal{C}$ contains both $v$-type and $\sigma$-type vortices. Our main tools are explicit results about boson condensation proven in Ref.~\cite{delmastro2021}, Appendix A; we state their results here without proof. We prove the (simpler) second statement first to introduce some techniques used in the proof of the first. 
        
        \begin{proposition}
        If some minimal modular extension $\C$ of $\mathcal{C}$ contains an Abelian fermion parity vortex, then no minimal modular extension of $\mathcal{C}$ contains both $v$-type and $\sigma$-type vortices.
        \end{proposition}
        
        \begin{proof}
        First suppose that $\C$ contains an Abelian parity vortex $v$. Suppose by way of contradiction that $\C$ also contains a $\sigma$-type parity vortex $\sigma$. Then $ \sigma \times v = a$ for a unique anyon $a \in \mathcal{C}$. Fuse $\psi$ into both sides of the above equation; by associativity of the fusion rules we obtain
        \begin{equation}
            \psi \times (\sigma \times v) = \psi \times a = (\psi \times \sigma) \times v = \sigma \times v = a
        \end{equation}
        Hence $a = \psi \times a$ for some $a \in \mathcal{C}$, which is impossible. Therefore, $\C$ contains only $v$-type parity vortices. Next, suppose $\C$ has chiral central charge $c_-$, and consider the minimal modular extension $\C'$ of $\mathcal{C}$ with chiral central charge $c_-+1/2$. This is obtained by condensing the bosonic $(\psi, \psi)$ pair in $\C \boxtimes \mathrm{Ising}$.
        Using the notation $(x,y) \in \C \boxtimes \mathrm{Ising}$, it is clear that deconfined parity vortices in $\C'$ descend from bound states $(x,\sigma)$ where $x \in \C_1$.
        Applying the results of~\cite{delmastro2021}, if $x$ is a $v$-type parity vortex, then every $(x,\sigma)$ is a deconfined simple parity vortex in $\C'$, while if $x$ is a $\sigma$-type parity vortex, then $(x,\sigma)$ splits into a pair of simple $v$-type parity vortices $(x,\sigma)_\pm$. Since $\C$ contains only $v$-type vortices, $\C'$ contains only $\sigma$-type vortices.
        We can repeat this process to obtain all minimal modular extensions, alternating between minimal modular extensions containing only $\sigma$-type and only $v$-type vortices, as desired.
        \end{proof}
        
        \begin{theorem} There exists a set of phases $\zeta_a$ which obey the fusion rules and have $\zeta_\psi = -1$ if and only if some minimal modular extension contains an Abelian fermion parity vortex.
        \end{theorem}
        
        \begin{proof}
        The ``if" direction was already proven in Sec.~\ref{sec:LRNatIso}. Now suppose such a $\zeta_a$ exists. Then using the results of Appendix~\ref{app:characters}, there exists some minimal modular extension $\C$ of $\mathcal{C}$ containing some parity vortex $x$ such that $M_{x,a}\in \U$ for all $a \in \mathcal{C}$. Then
        \begin{align}
            1 = \sum_{y \in \C} |S_{x,y}|^2 &= \sum_{y \in \C_0} |M_{x,y} S_{1x}S_{1y}/S_{11}|^2 + \sum_{y \in \C_1} |S_{x,y}|^2\\
            &= \sum_{y \in \C_0} \frac{d_x^2 d_y^2}{\mathcal{D}_{\C}^2}+ \sum_{y \in \C_1} |S_{x,y}|^2\\
            &= d_x^2 \frac{\mathcal{D}_{\mathcal{C}}^2}{\mathcal{D}_{\C}^2} + \sum_{y \in \C_1} |S_{x,y}|^2\\
            &=\frac{d_x^2}{2} + \sum_{y \in \C_1} |S_{x,y}|^2 \geq \frac{d_x^2}{2}
        \end{align}
        Hence 
        \begin{equation}
            d_x \leq \sqrt{2}.
            \label{eqn:dxBound}
        \end{equation} Now change minimal modular extensions to $\C'$ by stacking with a copy of the Ising theory and condensing the bound state of the preferred fermions. If $x$ is $v$-type, then $(x,\sigma) \in \C'_{1,\sigma}$. Using the results of Ref.~\cite{delmastro2021}, we can obtain the $S$-matrix of the condensed theory; in the present case, we have
        \begin{equation}
            d_{(x,\sigma)}=S_{(x,\sigma),(1,1)}S_{(1,1),(1,1)} = d_x d_{\sigma} = \sqrt{2} d_x
        \end{equation}
        But using the same formula, we can also compute that for $a \in \C_0$,
        \begin{equation}
            M_{(x,\sigma),(a,1)} = M_{x,a} \in \U
        \end{equation}
        Hence Eq.~\ref{eqn:dxBound} applies to $(x,\sigma)$ as well, so $d_{(x,\sigma)}= \sqrt{2}d_x \leq \sqrt{2}$. Thus $d_x=1$, i.e. $x$ is Abelian, as desired.
        
        Suppose instead that $x$ is $\sigma$-type. Then $(x,\sigma)$ splits into $(x,\sigma)_\pm \in \C'_{1,v}$ and Ref.~\cite{delmastro2021} tells us instead
        \begin{equation}
            d_{(x,\sigma)_\pm}=S_{(x,\sigma)_\pm,(1,1)}S_{(1,1),(1,1)} = \frac{1}{2}d_x d_{\sigma} = \frac{d_x}{\sqrt{2}}
        \end{equation}
        But since $d_x \leq \sqrt{2}$ and $d_{(x,\sigma)_\pm} \geq 1$ we must have $d_x = \sqrt{2}$ and thus $(x,\sigma)_\pm$ is Abelian, as desired.
        
        Note that we can use this argument to show that if some $\C$ contains an Abelian parity vortex $v$, then all minimal modular extensions containing $v$-type parity vortices do as well. According to the above, $\C'$ contains only $\sigma$-type vortices, and in particular contains a particle $(v,\sigma)$ with quantum dimension $\sqrt{2}$. Layering Ising with $\C'$ and condensing, we obtain another minimal modular extension $\C''$ with $v$-type vortices. In this case, according to the above argument, $((v,\sigma),\sigma)_\pm$ is a $v$-type vortex in $\C''$ with quantum dimension 1, i.e., it is Abelian. This process can then be repeated to generate Abelian parity vortices in all 8 minimal modular extensions of $\C$ which have $v$-type parity vortices.
        \end{proof}
        
        \section{Proof of Eq.~\ref{eqn:multiplicationOfKramers}}
        \label{app:multiplicationOfKramersProof}
        
        Eq.~\ref{eqn:multiplicationOfKramers} reads
        \begin{equation}
            \frac{\eta_a({\bf T,T})\eta_b({\bf T,T})U_{\bf T}(a,\psi;a \times \psi)U_{\bf T}(b,\psi;b \times \psi)F^{a,\psi,\psi}F^{b,\psi,\psi}}{\eta_c({\bf T,T})}=-1
            \label{eqn:KramersMultStep1}
        \end{equation}
        for $\,^{\bf T}a = a \times \psi$, $\,^{\bf T}b = b \times \psi$, $\,^{\bf T}c = c$, and $N_{ab}^c = 1$. The factors of $\eta$ on the left-hand side can be replaced using the $\eta-U$ consistency condition Eq.~\ref{eqn:etaUConsistency} 
        \begin{equation}
            \frac{\eta_a({\bf T,T})\eta_b({\bf T,T})}{\eta_c({\bf T,T})} = U_{\bf T}(\,^{\bf T}a,\,^{\bf T}b;\,^{\bf T}c)U_{\bf T}^{\ast}(a,b;c) \label{eqn:etaUForTR}
        \end{equation}
        The consistency condition Eq.~\ref{eqn:UFconsistency} between $U$ and $F$ implies
        \begin{equation}
            U_{\bf T}(\,^{\bf T}a,\,^{\bf T}b;\,^{\bf T}c)U_{\bf T}^{\ast}(a,b;c) = F^{a,\psi,\,^{\bf T}b}_{c,\,^{\bf T}a,b}F^{\,^{\bf T}a,\psi,b}_{c,a,\,^{\bf T}b}U_{\bf T}(a \times \psi, \psi;a)U_{\bf T}^{\ast}(\psi,b;b \times \psi)
        \end{equation}
        Using the pentagon equation
        for the anyons $a,\psi,\psi,b$, we obtain
        \begin{equation}
            F^{a,\psi,\,^{\bf T}b}_{c,\,^{\bf T}a,b}F^{\,^{\bf T}a,\psi,b}_{c,a,\,^{\bf T}b} = F^{a,\psi, \psi}F^{\psi,\psi,b}
            \label{eqn:pentagonAPsiPsiB}
        \end{equation}
        Inserting Eqs.~\ref{eqn:etaUForTR}-\ref{eqn:pentagonAPsiPsiB} into Eq.~\ref{eqn:KramersMultStep1}, we find
        \begin{equation}
            \frac{\eta_a^{\bf T} \eta_b^{\bf T}}{\eta_c^{\bf T}} = \left(F^{a\psi \psi}\right)^2 F^{b \psi \psi}F^{\psi \psi b} U_{\bf T}(a \times \psi,\psi;a)U_{\bf T}(a,\psi;a \times \psi) U_{\bf T}(b,\psi;b \times \psi)U_{\bf T}^\ast(\psi,b;\psi \times b)
            \label{eqn:etaTRatioIntermediate}
        \end{equation}
        From the pentagon equation 
        for the anyons $a,\psi,\psi,\psi$, one finds
        \begin{equation}
            F^{a,\psi,\psi} = F^{a\times \psi,\psi,\psi}
        \end{equation}
        Hence, again using the $U-F$ consistency Eq.~\ref{eqn:etaUConsistency},
        \begin{equation}
            \left(F^{a,\psi,\psi}\right)^2 = F^{a,\psi,\psi}F^{\,^{\bf T}a,\psi,\psi} = U^{\ast}_{\bf T}(a,\psi;a \times \psi)U_{\bf T}^{\ast}(a \times \psi, \psi, a)
            \label{eqn:modifiedUF}
        \end{equation}
        Next, using the $U$-$R$ consistency Eq.~\ref{eqn:URconsistency} we find
        \begin{equation}
            U_{\bf T}^{\ast}(\psi,b;\psi \times b)U_{\bf T}(b,\psi;b\times \psi) = R^{b,\psi}R^{b\times \psi,\psi} \label{eqn:URforTR}
        \end{equation}
        Inserting Eq.~\ref{eqn:modifiedUF} into Eq.~\ref{eqn:etaTRatioIntermediate} shows that all of the factors involving $a$ cancel. Further inserting~\ref{eqn:URforTR} and applying the hexagon equation 
        twice,
        \begin{align}
            \frac{\eta_a^{\bf T} \eta_b^{\bf T}}{\eta_c^{\bf T}} = &= F^{\psi, \psi b}F^{b,\psi, \psi}R^{b,\psi}R^{b\times \psi,\psi}\\
            &=F^{\psi, b, \psi} \left(R^{b,\psi}\right)^\ast R^{b \times \psi, \psi}\\
            &=R^{\psi \psi}=-1
        \end{align}
        as claimed.
        
        \bibliography{TI}

\begin{thebibliography}{35}%
\makeatletter
\providecommand \@ifxundefined [1]{%
 \@ifx{#1\undefined}
}%
\providecommand \@ifnum [1]{%
 \ifnum #1\expandafter \@firstoftwo
 \else \expandafter \@secondoftwo
 \fi
}%
\providecommand \@ifx [1]{%
 \ifx #1\expandafter \@firstoftwo
 \else \expandafter \@secondoftwo
 \fi
}%
\providecommand \natexlab [1]{#1}%
\providecommand \enquote  [1]{``#1''}%
\providecommand \bibnamefont  [1]{#1}%
\providecommand \bibfnamefont [1]{#1}%
\providecommand \citenamefont [1]{#1}%
\providecommand \href@noop [0]{\@secondoftwo}%
\providecommand \href [0]{\begingroup \@sanitize@url \@href}%
\providecommand \@href[1]{\@@startlink{#1}\@@href}%
\providecommand \@@href[1]{\endgroup#1\@@endlink}%
\providecommand \@sanitize@url [0]{\catcode `\\12\catcode `\$12\catcode
  `\&12\catcode `\#12\catcode `\^12\catcode `\_12\catcode `\%12\relax}%
\providecommand \@@startlink[1]{}%
\providecommand \@@endlink[0]{}%
\providecommand \url  [0]{\begingroup\@sanitize@url \@url }%
\providecommand \@url [1]{\endgroup\@href {#1}{\urlprefix }}%
\providecommand \urlprefix  [0]{URL }%
\providecommand \Eprint [0]{\href }%
\providecommand \doibase [0]{https://doi.org/}%
\providecommand \selectlanguage [0]{\@gobble}%
\providecommand \bibinfo  [0]{\@secondoftwo}%
\providecommand \bibfield  [0]{\@secondoftwo}%
\providecommand \translation [1]{[#1]}%
\providecommand \BibitemOpen [0]{}%
\providecommand \bibitemStop [0]{}%
\providecommand \bibitemNoStop [0]{.\EOS\space}%
\providecommand \EOS [0]{\spacefactor3000\relax}%
\providecommand \BibitemShut  [1]{\csname bibitem#1\endcsname}%
\let\auto@bib@innerbib\@empty
\bibitem [{\citenamefont {Wen}(2004)}]{wen04}%
  \BibitemOpen
  \bibfield  {author} {\bibinfo {author} {\bibfnamefont {X.-G.}\ \bibnamefont
  {Wen}},\ }\href@noop {} {\emph {\bibinfo {title} {Quantum Field Theory of
  Many-Body Systems}}}\ (\bibinfo  {publisher} {Oxford Univ. Press},\ \bibinfo
  {address} {Oxford},\ \bibinfo {year} {2004})\BibitemShut {NoStop}%
\bibitem [{\citenamefont {Barkeshli}\ \emph
  {et~al.}(2019{\natexlab{a}})\citenamefont {Barkeshli}, \citenamefont
  {Bonderson}, \citenamefont {Cheng},\ and\ \citenamefont
  {Wang}}]{barkeshli2019}%
  \BibitemOpen
  \bibfield  {author} {\bibinfo {author} {\bibfnamefont {M.}~\bibnamefont
  {Barkeshli}}, \bibinfo {author} {\bibfnamefont {P.}~\bibnamefont
  {Bonderson}}, \bibinfo {author} {\bibfnamefont {M.}~\bibnamefont {Cheng}},\
  and\ \bibinfo {author} {\bibfnamefont {Z.}~\bibnamefont {Wang}},\ }\bibfield
  {title} {\bibinfo {title} {Symmetry fractionalization, defects, and gauging
  of topological phases},\ }\href {https://doi.org/10.1103/PhysRevB.100.115147}
  {\bibfield  {journal} {\bibinfo  {journal} {Phys. Rev. B}\ }\textbf {\bibinfo
  {volume} {100}},\ \bibinfo {pages} {115147} (\bibinfo {year}
  {2019}{\natexlab{a}})},\ \Eprint {https://arxiv.org/abs/arXiv:1410.4540}
  {arXiv:1410.4540} \BibitemShut {NoStop}%
\bibitem [{\citenamefont {Essin}\ and\ \citenamefont
  {Hermele}(2013)}]{essin2013}%
  \BibitemOpen
  \bibfield  {author} {\bibinfo {author} {\bibfnamefont {A.~M.}\ \bibnamefont
  {Essin}}\ and\ \bibinfo {author} {\bibfnamefont {M.}~\bibnamefont
  {Hermele}},\ }\bibfield  {title} {\bibinfo {title} {Classifying
  fractionalization: Symmetry classification of gapped ${\mathbb{z}}_{2}$ spin
  liquids in two dimensions},\ }\href
  {https://doi.org/10.1103/PhysRevB.87.104406} {\bibfield  {journal} {\bibinfo
  {journal} {Phys. Rev. B}\ }\textbf {\bibinfo {volume} {87}},\ \bibinfo
  {pages} {104406} (\bibinfo {year} {2013})}\BibitemShut {NoStop}%
\bibitem [{\citenamefont {Essin}\ and\ \citenamefont
  {Hermele}(2014)}]{essin2014}%
  \BibitemOpen
  \bibfield  {author} {\bibinfo {author} {\bibfnamefont {A.~M.}\ \bibnamefont
  {Essin}}\ and\ \bibinfo {author} {\bibfnamefont {M.}~\bibnamefont
  {Hermele}},\ }\bibfield  {title} {\bibinfo {title} {Spectroscopic signatures
  of crystal momentum fractionalization},\ }\bibfield  {journal} {\bibinfo
  {journal} {Physical Review B}\ }\textbf {\bibinfo {volume} {90}},\ \href
  {https://doi.org/10.1103/physrevb.90.121102} {10.1103/physrevb.90.121102}
  (\bibinfo {year} {2014})\BibitemShut {NoStop}%
\bibitem [{\citenamefont {Qi}\ and\ \citenamefont
  {Cheng}(2018)}]{qi2018spinliquid}%
  \BibitemOpen
  \bibfield  {author} {\bibinfo {author} {\bibfnamefont {Y.}~\bibnamefont
  {Qi}}\ and\ \bibinfo {author} {\bibfnamefont {M.}~\bibnamefont {Cheng}},\
  }\bibfield  {title} {\bibinfo {title} {Classification of symmetry
  fractionalization in gapped z2 spin liquids},\ }\bibfield  {journal}
  {\bibinfo  {journal} {Physical Review B}\ }\textbf {\bibinfo {volume} {97}},\
  \href {https://doi.org/10.1103/physrevb.97.115138}
  {10.1103/physrevb.97.115138} (\bibinfo {year} {2018})\BibitemShut {NoStop}%
\bibitem [{\citenamefont {Manjunath}\ and\ \citenamefont
  {Barkeshli}(2021)}]{manjunath2021}%
  \BibitemOpen
  \bibfield  {author} {\bibinfo {author} {\bibfnamefont {N.}~\bibnamefont
  {Manjunath}}\ and\ \bibinfo {author} {\bibfnamefont {M.}~\bibnamefont
  {Barkeshli}},\ }\bibfield  {title} {\bibinfo {title} {Crystalline gauge
  fields and quantized discrete geometric response for abelian topological
  phases with lattice symmetry},\ }\bibfield  {journal} {\bibinfo  {journal}
  {Physical Review Research}\ }\textbf {\bibinfo {volume} {3}},\ \href
  {https://doi.org/10.1103/physrevresearch.3.013040}
  {10.1103/physrevresearch.3.013040} (\bibinfo {year} {2021})\BibitemShut
  {NoStop}%
\bibitem [{\citenamefont {Manjunath}\ and\ \citenamefont
  {Barkeshli}(2020)}]{manjunath2020}%
  \BibitemOpen
  \bibfield  {author} {\bibinfo {author} {\bibfnamefont {N.}~\bibnamefont
  {Manjunath}}\ and\ \bibinfo {author} {\bibfnamefont {M.}~\bibnamefont
  {Barkeshli}},\ }\href@noop {} {\bibinfo {title} {Classification of fractional
  quantum hall states with spatial symmetries}} (\bibinfo {year} {2020}),\
  \Eprint {https://arxiv.org/abs/2012.11603} {arXiv:2012.11603
  [cond-mat.str-el]} \BibitemShut {NoStop}%
\bibitem [{\citenamefont {Cheng}\ \emph {et~al.}(2016)\citenamefont {Cheng},
  \citenamefont {Zaletel}, \citenamefont {Barkeshli}, \citenamefont
  {Vishwanath},\ and\ \citenamefont {Bonderson}}]{cheng2016lsm}%
  \BibitemOpen
  \bibfield  {author} {\bibinfo {author} {\bibfnamefont {M.}~\bibnamefont
  {Cheng}}, \bibinfo {author} {\bibfnamefont {M.}~\bibnamefont {Zaletel}},
  \bibinfo {author} {\bibfnamefont {M.}~\bibnamefont {Barkeshli}}, \bibinfo
  {author} {\bibfnamefont {A.}~\bibnamefont {Vishwanath}},\ and\ \bibinfo
  {author} {\bibfnamefont {P.}~\bibnamefont {Bonderson}},\ }\bibfield  {title}
  {\bibinfo {title} {Translational symmetry and microscopic constraints on
  symmetry-enriched topological phases: A view from the surface},\ }\href
  {https://doi.org/10.1103/PhysRevX.6.041068} {\bibfield  {journal} {\bibinfo
  {journal} {Phys. Rev. X}\ }\textbf {\bibinfo {volume} {6}},\ \bibinfo {pages}
  {041068} (\bibinfo {year} {2016})}\BibitemShut {NoStop}%
\bibitem [{\citenamefont {Barkeshli}\ and\ \citenamefont
  {Cheng}(2019)}]{barkeshli2019rel}%
  \BibitemOpen
  \bibfield  {author} {\bibinfo {author} {\bibfnamefont {M.}~\bibnamefont
  {Barkeshli}}\ and\ \bibinfo {author} {\bibfnamefont {M.}~\bibnamefont
  {Cheng}},\ }\href@noop {} {\bibinfo {title} {Relative anomalies in (2+1)d
  symmetry enriched topological states}} (\bibinfo {year} {2019}),\ \Eprint
  {https://arxiv.org/abs/arXiv:1906.10691} {arXiv:1906.10691} \BibitemShut
  {NoStop}%
\bibitem [{\citenamefont {Barkeshli}\ \emph
  {et~al.}(2019{\natexlab{b}})\citenamefont {Barkeshli}, \citenamefont
  {Bonderson}, \citenamefont {Cheng}, \citenamefont {Jian},\ and\ \citenamefont
  {Walker}}]{barkeshli2019tr}%
  \BibitemOpen
  \bibfield  {author} {\bibinfo {author} {\bibfnamefont {M.}~\bibnamefont
  {Barkeshli}}, \bibinfo {author} {\bibfnamefont {P.}~\bibnamefont
  {Bonderson}}, \bibinfo {author} {\bibfnamefont {M.}~\bibnamefont {Cheng}},
  \bibinfo {author} {\bibfnamefont {C.-M.}\ \bibnamefont {Jian}},\ and\
  \bibinfo {author} {\bibfnamefont {K.}~\bibnamefont {Walker}},\ }\bibfield
  {title} {\bibinfo {title} {Reflection and time reversal symmetry enriched
  topological phases of matter: Path integrals, non-orientable manifolds, and
  anomalies},\ }\bibfield  {journal} {\bibinfo  {journal} {Communications in
  Mathematical Physics}\ }\href {https://doi.org/10.1007/s00220-019-03475-8}
  {10.1007/s00220-019-03475-8} (\bibinfo {year} {2019}{\natexlab{b}}),\ \Eprint
  {https://arxiv.org/abs/arXiv:1612.07792} {arXiv:1612.07792} \BibitemShut
  {NoStop}%
\bibitem [{\citenamefont {Bulmash}\ and\ \citenamefont
  {Barkeshli}(2020)}]{bulmash2020}%
  \BibitemOpen
  \bibfield  {author} {\bibinfo {author} {\bibfnamefont {D.}~\bibnamefont
  {Bulmash}}\ and\ \bibinfo {author} {\bibfnamefont {M.}~\bibnamefont
  {Barkeshli}},\ }\bibfield  {title} {\bibinfo {title} {Absolute anomalies in
  (2+1)d symmetry-enriched topological states and exact (3+1)d constructions},\
  }\href {https://doi.org/10.1103/PhysRevResearch.2.043033} {\bibfield
  {journal} {\bibinfo  {journal} {Phys. Rev. Research}\ }\textbf {\bibinfo
  {volume} {2}},\ \bibinfo {pages} {043033} (\bibinfo {year}
  {2020})}\BibitemShut {NoStop}%
\bibitem [{\citenamefont {Tata}\ \emph {et~al.}(2021)\citenamefont {Tata},
  \citenamefont {Kobayashi}, \citenamefont {Bulmash},\ and\ \citenamefont
  {Barkeshli}}]{TataSETAnomaly}%
  \BibitemOpen
  \bibfield  {author} {\bibinfo {author} {\bibfnamefont {S.}~\bibnamefont
  {Tata}}, \bibinfo {author} {\bibfnamefont {R.}~\bibnamefont {Kobayashi}},
  \bibinfo {author} {\bibfnamefont {D.}~\bibnamefont {Bulmash}},\ and\ \bibinfo
  {author} {\bibfnamefont {M.}~\bibnamefont {Barkeshli}},\ }\href@noop {}
  {\bibinfo {title} {{Anomalies in (2+1)D fermionic topological phases and
  (3+1)D path integral state sums for fermionic SPTs}}} (\bibinfo {year}
  {2021}),\ \Eprint {https://arxiv.org/abs/2104.14567} {arXiv:2104.14567}
  \BibitemShut {NoStop}%
\bibitem [{\citenamefont {Bruillard}\ \emph
  {et~al.}(2017{\natexlab{a}})\citenamefont {Bruillard}, \citenamefont
  {Galindo}, \citenamefont {Hagge}, \citenamefont {Ng}, \citenamefont
  {Plavnik}, \citenamefont {Rowell},\ and\ \citenamefont
  {Wang}}]{bruillard2017a}%
  \BibitemOpen
  \bibfield  {author} {\bibinfo {author} {\bibfnamefont {P.}~\bibnamefont
  {Bruillard}}, \bibinfo {author} {\bibfnamefont {C.}~\bibnamefont {Galindo}},
  \bibinfo {author} {\bibfnamefont {T.}~\bibnamefont {Hagge}}, \bibinfo
  {author} {\bibfnamefont {S.-H.}\ \bibnamefont {Ng}}, \bibinfo {author}
  {\bibfnamefont {J.~Y.}\ \bibnamefont {Plavnik}}, \bibinfo {author}
  {\bibfnamefont {E.~C.}\ \bibnamefont {Rowell}},\ and\ \bibinfo {author}
  {\bibfnamefont {Z.}~\bibnamefont {Wang}},\ }\bibfield  {title} {\bibinfo
  {title} {Fermionic modular categories and the 16-fold way},\ }\href
  {https://doi.org/10.1063/1.4982048} {\bibfield  {journal} {\bibinfo
  {journal} {Journal of Mathematical Physics}\ }\textbf {\bibinfo {volume}
  {58}},\ \bibinfo {pages} {041704} (\bibinfo {year}
  {2017}{\natexlab{a}})}\BibitemShut {NoStop}%
\bibitem [{\citenamefont {Bruillard}\ \emph
  {et~al.}(2017{\natexlab{b}})\citenamefont {Bruillard}, \citenamefont
  {Galindo}, \citenamefont {Ng}, \citenamefont {Plavnik}, \citenamefont
  {Rowell},\ and\ \citenamefont {Wang}}]{bruillard2017b}%
  \BibitemOpen
  \bibfield  {author} {\bibinfo {author} {\bibfnamefont {P.}~\bibnamefont
  {Bruillard}}, \bibinfo {author} {\bibfnamefont {C.}~\bibnamefont {Galindo}},
  \bibinfo {author} {\bibfnamefont {S.-H.}\ \bibnamefont {Ng}}, \bibinfo
  {author} {\bibfnamefont {J.~Y.}\ \bibnamefont {Plavnik}}, \bibinfo {author}
  {\bibfnamefont {E.~C.}\ \bibnamefont {Rowell}},\ and\ \bibinfo {author}
  {\bibfnamefont {Z.}~\bibnamefont {Wang}},\ }\href@noop {} {\bibinfo {title}
  {Classification of super-modular categories by rank}} (\bibinfo {year}
  {2017}{\natexlab{b}}),\ \Eprint {https://arxiv.org/abs/1705.05293}
  {arXiv:1705.05293 [math.QA]} \BibitemShut {NoStop}%
\bibitem [{\citenamefont {Bonderson}\ \emph {et~al.}(2018)\citenamefont
  {Bonderson}, \citenamefont {Rowell}, \citenamefont {Wang},\ and\
  \citenamefont {Zhang}}]{bonderson2018}%
  \BibitemOpen
  \bibfield  {author} {\bibinfo {author} {\bibfnamefont {P.}~\bibnamefont
  {Bonderson}}, \bibinfo {author} {\bibfnamefont {E.}~\bibnamefont {Rowell}},
  \bibinfo {author} {\bibfnamefont {Z.}~\bibnamefont {Wang}},\ and\ \bibinfo
  {author} {\bibfnamefont {Q.}~\bibnamefont {Zhang}},\ }\bibfield  {title}
  {\bibinfo {title} {Congruence subgroups and super-modular categories},\
  }\href {https://doi.org/10.2140/pjm.2018.296.257} {\bibfield  {journal}
  {\bibinfo  {journal} {Pacific Journal of Mathematics}\ }\textbf {\bibinfo
  {volume} {296}},\ \bibinfo {pages} {257–270} (\bibinfo {year}
  {2018})}\BibitemShut {NoStop}%
\bibitem [{\citenamefont {Davydov}(2014)}]{davydov2014}%
  \BibitemOpen
  \bibfield  {author} {\bibinfo {author} {\bibfnamefont {A.}~\bibnamefont
  {Davydov}},\ }\bibfield  {title} {\bibinfo {title} {Bogomolov multiplier,
  double class-preserving automorphisms, and modular invariants for
  orbifolds},\ }\href {https://doi.org/10.1063/1.4895764} {\bibfield  {journal}
  {\bibinfo  {journal} {J. Math. Phys.}\ }\textbf {\bibinfo {volume} {55}},\
  \bibinfo {pages} {092305} (\bibinfo {year} {2014})}\BibitemShut {NoStop}%
\bibitem [{bon()}]{bondersonNote}%
  \BibitemOpen
  \bibinfo {note} {We thank Parsa Bonderson for mentioning these results, which
  we independently reproduced.}\BibitemShut {Stop}%
\bibitem [{\citenamefont {Barkeshli}\ and\ \citenamefont
  {Cheng}(2018{\natexlab{a}})}]{barkeshli2018}%
  \BibitemOpen
  \bibfield  {author} {\bibinfo {author} {\bibfnamefont {M.}~\bibnamefont
  {Barkeshli}}\ and\ \bibinfo {author} {\bibfnamefont {M.}~\bibnamefont
  {Cheng}},\ }\bibfield  {title} {\bibinfo {title} {Time-reversal and
  spatial-reflection symmetry localization anomalies in (2+1)-dimensional
  topological phases of matter},\ }\href
  {https://doi.org/10.1103/PhysRevB.98.115129} {\bibfield  {journal} {\bibinfo
  {journal} {Phys. Rev. B}\ }\textbf {\bibinfo {volume} {98}},\ \bibinfo
  {pages} {115129} (\bibinfo {year} {2018}{\natexlab{a}})}\BibitemShut
  {NoStop}%
\bibitem [{\citenamefont {Fidkowski}\ and\ \citenamefont
  {Vishwanath}(2015)}]{fidkowski2015}%
  \BibitemOpen
  \bibfield  {author} {\bibinfo {author} {\bibfnamefont {L.}~\bibnamefont
  {Fidkowski}}\ and\ \bibinfo {author} {\bibfnamefont {A.}~\bibnamefont
  {Vishwanath}},\ }\href@noop {} {\bibinfo {title} {Realizing anomalous anyonic
  symmetries at the surfaces of 3d gauge theories}} (\bibinfo {year} {2015}),\
  \Eprint {https://arxiv.org/abs/arXiv:1511.01502} {arXiv:1511.01502}
  \BibitemShut {NoStop}%
\bibitem [{\citenamefont {{Fidkowski}}\ \emph {et~al.}(2013)\citenamefont
  {{Fidkowski}}, \citenamefont {{Chen}},\ and\ \citenamefont
  {{Vishwanath}}}]{Fidkowski13}%
  \BibitemOpen
  \bibfield  {author} {\bibinfo {author} {\bibfnamefont {L.}~\bibnamefont
  {{Fidkowski}}}, \bibinfo {author} {\bibfnamefont {X.}~\bibnamefont
  {{Chen}}},\ and\ \bibinfo {author} {\bibfnamefont {A.}~\bibnamefont
  {{Vishwanath}}},\ }\bibfield  {title} {\bibinfo {title} {{Non-Abelian
  Topological Order on the Surface of a 3D Topological Superconductor from an
  Exactly Solved Model}},\ }\href {https://doi.org/10.1103/PhysRevX.3.041016}
  {\bibfield  {journal} {\bibinfo  {journal} {Phys. Rev. X}\ }\textbf {\bibinfo
  {volume} {3}},\ \bibinfo {eid} {041016} (\bibinfo {year} {2013})},\ \Eprint
  {https://arxiv.org/abs/arXiv:1305.5851} {arXiv:1305.5851 [cond-mat.str-el]}
  \BibitemShut {NoStop}%
\bibitem [{\citenamefont {Metlitski}\ \emph {et~al.}(2014)\citenamefont
  {Metlitski}, \citenamefont {Fidkowski}, \citenamefont {Chen},\ and\
  \citenamefont {Vishwanath}}]{metlitski2014}%
  \BibitemOpen
  \bibfield  {author} {\bibinfo {author} {\bibfnamefont {M.}~\bibnamefont
  {Metlitski}}, \bibinfo {author} {\bibfnamefont {L.}~\bibnamefont
  {Fidkowski}}, \bibinfo {author} {\bibfnamefont {X.}~\bibnamefont {Chen}},\
  and\ \bibinfo {author} {\bibfnamefont {A.}~\bibnamefont {Vishwanath}},\
  }\href@noop {} {\bibinfo {title} {{Interaction effects on 3D topological
  superconductors: surface topological order from vortex condensation, the 16
  fold way and fermionic Kramers doublets}}} (\bibinfo {year} {2014}),\ \Eprint
  {https://arxiv.org/abs/1406.3032} {arXiv:1406.3032} \BibitemShut {NoStop}%
\bibitem [{\citenamefont {Fidkowski}\ \emph {et~al.}(2018)\citenamefont
  {Fidkowski}, \citenamefont {Vishwanath},\ and\ \citenamefont
  {Metlitski}}]{fidkowski2018}%
  \BibitemOpen
  \bibfield  {author} {\bibinfo {author} {\bibfnamefont {L.}~\bibnamefont
  {Fidkowski}}, \bibinfo {author} {\bibfnamefont {A.}~\bibnamefont
  {Vishwanath}},\ and\ \bibinfo {author} {\bibfnamefont {M.~A.}\ \bibnamefont
  {Metlitski}},\ }\href@noop {} {\bibinfo {title} {Surface topological order
  and a new 't hooft anomaly of interaction enabled 3+1d fermion spts}}
  (\bibinfo {year} {2018}),\ \Eprint {https://arxiv.org/abs/1804.08628}
  {arXiv:1804.08628 [cond-mat.str-el]} \BibitemShut {NoStop}%
\bibitem [{\citenamefont {Galindo}\ and\ \citenamefont
  {Venegas-Ram\'irez}(2017)}]{galindo}%
  \BibitemOpen
  \bibfield  {author} {\bibinfo {author} {\bibfnamefont {C.}~\bibnamefont
  {Galindo}}\ and\ \bibinfo {author} {\bibfnamefont {C.}~\bibnamefont
  {Venegas-Ram\'irez}},\ }\href@noop {} {\bibinfo {title} {Categorical
  fermionic actions and minimal modular extensions}} (\bibinfo {year} {2017}),\
  \Eprint {https://arxiv.org/abs/1712.07097} {arXiv:1712.07097} \BibitemShut
  {NoStop}%
\bibitem [{\citenamefont {Bonderson}(2007)}]{Bonderson07b}%
  \BibitemOpen
  \bibfield  {author} {\bibinfo {author} {\bibfnamefont {P.~H.}\ \bibnamefont
  {Bonderson}},\ }\emph {\bibinfo {title} {Non-{A}belian {A}nyons and
  {I}nterferometry}},\ \href@noop {} {Ph.D. thesis},\ \bibinfo  {school}
  {California Institute of Technology} (\bibinfo {year} {2007})\BibitemShut
  {NoStop}%
\bibitem [{\citenamefont {Cheng}(2019)}]{MengFermionicLSM}%
  \BibitemOpen
  \bibfield  {author} {\bibinfo {author} {\bibfnamefont {M.}~\bibnamefont
  {Cheng}},\ }\bibfield  {title} {\bibinfo {title} {Fermionic
  lieb-schultz-mattis theorems and weak symmetry-protected phases},\ }\href
  {https://doi.org/10.1103/PhysRevB.99.075143} {\bibfield  {journal} {\bibinfo
  {journal} {Phys. Rev. B}\ }\textbf {\bibinfo {volume} {99}},\ \bibinfo
  {pages} {075143} (\bibinfo {year} {2019})}\BibitemShut {NoStop}%
\bibitem [{\citenamefont {Johnson-Freyd}\ and\ \citenamefont
  {Reutter}(2021)}]{johnsonFreydModularExt}%
  \BibitemOpen
  \bibfield  {author} {\bibinfo {author} {\bibfnamefont {T.}~\bibnamefont
  {Johnson-Freyd}}\ and\ \bibinfo {author} {\bibfnamefont {D.}~\bibnamefont
  {Reutter}},\ }\href@noop {} {\bibinfo {title} {Minimal nondegenerate
  extensions}} (\bibinfo {year} {2021}),\ \Eprint
  {https://arxiv.org/abs/2105.15167} {arXiv:2105.15167} \BibitemShut {NoStop}%
\bibitem [{\citenamefont {Delmastro}\ \emph {et~al.}(2021)\citenamefont
  {Delmastro}, \citenamefont {Gaiotto},\ and\ \citenamefont
  {Gomis}}]{delmastro2021}%
  \BibitemOpen
  \bibfield  {author} {\bibinfo {author} {\bibfnamefont {D.}~\bibnamefont
  {Delmastro}}, \bibinfo {author} {\bibfnamefont {D.}~\bibnamefont {Gaiotto}},\
  and\ \bibinfo {author} {\bibfnamefont {J.}~\bibnamefont {Gomis}},\
  }\href@noop {} {\bibinfo {title} {Global anomalies on the hilbert space}}
  (\bibinfo {year} {2021}),\ \Eprint {https://arxiv.org/abs/2101.02218}
  {arXiv:2101.02218 [hep-th]} \BibitemShut {NoStop}%
\bibitem [{\citenamefont {Bulmash}\ and\ \citenamefont
  {Barkeshli}(2021)}]{bulmash2021Anomalies}%
  \BibitemOpen
  \bibfield  {author} {\bibinfo {author} {\bibfnamefont {D.}~\bibnamefont
  {Bulmash}}\ and\ \bibinfo {author} {\bibfnamefont {M.}~\bibnamefont
  {Barkeshli}},\ }\href@noop {} {\bibinfo {title} {{Anomaly cascade in (2+1)D
  fermionic topological phases}}} (\bibinfo {year} {2021}),\ \Eprint
  {https://arxiv.org/abs/2109.10922} {arXiv:2109.10922} \BibitemShut {NoStop}%
\bibitem [{\citenamefont {Benini}\ \emph {et~al.}(2019)\citenamefont {Benini},
  \citenamefont {C{\'o}rdova},\ and\ \citenamefont {Hsin}}]{benini2019}%
  \BibitemOpen
  \bibfield  {author} {\bibinfo {author} {\bibfnamefont {F.}~\bibnamefont
  {Benini}}, \bibinfo {author} {\bibfnamefont {C.}~\bibnamefont
  {C{\'o}rdova}},\ and\ \bibinfo {author} {\bibfnamefont {P.-S.}\ \bibnamefont
  {Hsin}},\ }\bibfield  {title} {\bibinfo {title} {On 2-group global symmetries
  and their anomalies},\ }\href@noop {} {\bibfield  {journal} {\bibinfo
  {journal} {Journal of High Energy Physics}\ }\textbf {\bibinfo {volume}
  {2019}},\ \bibinfo {pages} {118} (\bibinfo {year} {2019})}\BibitemShut
  {NoStop}%
\bibitem [{\citenamefont {Aasen}\ \emph {et~al.}(2021)\citenamefont {Aasen},
  \citenamefont {Bonderson},\ and\ \citenamefont {Knapp}}]{aasen21ferm}%
  \BibitemOpen
  \bibfield  {author} {\bibinfo {author} {\bibfnamefont {D.}~\bibnamefont
  {Aasen}}, \bibinfo {author} {\bibfnamefont {P.}~\bibnamefont {Bonderson}},\
  and\ \bibinfo {author} {\bibfnamefont {C.}~\bibnamefont {Knapp}},\
  }\href@noop {} {\bibinfo {title} {{Characterization and Classification of
  Fermionic Symmetry Enriched Topological Phases}}} (\bibinfo {year} {2021}),\
  \Eprint {https://arxiv.org/abs/2109.10911} {arXiv:2109.10911} \BibitemShut
  {NoStop}%
\bibitem [{\citenamefont {Wang}\ and\ \citenamefont {Gu}(2020)}]{WangGu}%
  \BibitemOpen
  \bibfield  {author} {\bibinfo {author} {\bibfnamefont {Q.-R.}\ \bibnamefont
  {Wang}}\ and\ \bibinfo {author} {\bibfnamefont {Z.-C.}\ \bibnamefont {Gu}},\
  }\bibfield  {title} {\bibinfo {title} {{Construction and Classification of
  Symmetry-Protected Topological Phases in Interacting Fermion Systems}},\
  }\href {https://doi.org/10.1103/PHYSREVX.10.031055} {\bibfield  {journal}
  {\bibinfo  {journal} {Phys. Rev. X}\ }\textbf {\bibinfo {volume} {10}},\
  \bibinfo {pages} {31055} (\bibinfo {year} {2020})}\BibitemShut {NoStop}%
\bibitem [{\citenamefont {Barkeshli}\ and\ \citenamefont
  {Cheng}(2018{\natexlab{b}})}]{BarkeshliChengLocalizationAnomalies}%
  \BibitemOpen
  \bibfield  {author} {\bibinfo {author} {\bibfnamefont {M.}~\bibnamefont
  {Barkeshli}}\ and\ \bibinfo {author} {\bibfnamefont {M.}~\bibnamefont
  {Cheng}},\ }\bibfield  {title} {\bibinfo {title} {{Time-reversal and spatial
  reflection symmetry localization anomalies in (2+1)D topological phases of
  matter}},\ }\href {https://doi.org/10.1103/PhysRevB.98.115129} {\bibfield
  {journal} {\bibinfo  {journal} {Phys. Rev. B}\ }\textbf {\bibinfo {volume}
  {98}},\ \bibinfo {pages} {115129} (\bibinfo {year}
  {2018}{\natexlab{b}})}\BibitemShut {NoStop}%
\bibitem [{\citenamefont {Barkeshli}\ \emph {et~al.}(2021)\citenamefont
  {Barkeshli}, \citenamefont {Chen}, \citenamefont {Hsin},\ and\ \citenamefont
  {Manjunath}}]{manjunath21inv}%
  \BibitemOpen
  \bibfield  {author} {\bibinfo {author} {\bibfnamefont {M.}~\bibnamefont
  {Barkeshli}}, \bibinfo {author} {\bibfnamefont {Y.-A.}\ \bibnamefont {Chen}},
  \bibinfo {author} {\bibfnamefont {P.-S.}\ \bibnamefont {Hsin}},\ and\
  \bibinfo {author} {\bibfnamefont {N.}~\bibnamefont {Manjunath}},\ }\href@noop
  {} {\bibinfo {title} {Classification of (2+1)d invertible fermionic
  topological phases with symmetry}} (\bibinfo {year} {2021}),\ \Eprint
  {https://arxiv.org/abs/2109.11039} {arXiv:2109.11039} \BibitemShut {NoStop}%
\bibitem [{cha()}]{charNote}%
  \BibitemOpen
  \bibinfo {note} {These results were also derived independently in unpublished
  work \cite{bondersonChar}}\BibitemShut {NoStop}%
\bibitem [{\citenamefont {Bonderson}\ \emph {et~al.}()\citenamefont
  {Bonderson}, \citenamefont {Cheng}, \citenamefont {Mong},\ and\ \citenamefont
  {Tran}}]{bondersonChar}%
  \BibitemOpen
  \bibfield  {author} {\bibinfo {author} {\bibfnamefont {P.}~\bibnamefont
  {Bonderson}}, \bibinfo {author} {\bibfnamefont {M.}~\bibnamefont {Cheng}},
  \bibinfo {author} {\bibfnamefont {R.}~\bibnamefont {Mong}},\ and\ \bibinfo
  {author} {\bibfnamefont {A.}~\bibnamefont {Tran}},\ }\bibinfo {note}
  {unpublished}\BibitemShut {NoStop}%
\end{thebibliography}%
        
\end{document}